\RequirePackage{snapshot}
\documentclass[onecolumn]{IEEEtran}
\usepackage{amsfonts,amsmath,amsthm,amssymb,dsfont}
\usepackage[hyphens]{url}

\usepackage[usenames,dvipsnames]{color}

\usepackage[pdftex]{graphicx}  %remove draft option to print figures 
\DeclareGraphicsExtensions{.pdf} 
\usepackage[numbers,sort&compress]{natbib} 
\usepackage{hypernat}

\graphicspath{{./}{./Plots/}}

\usepackage{verbatim}
\usepackage{fancyvrb} 
\usepackage{fancyhdr}
\usepackage{pdfpages}

\newcommand{\E}{\mathbf{E}}

\newcommand{\calI}{\mathcal{I}}
\newcommand{\calJ}{\mathcal{J}}
\newcommand{\calS}{\mathcal{S}}

\newcommand{\calU}{\mathcal{U}}
\newcommand{\calP}{\mathcal{P}}

\newcommand{\Prob}{\mathbf{P}}

\newcommand{\Ind}{\mathds{1}}
\newcommand{\R}{\mathbb{R}}

\newcommand{\calN}{\mathcal{N}}

\newcommand{\Lap}{\mathcal{L}}%laplace transform
\newcommand{\SINR}{\text{SINR}}
\newcommand{\STINR}{\text{STINR}}

\newcommand{\SNR}{\text{SNR}}

\newcommand{\hatS}{E}

\newcommand{\hatM}{\hat{E}_M}
\newcommand{\hatT}{\hat{T}}
\newcommand{\hatt}{\hat{t}}

\newcommand{\bbK}{\mathbb{K}}

\newcommand{\bbA}{\mathbb{A}}
\newcommand{\bbI}{\mathbb{I}}
\newcommand{\bsz}{\boldsymbol{z}}

\newcommand{\bst}{\boldsymbol{t}}
\newcommand{\hatbst}{\hat\bst}
\newcommand{\bfz}{\mathbf{z}}
\newcommand{\bfx}{\mathbf{x}}
\newcommand{\bfeta}{\boldsymbol{\eta}}
\newcommand{\1}{\boldsymbol{1}}

\newcommand{\modZ}{Z'}
\newcommand{\modPsi}{\Psi'}
\newcommand{\modM}{M'}

\newcommand{\markTheta}{\widetilde{\Theta}}
\newcommand{\markPsi}{\widetilde{\Psi}}
\newcommand{\markPhi}{\widetilde{\Phi}}

\newcommand{\tmin}{ t_{\text{min}}}

\newtheorem{thm}{Theorem}
\newenvironment{theorem}{\bf\begin{thm}\rm\em}{\end{thm}} % corollary 
\newtheorem{corollary}[thm]{Corollary}
\newtheorem{lemma}[thm]{Lemma}
\newtheorem{proposition}[thm]{Proposition}
\newtheorem{rem}[thm]{Remark}
\newenvironment{remark}{\bf\begin{rem}\rm}{\end{rem}} % Numbered remark 

\title{
Studying the SINR process of the typical user in Poisson networks by using its factorial moment measures
  \\[2ex]
}
\author{\IEEEauthorblockN{Bart{\l}omiej~B{\l}aszczyszyn\IEEEauthorrefmark{1} and 
Holger~Paul~Keeler\IEEEauthorrefmark{1}}\\
\IEEEauthorblockA{\IEEEauthorrefmark{1}Inria-ENS,
23 Avenue d'Italie, 75214 
Paris, FRANCE\\
Bartek.Blaszczyszyn@ens.fr, Holger.Keeler@inria.fr}}

\begin{document}
\maketitle
\thispagestyle{empty}

\begin{abstract}
Based on a stationary Poisson point process, a wireless network model with random propagation effects (shadowing and/or fading) is considered in order to examine the process formed by the signal-to-interference-plus-noise ratio (SINR) values experienced by a typical user with respect to all base stations in the down-link channel. This SINR process is completely characterized by deriving its factorial moment measures, which involve numerically tractable, explicit integral expressions.  This novel framework naturally leads to expressions for the $k$-coverage probability, including the case of random SINR threshold values considered in multi-tier network models.  While the $k$-coverage probabilities correspond to the marginal distributions of the order statistics of the SINR process, a more general relation is presented connecting the factorial moment measures of the SINR process to the joint densities of these order statistics. This gives a way for calculating exact values of the coverage probabilities arising in a general scenario of signal combination and interference cancellation between base stations.  The presented framework consisting of mathematical representations of SINR characteristics with respect to the factorial moment measures holds for the whole domain of SINR and is amenable to considerable model extension.
\end{abstract}


\section{Introduction}
The steady rise of heterogeneous cellular networks, owing to the
deployment of recent technologies such as femtocells and picocells to
handle increased user-data, is driving the need for new and more
robust analytic methods. Based on information theoretic arguments, a
key performance metric is the signal-to-interference-plus-nose ratio
(SINR)  in the downlink channel experienced by a typical user in the
network. The SINR (or without noise simply SIR) is a function of {\em
  propagation processes}, which incorporate the distance-dependent
path-loss function and (often assumed to be random) fading and/or
shadowing, which we refer to as simply {\em propagation
  effects}. Knowledge of the distribution of the SINR leads to the
{\em coverage probability} and other SINR-based characteristics (such
as e.g the spectral efficiency) of the cellular networks, which can be used to better design and implement such networks. 

The irregularity of network base station locations suggests that their placement is often best assumed to be random. This claim has been supported in recent years with tractable stochastic geometry models based on the Poisson point process yielding accurate solutions~\cite{ANDREWS2011}. In addition to the tractability and `worst-case' arguments for Poisson models, recent convergence results~\cite{hextopoi,hextopoi-journal} has shown that a wide class of stationary network configurations give results for functions of propagation processes, such as the SINR,  as though the placement of the base stations is a Poisson process when sufficiently strong shadowing  is incorporated into the mode, which was previously predicted via simulations~\cite{brown2000cellular}. This adds weight to the argument for using the Poisson model, which leads to more tractable expressions for the SINR and related quantities.

Motivated by the ongoing deployment of heterogeneous 
networks and recent investigations for mathematical expressions
of the distribution of the SINR in Poisson models of these  and other
networks (see related work in Section~\ref{ss.RelatedWork}),
we wish to introduce a general 
framework for studying arbitrary functions of the
process formed by the SINR values experienced by a typical user with respect to
all base stations in the down-link channel. 
To meet this end we completely characterize this process
by deriving explicit, numerically tractable
integral expressions for its factorial moment measures of all orders $n\ge1$.
These measures represent  the expected number of ways that the typical
user  can connect to $n$ different base stations at different  
SINR threshold values.
It is known that various characteristics of an arbitrary 
point process admit series expansion representations  involving
integrals with respect to its factorial moment measures~\cite{fme,fme_spatial}.
Surprisingly, in the case of the SINR
process many interesting characteristics (including $k$-coverage
probability, coverage under signal combination
and interference cancellation) 
can be written as expansions with finite number of
terms, hence the expansions need not be approximations. This observation, reminiscent of the existence of the pole capacity in cellular networks, 
is a consequence of an algebraic
property of the SINR process, limiting the maximum number of stations
that can cover a given point with a SINR bounded away from zero.
The calculation of the corresponding expansion terms 
is arguably more intuitive and computationally more efficient than
previous methods  based on inversions of Laplace transforms for studying similar problems.

The SINR process, considered above, is itself a function
of  the propagation process, already studied in cellular networks,
and thus inherits the invariance/equivalence properties with respect to the
distribution of the random propagation effects (fading and/or
shadowing)~\cite{equivalence2013}. In fact, we also consider  a  process 
formed by the values of the  
signal-to-total-interference-and-noise ratio
  (STINR).
Working with this latter process is more
convenient and a simple mapping allows one to interpret the 
results in terms of the  original SINR process.

In summary, we define a  point  process formed
by the SINR values experienced by a typical user with respect to
all base stations in the down-link channel.
The key result is the complete characterization of the distribution of
this process via its factorial moment measures. These measures 
lead in a natural way to the finite-dimensional
distributions  of the order statistics of the SINR process
thus allowing one to study general functions of several strongest values of
the SINR process.  We present particular probability results for $k$-coverage and coverage with signal
combination and interference cancellation.

\subsection{Related work} \label{ss.RelatedWork}

\subsubsection{Propagation equivalence}
In this work we employ a useful result, recently called {\em  propagation invariance} by B\l aszczyszyn and Keeler~\cite{equivalence2013}, that arises due to the model being based on a Poisson process and the path-loss function being a singular power-law. A classical and related result on the invariance of shot noise for many types of functions was discovered long ago by Gilbert and Pollak~\cite{GILBERT1960},  and later examined in further detail by Lowen and Teich~\cite{LOWEN1990} for power-law functions, hence revealing that Poisson shot noise only depends one moment. 

For wireless networks,  Haenggi~\cite[Proposition 3]{HAENGGI2008} observed the so-called propagation invariance, as a particular case, and derived related results under a Nakagami fading model. In the context of  SINR of cellular networks, B\l aszczyszyn  et al.~\cite{blaszczyszyn2010impact} independently observed this invariance characteristic for interference and propagation losses in general, thus allowing propagation effects to be incorporated into the model by just one moment. Pinto et.al~\cite{PINTO2012} derived and used a similar result to show that the node degree of secrecy graphs (based on Poisson processes) is invariant for the distribution of propagation effects. In all three  papers~\cite{HAENGGI2008,blaszczyszyn2010impact,PINTO2012}, the  invariance results are obtained by defining a point process on the positive real line, which we call the {\em propagation (loss) process}, and showing that if the base station configuration forms a homogeneous Poisson process, then the propagation process is an inhomogeneous Poisson point process on the positive real line. More specifically, B\l aszczyszyn  et al.~\cite[Proposition 5.5]{blaszczyszyn2010impact} observed that this  propagation process depends on only one moment of the propagation effects, and not its distribution, by explicitly deriving its density measure. Conversely, Pinto et.al~\cite[Theorem 3.3]{PINTO2012} did not strictly observe propagation invariance in our sense as they used a more general path-loss function and not necessarily the singular power-law\footnote{Under certain conditions, there exists a trade-off between how general the path-loss function can be and to which stochastic process the invariance result applies. Assuming a singular power-law path-loss function results in the invariance property applying to the propagation processes. The invariance result of Pinto et al.~\cite{PINTO2012} applies to a function of propagation processes, hence it holds under more general path-loss functions in the spirit of Gilbert and Pollak~\cite{GILBERT1960}.}, and proved their invariance result for another quantity by mapping different (inhomogeneous Poisson) propagation processes and obtaining an identical process.  Haenggi~\cite{HAENGGI2008} independently defined the propagation process (as ``path  loss process with fading'') and observed~\cite[Proposition 1]{HAENGGI2008}  that it is a (possibly inhomogeneous) Poisson process without stating its density measure in general form.

\subsubsection{SINR coverage}
%\cite[mukherjee2012ICC]{mukherjee2012ICC}\\
%\cite[mukherjee2011Allerton]{mukherjee2011Allerton}
%\cite[Decreusefond]{VU2012}
%\cite{VU2012,kcovsingle}.

Several authors 
recently developed and studied  single and 
multi-tier network models~\footnote{Terms
 ``multi-tier'' and ``heterogeneous'' networks are often used
interchangeably, but we later adopt a model-specific terminology in
which the former is an instance of the latter.} based on the Poisson
process that led to more or less closed-form expressions for the
(downlink) SINR
coverage probability of the typical
user.
%~\cite{dhillon2011tractable,MADBROWN2011,mukherjee2011Allerton,DHILLON2012,MADBROWN2012,mukherjee2012ICC,MUKHERJEE2012}.
In general, all these models assume that the serving base station(s)
are those that somehow offer the strongest received signal power.
However, they make different assumptions on what kind of 
propagation effects are taken into account when maximizing the received
signal.

Dhillon et al.~\cite{dhillon2011tractable,DHILLON2012}, 
Madhusudhanan et al.~\cite{MADBROWN2011,MADBROWN2012} and
Mukherjee~\cite{mukherjee2011Allerton} assume that 
the selection is made on the {\em instantaneous} (possibly
tier-biased) received power, i.e., when all the propagation effects
are taken into account. In this case the actual distribution of the
propagation effects does not matter (propagation invariance 
admitted in~\cite{MADBROWN2011,MADBROWN2012}), and the results 
should be the same (up to constants multiplying the node densities)
regardless whether the exponential shadowing-and-fading distribution is
assumed for mathematical 
convenience~\cite{mukherjee2011Allerton,dhillon2011tractable,DHILLON2012}
or the expressions are developed using an equivalent  network
representation with no random propagation
effects (the serving station is the closes one)~\cite{MADBROWN2011,MADBROWN2012}. All these works observe
that the analysis simplifies when the SINR (or SIR) threshold $\tau$, 
which is the technology-dependent level that the SINR must exceed to establish a connection, is greater than one (or zero in dB). This assumption on the $\tau$ value (or multiple $\tau$ values in the multi-tier model)  implies that at maximum only one base station in the entire network at any instant can cover the typical user. 
Mukherjee~\cite{mukherjee2011Allerton} goes beyond this assumption
(considers $\tau<1$) using a variant of the
inclusion-exclusion formula. This approach  was used independently 
by Keeler at. al~\cite{kcovsingle}
to study the $k$-coverage in a single-tier network.
Madhusudhanan et al.~\cite{MADBROWN2011,MADBROWN2012} also consider
$\tau<1$ but their results involve the  inversion of Laplace
transforms and hence are less explicit. Inversion of Laplace
transforms is also exploited  in~\cite{hextopoi}.

A different class of models is considered in
Mukherjee~\cite{MUKHERJEE2012}, Jo et al.~\cite{JO2012},
and also, for a single-tier network, in Vu et al.~\cite{VU2012}
and Keeler et al.~\cite[Section IV.C]{kcovsingle}.
These models assume that 
some  propagation effects (say rapidly changing fading) impact the SINR 
value but are  not taken into account when choosing the strongest
(serving) base station. In this case the  coverage probability
depends on the particular distribution of this fading and only
Rayleigh case leads to closed form expressions. 
Moreover, in this case the assumption $\tau\ge1$ is no longer crucial. 
This way of taking into account
the fading (obviously) leads to smaller coverage probabilities and
the difference is relatively important for smaller values of $\tau$,
as illustrated in~\cite{kcovsingle}.

A combination of the two scenarios is considered in
Mukherjee~\cite{mukherjee2012ICC}, where the Rayleigh fading is  assumed
and the serving base station is either the instantaneously the 
strongest one (when accounting for the fading) 
or the one nearest to the user (after accounting for cell
    selection bias).

For tractability, most of the 
aforementioned results assume constant path-loss exponents.
Models with  different (but constant)  path-loss exponent on each tier
are considered in~\cite{JO2012,MADBROWN2012,MUKHERJEE2012,mukherjee2012ICC}.
Authors in~\cite{MADBROWN2012} observe that a multi-tier network is stochastically equivalent to a single-tier network with unity parameters while all the original parameters are incorporated into the density of the (inhomogeneous Poisson) propagation process. This result was later extended in~\cite{equivalence2013} to the case of random path-loss exponents and other parameters in a Poisson network, thus showing a random heterogeneous network is equivalent to a network with constant parameters and an isotropic base station density. This equivalence result,  a variation of which we use here, allowed for the comparison of markedly different networks, for example single-tier and multi-tier types, by examining their equivalent (isotropic) forms.

\subsubsection{Cooperation and interference cancellation}
In this work we do not aim to perform detailed analysis on various types of cooperation or interference cancellation, but rather demonstrate the power of the factorial moment approach by deriving coverage probability for two proposed models of these respective schemes. That said, research featuring stochastic geometry models of cooperation or interference cancellation has remained relatively untouched until recently. Baccelli and  Giovanidis~\cite{giovanidis2013stochastic} recently introduced a stochastic geometry framework for studying pair-wise base station cooperation in cellular networks. They make the tractability assumption of Rayleigh fading~\footnote{It is tempting to suggest that some of these results~\cite{giovanidis2013stochastic} may form part of other work, based on the Poisson point process and a singular path-loss function, that holds under arbitrary propagation effects, but such a thorough and conclusive analysis is not part of this paper.} and use Laplace techniques to derive expressions for the coverage probability under geometry-based cooperation policies with a parameter (for optimization purposes) representing the degree of cooperation. Parallel to this work Akoum and Heath~\cite{akoum2010limited} developed a Poisson-based stochastic geometry model in order to examine interference coordination by introducing the concept of cooperation clusters. They assume Rayleigh fading and a non-singular power-law path-loss function and represent coordination clusters with a  (Poisson-Poisson) cluster process, which leads to bounds on the coverage probability expression.

In the setting of interference cancellation, Zhang and Haenggi~\cite{zhang2012performance,zhangdecoding} tackled the problem of calculating the probability of decoding and removing $n$ interfering signals in  Poisson networks by independently using the aforementioned approach of mapping propagation effects and the path-loss values to the positive real line. Provided  $\tau\geq 1$, they derived~\cite{zhangdecoding}  a closed-form expression  for the SIR-based probability of decoding the $n$-th strongest signal, and used it to obtain bounds   for the mean number of decoded users and the probability of decoding $n$ users. Later Zhang and Haenggi~\cite{zhangsuccessive} applied these results to a heterogeneous network model (with single $\tau$) and concluded that a large part of the performance gained from interference cancellation is possible by just removing the largest interfering signal. Parallel to this work, Quek and et al.~\cite{wildemeersch2013successive}  developed a Poisson-based model of a heterogeneous network with arbitrary propagation effects (and single $\tau$) in order to study the probability of removing the $n$-th strongest signal and to decode the signal of interest after $n$ removals. They used two competing methods (based on Laplace transforms and truncated stable distributions) to derive approximations for the probability of decoding the $n$-th strongest signal given that the previous $n-1$ signals were decoded.   The Laplace method was then used to derive an expression for the coverage probability (given $n$ signals are removed by interference cancellation), which agreed well with simulation results. In these works~\cite{zhang2012performance,zhangdecoding,zhangsuccessive,wildemeersch2013successive}  the line of thought of using the order statistics of the propagation process arises,  but it is then observed that such an approach appears quite intractable (an issue circumvented here by instead deriving and using the order statistics of the SINR process).

\subsubsection{Factorial-moment expansion}
In the theory of point processes factorial moment measures are an
important set of tools that completely characterize a simple point
process, hence they are used in stochastic geometry;  see Daley and
Vere-Jones~\cite{daleyPPI2003,daleyPPII2008} or Stoyan et
al.~\cite{SKM:1995}. It has been shown that expectations of general
functions of simple point processes can be written as (possibly
infinite) expansions of the corresponding factorial moment measures,
provided some convergence condition. This Taylor-like expansion
theorem was first developed for unmarked point processes in one
dimension by B\l aszczyszyn~\cite{fme}, then generalized to higher
dimensions by B\l aszczyszyn et al.~\cite{fme_spatial}. 
Kroese and Schmidt~\cite{fme_marked}
  specifically considered  also independently marked point processes.
It was recently applied in the setting of wireless networks to derive series expansions of the interference process induced by non-Poisson base station configurations~\cite{fmewireless}. Apart from these results, we are unaware of other work where factorial moment measure expansions have been used in a similar setting to that presented here.

To study the SINR process we define a related {\em signal-to-total-interference-and-noise ratio}
  (STINR)  process which is more
tractable. It turns out, this latter process is similar to (and, as recently observed
  in~\cite{sinrPD}, in the case of a noiseless, an  instance of) the well-studied
Poisson-Dirichlet process.  Pitman and Yor~\cite{pitman1997two} surveyed in detail a generalized version of this process and derived various theoretical results. More recently, the factorial moment measures of this generalized version were derived by
Handa~\cite{handa2009two}, the work of which partly inspired us to
derive the joint probability density of the order statistics of the
SINR process, which we use to calculate the coverage probability under
proposed signal combination and interference cancellation models.

\section{Heterogeneous network model and quantities of interest}
We first present a cellular network model under which we define three related point processes on the positive half-line.

\subsection{Network and propagation model}
We consider the ``typical user'' approach where one assumes a typical user is located at the origin and examines what he perceives in the network.  On $\R^2$, we model the base stations with a homogeneous or stationary Poisson point process  $\Phi=\{X\}$ with density $\lambda$. Define the path-loss function as
\begin{equation}\label{PATHLOSS}
\ell(|x|)=(K|x|)^{\beta},
\end{equation}
with path-loss constant $K>0$ and path-loss
exponent $\beta>2$.   Given $\Phi$, let $\{(S_X,P_X)\}_{X\in\Phi}$ be a collection of
independent (across $X$) and identically and \emph{arbitrarily} distributed positive random vectors that form independent marks of $X$.  Let $S_X$ represent  the random propagation effects~\footnote{$S$ may be written as a product of two random variables representing the  fading and shadowing.} from the origin to $X$. For a signal emanating from a base station at $X$, let the $P_X$ represent the power of that signal. Let $(S,P)$ be equal in distribution to $\{(S_X,P_X)\}$ and note that $S$ and $P$ are not necessarily independent. 
We will refer to the above network (model) 
as a {\em heterogeneous network} (model).

\subsection{Propagation (loss) process}
We define the propagation (loss) process, considered as a point process on the positive half-line $\mathbb{R}^{+}$, as
\begin{equation}\label{e.Y}
\Theta=\{Y\}:= \left\{\frac{\ell(|X|)}{P_X S_X } :X\in\Phi   \right\},
\end{equation}
which is a Poisson process; for proof, see~\cite{blaszczyszyn2010impact,hextopoi} where $P$ is not random.
\begin{lemma}[Propagation invariance]\label{l.invariance}
Assume that
\begin{equation}\label{momcond}
  \E[(PS)^{\frac{2}{\beta}}]  <\infty .
\end{equation}
Then the propagation process $\{Y\}$ is an inhomogeneous Poisson point
process with intensity measure
$\Lambda\left(  \left[  0,t\right)  \right)  =a t^{\frac{2}{\beta}}$,
where the propagation constant is
\begin{equation}\label{e.a}
a:=\frac{\lambda\pi
  \E[(PS)^{\frac{2}{\beta}}]}{K^{2}}\,.
\end{equation}
\end{lemma}
\begin{remark} 
In terms of propagation processes, this simple yet useful result allows one to represent propagation effects by setting the product $PS=1$, for example, and replacing $\lambda$ with
$\lambda'=\lambda\E[ (PS)^{2/\beta}]$; for more details see~\cite{equivalence2013}. \footnote{The typical user approach coupled
    with the singular path-loss function (\ref{PATHLOSS})
    allows one to extend the model to $d\,$dimensional space and also
    replace  the constant density of base
    stations by an isotropic power-law function $r^\alpha$ with
    $-d<\alpha<\beta-d$. This
    generalization can be done by simply  replacing  $2/\beta$ by
    $\alpha/\beta+d/\beta$ in the whole statement of
    Lemma~\ref{l.invariance} and $\pi/K^2$ by
    $\nu_d/((1+\alpha/d)K^d)$ in~(\ref{e.a}), where $\nu_d=\pi^{d/2}/\Gamma(1+d/2)$ is the
    volume of the unit-radius $d$-dimensional ball; cf
    e.g.~\cite[Lemma~1]{zhangdecoding}.}
\end{remark}

\subsection{SINR process}
We define the {\em SINR process} on the positive half-line $\mathbb{R}^{+}$ for a typical user as
\begin{equation}\label{SINR}
\Psi=\{Z\}:=\left\{\SINR(X) :X\in\Phi \right\},
\end{equation}
where
\begin{equation}\label{SINRX}
\SINR(X): =\frac{Y^{-1}}{W+\gamma (I-Y^{-1} )},
\end{equation}
the constant $W\geq0$ is the additive noise power, and 
\begin{equation}
I=\sum_{Y\in\Theta}Y^{-1},
\end{equation}
is the power received form the entire network (so that 
$I-Y^{-1}$ is the interference), and the constant $\gamma\in[0,1]$ is parameter that represents the
``strength'' of  interference cancellation techniques.~
If we set $\gamma=0$, then the signal-to-noise ratio (SNR) process, $\SNR(X)=Y^{-1}/W$, is the inverse of the propagation (loss) process rescaled.
\begin{remark}\label{r.gamma}
We consider the model that incorporates both the noise power $W$ and the ``interference factor'' $\gamma$.
Note, however, that the SINR process with general  $W>0, \gamma>0$ is equivalent 
to the  SINR process rescaled by $\gamma^{-1}$ and evaluated for the noise power equal to $W\gamma^{-1}$ with
the interference factor set to~1. Thus $\gamma=1$ can be considered without loss of  mathematical generality. We keep however this parameter to facilitate comparisons to some previous works.
\end{remark}

\subsection{STINR process}
To study $\Psi$, it will be helpful to define the STINR process on  $(0,1/\gamma]$ as
\begin{equation}\label{MSINR}
\Psi'=\{Z'\}:=\left\{\STINR(X) :X\in\Phi \right\}.
\end{equation}
where
\begin{equation}\label{MSINRX}
\STINR(X): =\frac{Y^{-1}}{W+\gamma I},
\end{equation}
Note that the STINR is actually the
    {\em signal-to-total-received-power-and-noise ratio}.\footnote{Total interference
    means the total power received from all transmitters, including the useful signal; cf~\cite[Notes, page 204]{gunnarsson2004power}.}
Naturally, working with $\Psi'$ simplifies the algebra due to the
common denominator in its definition and its
  bounded domain $(0,1/\gamma)$. Information on $\Psi'$ gives information on  $\Psi$ by the relation
\begin{equation}\label{Zrelation1}
Z=\frac{Z'}{1-\gamma Z'},
\end{equation}
or equivalently
\begin{equation}\label{Zrelation2}
Z'=\frac{Z}{1+\gamma Z}.
\end{equation}
The STINR process reveals  also some 
relations with some other, well studied mathematical and physical
models as remarked in what follows.
\begin{remark}\label{rem.STIR-PD}
Note that  $\{Y^{-1}\}$  is
an inhomogeneous Poisson process with intensity measure $(2a/\beta)
t^{-1-2/\beta}dt$. By the construction~(\ref{MSINR}) the STIR process ($W=0$) with $\gamma=1$ is equal to,
what is called in physics, a  Poisson-Dirichlet
process~\cite{panchenko2013sherrington} with parameter $\theta=2/\beta$.~\footnote{It appears
as the  thermodynamic (large  system) limit in the low temperature
regime  of Derrida's  random energy model
and a key component of the so-called Ruelle probability cascades,
which are used to represent the thermodynamic limit of  the
Sherrington-Kirkpatrick model for spin glasses. It should not be confused with 
the (perhaps) better known  Poisson-Dirichlet process of Kingman~\cite{KINGMAN:1993}
for which  $\{Y^{-1}\}$  should be Poisson process with intensity (measure) $\alpha t^{-1}e^{-t}dt$
with some constant $\alpha>0$. In fact 
both processes are special cases of two-parameter Poisson-Dirichlet process extensively studied
in~\cite{pitman1997two}.}
Much is known about the Poisson-Dirichlet process,
including its factorial moment measures, but we cannot use these results in a straightforward manner 
in the case  of STINR process (noisy model). However,  some of our results (in Section \ref{jointorder}) are
partly inspired by a recent approach  to the Poisson-Dirichlet process presented in~\cite{handa2009two}.
 See~\cite{sinrPD} for more details on the relations between the two models and  a few results, which can be
 directly derived using  these relations.
 \end{remark}

\section{Factorial moment measures of the SINR process}
We will derive the factorial moment measures of the STINR process $\{\modZ \} $ defined for $ n\geq 1$ as
% 2c format
%\begin{align}
%\modM^{(n)}&(t_1',\dots , t_n') \nonumber\\
%:&=\modM^{(n)}\left((t_1',1/\gamma]\times\dots \times (t_n',1/\gamma] \right) \nonumber
%\\ 
%&=\E \left(\sum_{{(\modZ_1,\ldots,\modZ_n)\in(\modPsi)^{\times n}\atop\text{distinct}}}
%\prod_{j=1}^{n}  \Ind(\modZ_{j}>t_j')\right) \label{e.Mn},
%\end{align}
\begin{align}
\modM^{(n)}(t_1',\dots , t_n') :=\modM^{(n)}\left((t_1',1/\gamma]\times\dots \times (t_n',1/\gamma] \right) 
=\E \left(\sum_{{(\modZ_1,\ldots,\modZ_n)\in(\modPsi)^{\times n}\atop\text{distinct}}}
\prod_{j=1}^{n}  \Ind(\modZ_{j}>t_j')\right) \label{e.Mn},
\end{align}
where $\Ind$ denotes the indicator function.  The factorial measures
of $\{Z'\}$, which in turn give those of $\{Z\}$, not only completely
characterize these processes, but  
allows one to express in a  natural manner the
  probabilities of various events related to SINR coverage, as we shall see in this paper.

Before deriving $\modM^{(n)}$, we introduce two useful integrals, the first of which arose in the $k$-coverage problem \cite{kcovsingle} while the second is a generalization of another integral in the same work. For $x\ge0$ define
\begin{equation}\label{In}
\calI_{n,\beta}(x)=\frac{2^n
\int\limits_0^{\infty} u^{2n-1}e^{-u^2-u^\beta x\Gamma(1-2/\beta)^{-\beta/2}} du
}{\beta^{n-1}(C'(\beta))^n(n-1)!}
\end{equation}
where 
\begin{equation}
 C'(\beta)=\frac{2\pi}{\beta\sin(2\pi/\beta)}=
\Gamma(1-2/\beta)\Gamma(1+2/\beta).
\end{equation}
Note that
\begin{equation}\label{I0}
\calI_{n,\beta}(0)=\frac{2^{n-1}
}{\beta^{n-1}(C'(\beta))^n}.
\end{equation}
For all $x_i\ge0$ define
% 2c format
%\begin{align}\label{Jn}
%& \calJ_{n,\beta}(x_1,\dots,x_n)= \frac{(1+\sum_{j=1}^{n}   x_j) }{n} \nonumber \\ 
%&\times  \int\limits_{[0,1]^{n-1}}  \frac{  \prod_{i=1}^{n-1}   v_i^{i(2/\beta+1)-1}(1-v_i)^{2/\beta}  }   {  \prod_{i=1}^n  (x_i+\eta_i)}
%   dv_1\dots
%dv_{n-1},
%\end{align}
\begin{align}\label{Jn}
 \calJ_{n,\beta}(x_1,\dots,x_n)= \frac{(1+\sum_{j=1}^{n}   x_j) }{n} 
  \int\limits_{[0,1]^{n-1}}  \frac{  \prod_{i=1}^{n-1}   v_i^{i(2/\beta+1)-1}(1-v_i)^{2/\beta}  }   {  \prod_{i=1}^n  (x_i+\eta_i)}
   dv_1\dots
dv_{n-1},
\end{align}
where 
\begin{equation}
\begin{cases}
\eta_1&= v_1v_2\dots v_{n-1}\\
\eta_2&= (1-v_1)v_2\dots v_{n-1}\\
\eta_3&= (1-v_2)v_3\dots v_{n-1}\\
&\cdots \\
\eta_n&= 1- v_{n-1}.
\end{cases}
\label{e.eta-v}
\end{equation}
\begin{remark}
It is not straightforward to observe (but follows  from 
Theorem~\ref{mainResult}) that 
 $  \calJ_{n,\beta}(x_1,\dots,x_n)$ is invariant under any variable
 permutation. It has been defined to be analogous to the
 single-variable version $\calJ_{n,\beta}(x)$  found in
 \cite{kcovsingle}, with $ \calJ_{n,\beta}(x,\dots,x)= \calJ_{n,\beta}(x)$.
For more remarks on
  $\calJ_{n,\beta}(x_1,\dots,x_n)$, including its representation as a
  functional of $n-1$ independent beta random variables see
  Appendix~\ref{Jnremarks}.
\end{remark} 
We now observe that the factorial moment measure is zero outside a simplex defined by $t_i'$ values.
\begin{lemma}\label{l.modmeasure0}
For $t_i'\in(0,1/\gamma]$, the factorial moment measure of the STINR process  (\ref{MSINR})  satisfies
\begin{equation} \label{momMeasure0}
\modM^{(n)}(t_1',\dots, t_n') =0,
\end{equation}
when 
\begin{equation} 
\gamma \sum_{i=1}^n t_i' \geq 1.
\end{equation}
\end{lemma}
\begin{proof}
This is a direct consequence of a well-known result~\cite[Proposition 6.2]{FnT1}, which states that if the intersection of SINR cells~\cite[Definition 5.1]{FnT1} is not empty, then $\sum_{i=1}^n t_i'\leq 1/\gamma$. 
\end{proof}

We define 
 \begin{equation}\label{hattn}
\hatt_{i}=\hatt_{i}(t_1',\dots,t_n'):=\frac{\gamma t_i'}{1-\gamma\sum\limits_{j=1}^n t_j'}
\end{equation}
and present the main result, which characterizes the STINR process (\ref{MSINR}).
\begin{thm}[Factorial moment measure of $\modPsi$]\label{mainResult}
Assume that moment condition (\ref{momcond}) holds. Then for $t_i'\in(0,1/\gamma]$, the factorial moment measure of order $n \ge1$ of the STINR process  (\ref{MSINR})  satisfies
% 2c format
%\begin{align} \label{momMeasure}
%&\modM^{(n)}(t_1',\dots  t_n') \nonumber\\
%&= n!  \left( \prod\limits_{i=1}^{n}\hatt_i^{-2/\beta} \right)
%\calI_{n,\beta}((W/\gamma)a^{-\beta/2}) 
% \calJ_{n,\beta}(\hatt_1,\dots,\hatt_n),
%\end{align}
\begin{align} \label{momMeasure}
\modM^{(n)}(t_1',\dots  t_n') = n!  \left( \prod\limits_{i=1}^{n}\hatt_i^{-2/\beta} \right)
\calI_{n,\beta}((W/\gamma)a^{-\beta/2}) 
 \calJ_{n,\beta}(\hatt_1,\dots,\hatt_n),
\end{align}
when 
 \begin{equation}\label{e.tplimits}
\gamma\sum_{i=1}^n t_n' <1,
\end{equation}
and 
\begin{equation}  \label{momMeasure00}
\modM^{(n)}(t_1',\dots  t_n') =0,
\end{equation}
otherwise.
\end{thm}

\begin{proof} 
Expression (\ref{momMeasure00}) is due to Lemma \ref{l.modmeasure0}. For expression (\ref{momMeasure}), the proof is  included in  Appendix~\ref{mainResultProofPart}. It follows in a similar fashion to that of the main theorem in~\cite{kcovsingle} . 

\end{proof}
We immediately obtain the moment measure of the SINR process, which is defined by
% 2c format
%\begin{align}
%&M^{(n)}\left(t_1,\dots t_n,\right)  \nonumber \\
%:&=M^{(n)}\left((t_1,\infty]\times\dots \times (t_n,\infty] \right)   \\ 
%&=\E \left(\sum_{{(Z_1,\ldots,Z_n)\in (\Psi)^{\times n}\atop\text{distinct}}}
%\prod_{j=1}^{n}  \Ind(Z_{j}>t_j)\right)\label{defMomMeasureZ}.
%\end{align}
\begin{align}
M^{(n)}\left(t_1,\dots t_n,\right)  
:=M^{(n)}\left((t_1,\infty]\times\dots \times (t_n,\infty] \right)  
=\E \left(\sum_{{(Z_1,\ldots,Z_n)\in (\Psi)^{\times n}\atop\text{distinct}}}
\prod_{j=1}^{n}  \Ind(Z_{j}>t_j)\right)\label{defMomMeasureZ}.
\end{align}

\begin{corollary}[Factorial moment measure of $\Psi$]\label{c.zmom}
Assume the propagation moment condition (\ref{momcond}). Then for $t_i\in(0,\infty)$, the SINR process  (\ref{SINR}) has the moment measure
\begin{equation}\label{e.MmodM}
M^{(n)}\left(t_1,\dots t_n,\right)  = \modM^{(n)}\left(t_1',\dots, t_n'\right),
 \end{equation}
 where
\begin{equation}
t'_i=\frac{t_i}{1+\gamma t_i}.
\end{equation}
\end{corollary}

\begin{proof}
The result follows from the relationship between $\Psi_i$ and $\Psi_i'$ captured in expressions (\ref{Zrelation1}) and (\ref{defMomMeasureZ}). 
\end{proof}

\begin{remark}[Noise factorization]
\label{rem.factorization}
Denote by  $ \modM^{(n)}_{W=0,\gamma=1} $ and  $ M^{(n)}_{W=0,\gamma=1}$  respectively the
factorial moment measures of the STIR and SIR processes %i.e.,  $\modM^{(n)} $ and  $ M^{(n)}$ 
with   $W=0$ and $\gamma=1$.  
By~(\ref{momMeasure}),  Remark~\ref{r.gamma} and Corollary~\ref{c.zmom} we have immediately the  following 
factorization of the noise in the factorial moment measures with arbitrary $W$ and $\gamma$
\begin{align}
\modM^{(n)}(t_1',\ldots,t_n')&=\bar{ \calI}_{n,\beta}((W/\gamma)a^{-\beta/2})
\modM^{(n)}_{W=0,\gamma=1}(\gamma t_1',\ldots,\gamma t_n')\label{noiseM},\\
M^{(n)}(t_1,\ldots,t_n)&= \bar{ \calI}_{n,\beta}((W/\gamma)a^{-\beta/2})M^{(n)}_{W=0,\gamma=1}(\gamma t_1,\ldots,\gamma t_n) ,
 \end{align}
where
\begin{equation}
\bar{ \calI}_{n,\beta}(x)=\frac{ \calI_{n,\beta}(x)} { \calI_{n,\beta}(0)} .
 \end{equation}
\end{remark}
This ability to factor out 
the noise effect in the factorial moment measures is  convenient and is
reminiscent of factoring out the noise term in the distribution of the SINR under Rayleigh fading, an
assumption that is not required, however, in our present setting. In particular it allows one to 
express   the densities of the factorial moment measures $M'^{(n)}$ of the STINR process 
as follows, using the corresponding densities of the Poisson-Dirichlet process from~\cite[Theorem
  2.1]{handa2009two}; cf Remark~\ref{rem.STIR-PD}.   
For $n\ge0$ denote 
$$c_{n,\alpha,\theta}=\prod_{i=1}^n\frac{\Gamma(\theta
+1+(i-1)\alpha)}{\Gamma(1-\alpha)\Gamma(\theta +i\alpha)}\,,$$
in particular  $c_{n,2/\beta,0}=   (2/\beta)^{n-1} \Gamma(n )/(\Gamma(2n/\beta )\Gamma(1-2/\beta )^n)$.
\begin{corollary}\label{newmu_n}
For the STINR process $\Psi'$ $(W\geq 0)$, the $n\,$th factorial
moment density is given by
\begin{align} 
&\mu'^{(n)}(t_1',\dots  t_n') :=(-1)^n
\frac{\partial^{n}\modM^{(n)}(t_1',\dots  t_n')}{\partial
  t_1'\dots\partial t_n' }
\label{mu_n}
\\
&\hspace{-1em}= c_{n,2/\beta,0}\,\bar{ \calI}_{n,\beta}((W/\gamma)a^{-\beta/2})
\gamma^{n} \Bigl( \prod\limits_{i=1}^{n}(\gamma t_i')^{-(2/\beta+1)} \Bigr)
\Bigl(1- \sum\limits_{j=1}^{n}(\gamma t'_j) \Bigr)^{2n/\beta -1}\hspace{-3em}\,
\nonumber
\end{align}
for $(t_1',\dots,t_n')$ satisfying~(\ref{e.tplimits}) and 0 otherwise.
\end{corollary}
We are unaware of anybody showing  the equivalence of Propositions \ref{mainResult} and  Corollary~\ref{newmu_n}, either by
 differentiating the    measure (\ref{momMeasure}) or integrating the density (\ref{mu_n}).

\section{Coverage probabilities in heterogeneous networks with varying SINR thresholds}
In this section we will present a result in which the factorial
moments of the SINR process naturally arise, thus illustrating an
intuitive and convincing reason for their introduction before applying
them in a more general setting. We will see that the  $k$-coverage
probability framework introduced in \cite{kcovsingle} holds for very
general settings. In this regard, we will define marked versions of
the previously introduced point processes with the inclusion of SINR
threshold values. Informally, we will refer to a network (model) in
which each value of $S$, $P$ {\em and the SINR threshold} is random
and depends on each base station as a heterogeneous network (model)
{\em with varying SINR thresholds}.

\subsection{Heterogeneous network model  with varying SINR thresholds}
Given $\Phi=\{X\}$, let $\{T_X\}_{X\in\Phi}$ be a collection of positive random variables that each represent the SINR threshold of a base station. For each $X\in \Phi$, the positive random vector $(S_X,P_X,T_X)$ forms its independent mark. Our network model is now described by the independently marked point process
\begin{equation}
\markPhi:=\{( X,(S_X,P_X,T_X))\}.
\end{equation}
Note that the coordinates of $(S_X,P_X,T_X)$  are not  necessarily independent to each other. The marked process $\markPhi$ induces the independently marked propagation process 
\begin{equation}
\markTheta:=\{( Y,T)\}=\left\{\left(\frac{\ell(X)}{P_X S_X},T_X \right)\right\}.
\end{equation}
Effectively, what the typical user `experiences' in a heterogeneous network is captured by the independently marked point process $\markTheta$. This process admits a propagation invariance result~\cite{equivalence2013}, analogous to Lemma \ref{l.invariance}, which was originally given with random $\beta$, but is now presented with constant $\beta$ (see Appendix \ref{proofmarkedinvariance} for proof). 
\begin{lemma}[Marked propagation invariance]\label{l.markedinvariance}
Assume that 
\[
\E[ (PS)^{2/\beta}] < \infty.
\]
Then the propagation process $\markTheta$ is an independently marked inhomogeneous Poisson point process on $\R_+$ with intensity measure
% 2c format
%\begin{align}
%\Lambda(s,t) &:= \E[\sum_{(Y,T)\in\markTheta} \Ind ( Y\leq s, T\leq t )] \\
%&= \frac{\lambda \pi s^{2/\beta}}{K^2}  \E \left[(PS)^{2/\beta} \Ind ( T\leq t ) \right]\label{e.Lambda}. 
%\end{align}
\begin{align}
\Lambda(s,t) := \E[\sum_{(Y,T)\in\markTheta} \Ind ( Y\leq s, T\leq t )] 
= \frac{\lambda \pi s^{2/\beta}}{K^2}  \E \left[(PS)^{2/\beta} \Ind ( T\leq t ) \right]\label{e.Lambda}. 
\end{align}
\end{lemma}

If two network models induce the same marked propagation processes $\markTheta$, then we say they are {\em stochastically equivalent}.
Furthermore, two equivalent networks also induce the same independently marked SINR process
\begin{equation}
\markPsi:=\{(Z,T)\}.
\end{equation} 
where $Z=\SINR(X)$ as given by expression (\ref{SINRX}). 

Lemma \ref{l.markedinvariance} allows one to construct an equivalent network with some of the previously random marks set to constants.  There is a subtlety in this ability to `push' randomness away from certain marks onto others that serves as a useful technique in proofs, which we leverage in the next result.

\begin{corollary}\label{c.phistar}
For the network $\markPhi$, there is a stochastically equivalent network
\begin{equation}
\markPhi^*=\{( X^*,(S_{X^*}=1,P_{X^*}=1,T_{X^*}))\},
\end{equation}  
with density 
\begin{equation}
\lambda^*:=\lambda \E[(PS)^{2/\beta}]=a K^2/\pi,
\end{equation}
where given $\{(X^*) \}$, the marks $T_{X^*}$ are independent (across $X^*$), and with a distribution given by
\begin{equation}\label{e.FT}
F_{T^*}:=\Prob(T^* \leq t)=\frac{ \E \left[(PS)^{2/\beta} \Ind ( T\leq t ) \right],}{ \E[(PS)^{2/\beta}]}.
\end{equation}
Consequently, $\markPhi^*$ induces the marked SINR process
\begin{equation}
\markPsi^*:=\{(Z^*,T^*)\},
\end{equation}
which is equal in distribution to $\markPsi$, hence the factorial moment measures of $\{Z^*\}$ are equal to these of $\{Z\}$, namely $M^{(n)}$ with propagation constant $a=\pi\lambda^*/K^2$. The same is true for the STINR process induced by $\markPhi^*$.
\end{corollary}
\begin{proof}
Set the marks to $S=1$ and $P=1$, substitute $F_{T^*}$, and verify that the resulting propagations process forms a marked Poisson process with intensity measure $\Lambda(s,t)$ given by (\ref{e.Lambda}). 
\end{proof}

\subsection{Coverage number and symmetric sums}\label{s.symsum}
We now introduce some quantities of interest that first appeared in the single-tier case~\cite{kcovsingle}, but hold under the more general setting of heterogeneous networks. We consider the \emph{coverage number} of the typical user, which is defined as the
number of base stations to which the typical user can connect, namely
\begin{equation}\label{e.N}
\calN:=\sum_{(Z,T)\in\markPsi}\Ind \left[Z>T\right].
\end{equation}
The probability  of the typical user being covered by
at least $k$ base stations, or the \emph{$k$-coverage probability}, is
\begin{equation}\label{Pck-def}
\calP^{(k)}:=\Prob\{\,\calN\geq k\,\}.
\end{equation}
In particular, the coverage probability of the typical
user  is $\calP:=\calP^{(1)}$. In the spirit of the previous $k$-coverage result~\cite{kcovsingle}, we introduce the notion of symmetric sums. For $n\ge 1$, define the  $n\,$th symmetric sum  
\begin{equation}\label{e.SS}
\calS_n :=\E\Bigl[\sum_{{(Z_i,T_i)\in(\markPsi)^{\times n}\atop\text{distinct}}}
\Ind\left(\,Z_i> T_i,  \, i=1,\dots,n \,\right)\Bigr]
\end{equation}
We set $\calS_0:= 1$.
$\calS_n$ represents the expected number of ways of choosing $n$ different base stations $X_1,\ldots,X_m$ each received by the user at the origin with the 
 SINR larger than it specific value $T_{X_i}$. We can call it the {\em expected $n$-coverage number}.
%Observe that $\calS_n$ is the expected number of ways that the typical user can connect to $n$ base stations that exceed their specific SINR threshold values $T_X$.

The following $\calS_n$-based identities are akin to the famous inclusion-exclusion principle (for example, see~\cite[IV.5 and
IV.3]{Fel68} for~(\ref{e.ss1}) and~(\ref{e.ss2}), respectively) and stem from the Schuette-Nesbitt formula (which we recall for completeness 
in Appendix~\ref{App1}, see also  \cite{GERBER:1995}). 
\begin{lemma}\label{l.SS-represent}
We have for  $k\ge1$ 
\begin{eqnarray}
 \calP^{(k)}&=&\sum_{n=k}^{\infty}  (-1)^{n-k}{n-1\choose k-1}\calS_n\,,\label{e.ss1}\\
\Prob\{\,\calN=k\,\}&=&\sum_{n=k}^{\infty}  (-1)^{n-k}{n\choose k}\calS_n\,,\label{e.ss2}\\
 \E[z^{\calN}]&=&\sum\limits_{n=0}^{\infty} (z-1)^n \calS_n\,,
\quad z\in[0,1]\,, \label{e.ss3}\\
 \E[\calN]&=&\calS_1\,.\label{e.ss4}
\end{eqnarray}
\end{lemma}
We will see,  via Lemma \ref{l.finitesums}, that the above (apparently infinite) summations reduce to finite
sums (owing to $\calS_n=0$ for large enough $n$) under very reasonable model conditions.

\begin{proposition}\label{generalSn}
For the network $\markPhi$, 
under the assumptions of
  Theorem~\ref{mainResult}, the $n$\,th symmetric sum is 
\begin{equation}\label{e.Snform1}
\calS_n = \frac{1}{n!} \int_{(\R_+)^n}\left( \prod\limits_{i=1}^n F_{T^*} (t_i) \right) M^{(n)}(dt_1,\dots,  dt_n) ,
\end{equation}
which is equivalent to
\begin{equation} \label{e.Snform2}
\calS_n = \frac{1}{n!} \int_{(\R_+)^n}M^{(n)}(t_1,\dots,  t_n) \prod\limits_{i=1}^n F_{T^*} (dt_i),
\end{equation}
where $M^{(n)}$ and $F_{T^*}$ are given by (\ref{e.MmodM}) and  (\ref{e.FT}) respectively. 
\end{proposition}
\begin{proof}
Corollary \ref{c.phistar} implies
\[
\calS_n :=\E\Bigl[\sum_{{(Z_i^*,T_i^*)\in (\markPsi^*)^{\times n} \atop\text{distinct}}}
\Ind\left(\,Z_i^*> T_i^*,  \, i=1,\dots,n \,\right)\Bigr].
\]
Denote by $\Psi^*=\{Z^*\}$ the process of points of  $\markPsi^*$. With this notation, using
  the fact that $\markPsi^*$ is independently marked,  we have
\begin{align}
\calS_n &=\E\Bigl[\sum_{{(Z_i^*)\in(\Psi^*)^{\times n}\atop\text{distinct}}}
\Prob\left(\,Z_i^*> T_i^*,  \, i=1,\dots,n \, |\Psi^* \right)\Bigr] \nonumber \\
&=\E\Bigl[\sum_{{(Z_i^*)\in(\Psi^*)^{\times n} \atop\text{distinct}}}
\prod_{i=1}^n F_{T^*}(Z_i^*) \Bigr] \nonumber  \\
&= \frac{1}{n!} \int_{(\R_+)^n}\left( \prod\limits_{i=1}^n F_{T^*} (z_i) \right) M^{(n)}(dz_1,\dots,  dz_n) \nonumber ,
\end{align}
where the last line follows from Campbell's theorem for marked point processes and the definition of factorial moment measures~\cite[Chapter 4]{SKM:1995}. For equation (\ref{e.Snform2}), by  Lemma \ref{l.modmeasure0} and Corollary \ref{c.zmom}, 
\begin{equation}
\lim\limits_{t_i\rightarrow\infty}M^{(n)}(t_1,\dots,  t_n)=0, \qquad i\in[1,n].
\end{equation}
Note that $M^{(n)}(t_1,\dots,  t_n)$ is a decreasing function in $t_i$, and
\begin{equation}
\lim\limits_{t_1\rightarrow0} F_{T^*} (t_i) =0 .
\end{equation}
Then apply integration by parts 
% 2c format
%\begin{align}
% \int_{(\R_+)^n}&M^{(n)} (t_1,\dots,  t_n) \prod\limits_{i=1}^n F_{T^*}(dt_i) \nonumber \\
%=&  \int_{(\R_+)^n}M^{(n)} (dt_1,\dots,  t_n)  F_{T^*}(t_1) \prod\limits_{i=2}^n F_{T^*}(dt_i) ,
%\end{align}
\begin{align}
 \int_{(\R_+)^n}M^{(n)} (t_1,\dots,  t_n) \prod\limits_{i=1}^n F_{T^*}(dt_i) 
=  \int_{(\R_+)^n}M^{(n)} (dt_1,\dots,  t_n)  F_{T^*}(t_1) \prod\limits_{i=2}^n F_{T^*}(dt_i) ,
\end{align}
and repeat until proof is completed.
\end{proof}

The next result shows that the expressions involving infinite sums of $\calS_n$ reduce to finite sums. For real $x$ denote by $\lceil x\rceil$  the ceiling of $x$ (the smallest integer not less than $x$).

\begin{lemma}\label{l.finitesums}
Assume there exists $\tmin>0$ such that $\Prob(T \ge \tmin)=1$. Then
$\calS_n=0$ for $n\ge 1/(\gamma \tmin)+1 $, which implies that one can
replace $\infty$ by $\lceil 1/(\gamma \tmin) \rceil$ in the sums in
expressions given in Lemma~\ref{l.SS-represent}, namely equations (\ref{e.ss1})--(\ref{e.ss3})
\end{lemma}
\begin{proof}
If $T \ge \tmin$ almost surely, then the same
  holds true for $T^*$. Consequently by~(\ref{e.Snform2}) and 
  Lemma~\ref{l.modmeasure0}  $\calS_n=0$ when
  $n\gamma\tmin/(1+\gamma\tmin)\ge 1 $. The largest $n$ for which the
  opposite, strict inequality holds is $\lceil 1/(\gamma \tmin) \rceil$.
\end{proof}

\subsection{Multi-tier network example}\label{s.multitier}
We illustrate our $k$-coverage  framework by examining an increasingly popular multi-tier cellular network model. Informally, we refer to a heterogeneous network as a multi-tier network  if the marks $T$ are set to constants that depend only on the tier to which the base station belongs. 
%~\cite{DHILLON2012,MADBROWN2011,MUKHERJEE2012}

On $\R^2$ we represent $m$ tiers of base stations with $m$ independent homogeneous Poisson point processes $\{\Phi_j\}$ with  densities $\{\lambda_j\}$. More precisely, consider the independently marked point process
\begin{equation}
\markPhi_j:=\{( X_j,(S_{X_j},P_{X_j},\tau_{j}))\},
\end{equation}
where the SINR threshold mark has been set to a non-random $\tau_j$, and where the distribution of the random vector $(P_j,S_j)$ is equal to that of $(S_{X_j},P_{X_j})$, such that the base station power and propagation effects may depend on the tier, and
\begin{equation} \label{momcondj}
\E[(P_jS_j)^{2/\beta}]<\infty, \qquad j=1,\dots,m.
\end{equation}
Assume without loss of generality that $\tau_j\not=\tau_k$ for $j\not=k$.
Let $K$ and $\beta$ again be the path-loss parameters for the entire network, and denote
\begin{equation}
\lambda_j^*=\lambda_j\E[(P_j S_j)^{2/\beta}] .
\end{equation}
We call the superposition $\markPhi:=\cup_{j=1}^m\markPhi_j$ the $m$-tier network model. 

\begin{corollary}\label{c.singletier}
An equivalent (to $\markPhi$) single-tier network $\markPhi^*$ exists with the path-loss exponent $\beta$, path-loss constant $K=1$, assumes $P^*\equiv1$, $S^*\equiv1$, and a base station density given by
\begin{equation}\label{e.lambda*}
\lambda^*=\sum\limits_{j=1}^m\lambda_j^*,  
\end{equation}
and marks $T^*$ with distribution
\begin{equation}
\Prob (T^* =\tau_j)=\frac{\lambda_j^*}{\lambda^*}, \qquad j=1,\dots, m.
\end{equation}
\end{corollary}

\begin{proof}
For each $j\,$ tier $\markPhi_j$, there is an equivalent tier $\markPhi^*_j$ with $P_j=1$, $S_j=1$, and density $\lambda^*_i$. Using the superposition theorem~\cite{KINGMAN:1993}, the result follows from Corollary \ref{c.phistar}. 
\end{proof}

\begin{remark}
Note that a ``randomly'' selected base station of $\Phi$ belongs to the $j\,$th tier with probability
$\lambda_j/\lambda$, where $\lambda=\sum_{j=1}^m\lambda_j$, while a  ``randomly'' selected signal from the propagation process $\{Z^*\}$ originates from a base station in the $j\,$th tier with probability $\lambda_j^*/\lambda^*$.
\end{remark}

The following result gives an explicit $\calS_n$ expression for the multi-tier model with arbitrary propagation effects in the domain of all $t_i>0 $, and thus, in conjunction with  Lemma~\ref{l.SS-represent}, allows one to calculate the $k$-coverage probabilities, which generalizes previous results~\cite{MADBROWN2011,DHILLON2012,kcovsingle}. 
\begin{corollary}\label{cor.symmetric}
Assume propagation moment condition (\ref{momcond}) for all $m$ tiers of a mult-tier network. 
Then 
\begin{equation}
\calS_n=\frac{1}{n!}\left[  \sum_{\substack{j\in\{1,\dots,n\}\\i_j\in \{1,\dots,m\}}}\mkern-25mu  M^{(n)}\left(\tau_{i_1},\ldots,\tau_{i_n} \right)   \prod_{j=1}^n \lambda^*_{i_j}/\lambda^*  \right].
\end{equation}
\end{corollary}

\begin{proof} 
This stems directly from  Corollary \ref{c.singletier} amd expression (\ref{e.Snform2}). 
\end{proof}

\subsection{Single-tier network example}\label{s.singletier}
For the special case of a single-tier stationary Poisson network model $\Phi=\{X\}$ with density $\lambda$ and a (deterministic) SINR threshold $\tau$  the $k$-coverage probability expressions derived in~\cite{kcovsingle} follow. The second statement (case with $\gamma\tau\ge1$)
  is a special case of~\cite[Theorem~1]{DHILLON2012}.
\begin{corollary}\label{singletier}
Assume the moment condition $\E[(PS)^{2/\beta}]<\infty$ for a single-tier network.
Then 
\begin{equation}
\calS_n=\calS_n(\tau)=\tau_n^{-2n/\beta}  \calI_{n,\beta}((W/\gamma)a^{-\beta/2})  \calJ_{n,\beta}(\tau_n),
\end{equation}
for $0<\gamma\tau<1/(n-1)$  and $\calS_n=0$ otherwise, 
where $a$ is  given by~(\ref{e.a}) and $\tau_n$ is given by
\[
\tau_n:=\tau_n(\tau)=\frac{\gamma\tau}{1-(n-1)\gamma\tau},
\] 
Moreover, the $k$-coverage probability is
% 2c format
%\begin{align*}
%\calP^{(k)}&=\calP^{(k)}(\tau)\\
%&=\sum_{n=k}^{\lceil 1/(\gamma\tau)\rceil}  
%\scriptstyle{(-1)^{n-k}{n-1\choose  k-1}}
%\tau_n^{-2n/\beta}  \calI_{n,\beta}((W/\gamma)a^{-\beta/2})
%\calJ_{n,\beta}(\tau_n)\,,
%\end{align*}
\begin{align*}
\calP^{(k)}=\calP^{(k)}(\tau)
=\sum_{n=k}^{\lceil 1/(\gamma\tau)\rceil}  
\scriptstyle{(-1)^{n-k}{n-1\choose  k-1}}
\tau_n^{-2n/\beta}  \calI_{n,\beta}((W/\gamma)a^{-\beta/2})
\calJ_{n,\beta}(\tau_n)\,,
\end{align*}
For $\gamma\tau\ge 1$ we have $\lceil 1/(\gamma\tau)\rceil=1$ and 
\begin{equation}\label{e.SINRexplici}
\calP^{(1)}(\tau)=\frac{2(\gamma\tau)^{-2/\beta}}{\Gamma(1+\frac{2}{\beta})}
\int_0^\infty \!\!\!\!\!\!ue^{-u^2\Gamma(1-2/\beta)- (W/\gamma)a^{-\beta/2} u^\beta}du\,.
\end{equation}
\end{corollary}

\begin{remark}
If we set $W=0$ (to consider an ``interference limited network''),
then $\{Z\}$ becomes the SIR process, which is clearly scale-invariant
due to  the definition of the regular and STIR process in
terms of the propagation process (also evident by the factorial moment
measures being independent of $\lambda$). Under the stationary
single-tier Poisson model with $\gamma=1$, the inverse of the largest SIR value, known as the interference factor $f$, has been studied~\cite{hextopoi}  by deriving its Laplace transform
\begin{align}
\Lap_f(\xi):&=\E [e^{-\xi f}]=\frac{1}{\varphi_{\beta}(\xi)},
\end{align}
where
\[
\varphi_{\beta}(\xi):=e^{-\xi} + \xi^{2/\beta}\gamma(1-2/\beta,\xi),
\]
and $\gamma(z,\xi)=\int_0^{\xi}t^{z-1}e^{-t}dt$ is the lower incomplete gamma function. The distribution of $f$ is directly related to $\calP$  (in 
the interference limited network) by $\Prob(f\leq s)=\calP(1/s)$, which gives the Laplace transform relation
\begin{align}
\xi\int_0^{\infty} \calP(1/s) e^{-\xi s} ds&= \Lap_f(\xi)\nonumber =\frac{1}{\varphi_{\beta}(\xi)}.
\end{align}
\end{remark}

\section{Finite-dimensional distributions of
the  SINR process}
In this section we show how the factorial moment measures of the SINR process can be used to derive joint densities of the $k$ strongest 
values of the SINR process. This allows one to express, for example, coverage probabilities 
under interference cancellation and/or base station cooperation.

\subsection{Joint densities of order statistics}\label{jointorder}
Denote the order statistics of the STINR process $\Psi'=\{Z'\}$ by 
\[
Z'_{(1)}> Z'_{(2)}> Z'_{(3)}\dots,
\]
such that $Z'_{(1)}$ is the largest STINR value in $\Psi'$. In complete analogy one can consider order statistics
$Z_{(1)}> Z_{(2)}> Z_{(3)}\dots$ of
the SINR process   $\Psi=\{Z\}$.  However, as we shall  see, 
$\Psi'$ is a more apt tool when  studying interference cancellation and base station cooperation.

\begin{remark}\label{r.orderstat}
By its definition~(\ref{Pck-def}), $\calP^{(k)}(\tau)$ evaluated in Corollary~\ref{singletier}, as a function of $\tau$, 
is the tail-distribution of the $k\,$th order statistics of the SINR process, $\calP^{(k)}(\tau)=\Prob\{Z_{(k)}>\tau\}$.
By the fact that the mappings (\ref{Zrelation1}) and
(\ref{Zrelation2}) are monotonic, increasing,
$\calP^{(k)}(\tau'/(1-\gamma\tau'))=\Prob\{Z'_{(k)}>\tau'\}$.
Moreover, given that the functions $x/A$ and $x/(A-x)=A/(A-x)-1$ are increasing in $x$, $Z_{(i)}$ (or $Z'_{(i)}$) represent the value of the 
SINR (or STINR) experienced by the typical user with respect to the base station offering the 
$k\,$th smallest propagation loss $Y^{(k)}$, where $Y^{(1)}<Y^{(2)}<\ldots$ is the process of order statistics of $\{Y\}$.
\end{remark}

We will express now the joint probability density $f'_{(k)}(z'_1,\dots,z'_k)$ 
of the vector 
$(Z'_{(1)},\dots,Z'_{(k)})$ of $k$ ($k\ge1$) largest order
statistics of the STINR process. 

For $i\ge 1$, write 
% 2c format
%\begin{align}
% &\mu_k'^{(k+i)}(z'_1,\dots,z'_k) \label{e.mu'ki} \\
%&= \int_{z'_k}^{1/\gamma} \hspace{-1em}\dots\int_{z'_k}^{1/\gamma} \mu'^{(k+i)}(z'_1,\dots,z'_k, \zeta'_1,\dots,\zeta'_i) \,d\zeta'_1\dots  d\zeta'_i,\nonumber
%\end{align}
\begin{align}
\mu_k'^{(k+i)}(z'_1,\dots,z'_k) \label{e.mu'ki} 
= \int_{z'_k}^{1/\gamma} \hspace{-1em}\dots\int_{z'_k}^{1/\gamma} \mu'^{(k+i)}(z'_1,\dots,z'_k, \zeta'_1,\dots,\zeta'_i) \,d\zeta'_i\dots  d\zeta'_1,
\end{align}
where $\mu'^{(n)}(t'_1,\dots,t'_n)=(-1)^n
\frac{\partial^nM'^{(n)}}{\partial t'_1\ldots\partial
  t'_n}(t'_1,\ldots,t'_n)$
is the density of the factorial moment measure $M'^{(n)}$ of the
STINR process  given in~Corollary~\ref{newmu_n}.
It is easy to see that~(\ref{e.mu'ki}) can be
evaluated more directly differentiating  $k$ times  $M'^{(n)}$ 
% 2c format
%begin{align}\label{e.mu'ki-der}
%  \mu_k'^{(k+i)}(z'_1,\dots,z'_k)\\
%  = (-1)^k \frac{\partial^k
%'^{(k+i)}(t'_1,...,t'_k,z'_k,...,z'_k)}{\partial t'_1\dots\partial t'_k}
%(t'_1=z'_1,...,t'_k=z'_k).\nonumber
%\end{align}
\begin{align}\label{e.mu'ki-der}
  \mu_k'^{(k+i)}(z'_1,\ldots,z'_k)
  = (-1)^k \frac{\partial^k
 M'^{(k+i)}(t'_1,\ldots,t'_k,z'_k,\ldots,z'_k)}{\partial t'_1\dots\partial t'_k}
(t'_1=z'_1,\ldots,t'_k=z'_k).
\end{align}
For convenience we define also $\mu_k'^{(k+i)}$ for $i=0$ as the
  density of $M'^{(k)}$; $\mu_k'^{(k)}:=\mu'^{(k)}$.
We will call 
$\mu_k'^{(n)}$ ($k\le n$) ``partial densities'' of $M'^{(n)}$. In
Appendix~\ref{densityMn} and~\ref{matrixformMn} we will show how these
partial densities can be explicitly calculated.

The equivalent notations apply to the SINR process $\Psi=\{Z\}$,
which we skip without loss of completeness.

\begin{proposition}\label{p.f'k}
Under the assumptions of Theorem~\ref{mainResult},
the joint probability density of the vector of $k$ strongest values 
of the STINR process $(Z'_{(1)},\dots,Z'_{(k)})$ is equal to
\begin{equation}\label{e.f'k}
f'_{(k)}(z'_1,\dots,z'_k)=\sum_{i=0}^{i_{\max}} \frac{(-1)^i}{i!}  \mu_k'^{(k+i)}(z'_1,\dots,z'_k),
\end{equation} 
for $z'_{1}>z'_{2}>\dots >z'_{k}$ and 
$f'_{(k)}(z'_1,\dots,z'_k)=0$ otherwise, 
where the upper summation limit  is bounded by
\begin{equation}\label{e.f'imax}
i_{\max}< \frac{1}{\gamma z'_k}-k.
\end{equation} 
\end{proposition}
\begin{proof}
The result is effectively that of Handa~\cite[Lemma 5.3]{handa2009two}
who details how it follows from a relationship between finite point
processes and their Janossy measures (see~\cite[Section
5.4]{daleyPPI2003}). 
The original result has the infinite series in~(\ref{e.f'k}).
However, by Lemma~\ref{l.modmeasure0}
$M'^{(i+k)}(t'_1,...,t'_k,z'_k,...,z'_k)=0$ when  $t'_1+\ldots+
t'_k+iz_k\ge 1/\gamma$ and consequently, by~(\ref{e.mu'ki-der}),
$\mu_k^{(k+i)}(z'_1,\dots,z'_k)=0$ whenever
$z'_1+...+z'_{k-1}+(i+1)z'_k\ge 1/\gamma$. Hence, a necessary condition for  
$\mu_k^{(k+i)}(z'_1,\dots,z'_k)>0$ is $z'_1+\dots (i+1)z'_k< 1/\gamma$
which, since $z'_1>\dots>z'_k$, implies $(k+i)z'_k< 1/\gamma$. This
proves~(\ref{e.f'imax}).
The original result of Handa has also a  convergence condition
$$\sum_{n=0}^\infty \frac{M'^{(n)}(c,\ldots,c)}{n!}(1+\varepsilon)^n<\infty$$
for each $c\in(0,1/\gamma)$ and some $\varepsilon=\varepsilon(c)$,
which is trivially satisfied in our case by the previous observation
regarding $i_{\max}$.
\end{proof}
The bound on $i_{\max}$ 
given in~(\ref{e.f'imax}) is only an upper bound of  the
maximum index of non-zero terms of the expansion
of $f'_{(k)}(t'_1,\ldots,t'_k)$ in terms of the partial densities of the factorial moment measures.
For some  particular domains of   $(t'_1,\ldots,t'_k)$, a smaller
value of  $i_{\max}$ can be given; cf the proof 
of Proposition~\ref{p.Haenggi-residual} 
in the next section.

Using Proposition~\ref{p.f'k} and Corollary~\ref{newmu_n}
we  obtain the following, more explicit form of the joint distribution of the two strongest values of the
STINR process.
\begin{corollary}\label{c.f'k}
Under the assumptions of Theorem~\ref{mainResult}, we have 
\begin{align*}
&f'_{(k)}(z'_1,\dots,z'_k)\\
=& \Bigl( \prod\limits_{j=1}^{k}(\gamma z_j')^{-(2/\beta+1)} \Bigr)
 \sum_{i=0}^{\lceil 1/\gamma z'_k+k-1\rceil} \frac{(-1)^i}{i!}c_{{k+i},2/\beta,0}\,\bar{ \calI}_{k+i,\beta}((W/\gamma)a^{-\beta/2})
\gamma^{k+i}\\
&\times\int_{z'_k}^{1/\gamma} \hspace{-1em}\dots\int_{z'_k}^{1/\gamma} 
%\mu'^{(k+i)}(z'_1,\dots,z'_k, \zeta'_1,\dots,\zeta'_i) \,d\zeta'_i\dots  d\zeta'_1,
 \Bigl(\prod\limits_{j=k+1}^{k+i}(\gamma \zeta_j')^{-(2/\beta+1)} \Bigr)
\Bigl(1- \sum\limits_{j=1}^{k}(\gamma z'_j)-\sum\limits_{j=k+1}^{k+i}(\gamma \zeta'_j) 
\Bigr)^{2(k+i)/\beta-1}
\Ind\Bigl(\sum\limits_{j=1}^{k}(\gamma z'_j)-\sum\limits_{j=k+1}^{k+i}(\gamma \zeta'_j) 
<1\Bigr) \,d\zeta'_1\dots  d\zeta'_i,
\end{align*}
for $z'_{1}>z'_{2}>\dots >z'_{k}$ and 
$f'_{(k)}(z'_1,\dots,z'_2)=0$ otherwise.
\end{corollary}
As a simple example of the application of Proposition~\ref{p.f'k} or Corollary \ref{c.f'k},
note that for any $\tau>0$ and $\tau'=\tau/(1+\gamma \tau)$
\begin{equation}
\int_{\tau'}^{1/\gamma} \hspace{-1em}\dots\int_{\tau'}^{1/\gamma}f'_{(k)}(z'_1,\dots,z'_k)\,dz'_1,\dots dz'_k=
\calP^{(k)}(\tau)
\end{equation}
is equal to the $k$-coverage probability evaluated  in Corollary \ref{singletier}.
More developed examples of such applications are studied in
the subsequent sections. 
\begin{remark}\label{rem.Dickman} Another way for calculating the joint density of the order statistics
of the STIR process, using its relations to the two-parameter family of Poisson-Dirichlet processes
 and a generalized version of the  Dickman function,  is offered by the result~\cite[Theorem 5.4]{handa2009two}; cf~\cite{sinrPD}
for more details.
\end{remark}

\subsection{SINR with 
interference cancellation and signal combination}
\label{ss.ICSC}
We will now demonstrate how to use the order statistics of the  STINR
process to express SINR values accounting for general interference
cancellation (or management) and signal combination techniques. 

Consider a set of $k\ge 1$ strongest signals received by the typical user
$(Y^{(i)})^{-1}$  (recall that $Y^{(1)}<Y^{(2)}<\ldots$ are order
statistics of the propagation process $\{Y\}$).  Consider a subset
$\calU\subset[k]:=\{1,\ldots,k\}$ of this set.
We define SINR  under {\em 
Interference Cancellation and Signal Combination} (ICSC-SINR)
\begin{equation}\label{e.ICSC-SINR}
\SINR_{\calU,k}=
\frac{\sum_{i\in\calU} (Y^{(i)})^{-1}}%
{W+\gamma I-\gamma \sum_{j=1}^k (Y^{(j)})^{-1}}\,,
\end{equation}
where signals in $\calU$ are combined and interference from all
remaining interferers in $[k]\setminus \calU$ is canceled.
It is easy to see that the ICSC-SINR coverage events 
can be expressed using the order statistics of the STINR
process as follows
\begin{align}
&\{\,\SINR_{\calU,k}> \tau\,\}\nonumber\\
&=
\Bigl\{\,\sum_{i\in\calU} (Y^{(i)})^{-1}>
\tau W+\gamma\tau I-\gamma\tau \sum_{j=1}^k (Y^{(j)})^{-1}\,\Bigr\}\nonumber\\
&=\Bigl\{\,(1+\gamma\tau)\sum_{i\in\calU}Z'_{(i)}+
\gamma\tau \sum_{j\in[k]\setminus\calU}Z'_{(j)}
> \tau\,\Bigr\}\nonumber\\
&=
\Bigl\{\,\sum_{i\in\calU}
Z'_{(i)}+\gamma\tau'\!\!\!\!\sum_{j\in[k]\setminus\calU}Z'_{(j)}
> \tau'\,\Bigr\}\,,\label{e.ICSC-SINR-coverage}
\end{align}
where, $\tau'=\tau/(1+\gamma\tau)$.

Note that the ICSC-SINR coverage~(\ref{e.ICSC-SINR-coverage})
says that the ratio of the combined signals to noise plus
non-canceled interference is larger than the threshold $\tau$.
The probabilities  
% 2c format
%\begin{align}\calP^{(\calU,k)}(\tau)&:=\Prob\{\,\SINR_{\calU,k}>\tau\}\nonumber\\
%&=\Prob\Bigl\{\,\sum_{i\in\calU}
%Z'_{(i)}+\gamma\tau'\!\!\!\!\sum_{j\in[k]\setminus\calU}Z'_{(j)}
%> \tau'\,\Bigr\}\label{e.ICSC-SINR-cp}
%\end{align}
\begin{align}\calP^{(\calU,k)}(\tau):=\Prob\{\,\SINR_{\calU,k}>\tau\}
=\Prob\Bigl\{\,\sum_{i\in\calU}
Z'_{(i)}+\gamma\tau'\!\!\!\!\sum_{j\in[k]\setminus\calU}Z'_{(j)}
> \tau'\,\Bigr\}\label{e.ICSC-SINR-cp}
\end{align}
as function of $\tau$ give the distribution of the  ICSC-SINR,
thus allowing one to study coverage, spectral efficiency and other
characteristics of this channel with  ICSC.

The interference cancellation may be imperfect
and only reduce the interference power by a factor $\bar\gamma<\gamma$;
cf~\cite{weber2007transmission}.
This can be taken into account by replacing $\gamma$ by $\bar\gamma$ in~(\ref{e.ICSC-SINR-cp}).
Also~(\ref{e.ICSC-SINR-coverage}) may yet ignore 
(practical) conditions under which signal 
cancellation can be effectively performed. In order to take them
into account one can  modify~(\ref{e.ICSC-SINR-coverage}) as follows
\begin{align}
\calP^{(\calU,k)}_{\text{IC}}(\tau):=\Prob\Bigl\{\,\sum_{i\in\calU}
Z'_{(i)}+\Ind(\text{IC})\gamma\tau'\!\!\!\!\sum_{j\in[k]\setminus\calU}Z'_{(j)}
> \tau'\,\Bigr\}\,,\label{e.ICSC-SINR-cp-cond}
\end{align}
where %$\Ind(\text{SC})$ and  
$\Ind(\text{IC})$ denotes %, respectively,
the indicator of a suitable condition  for %signal combination (SC) and
the feasibility of the interference cancellation (IC). 
Natural conditions for
 IC can also be expressed in terms of the values $Z'_{(i)}$
$i=1,\ldots,k$, and thus $\calP^{(\calU,k)}_{\text{IC}}$ can be
calculated using Proposition~\ref{p.f'k}.
Similarly, conditions for the feasibility of the 
signal combination can be introduced.
Here are two examples of natural IC conditions.

\subsubsection{ Successive interference cancellation (SIC)} 
\label{sss.SIC}
consists in first decoding  the strongest interfering signal (among $k$ strongest),
subtracting it from the interference, then decoding the second strongest
one, and so on. Assume that one can  decode  a signal if its SINR
is larger than some threshold $\epsilon$ (called {\em IC threshold}). 
Denote the indexes of the
interfering signals to be decoded and subtracted by
$j_1,\ldots,j_{k'}$, where $k'=k-|\calU|$, 
starting from the strongest one.
SIC condition can be written then as the superposition of the following conditions.
\begin{align*}
Z'_{j_1}>&\epsilon'\\
Z'_{j_2}+\gamma \epsilon'Z'_{j_1}>&\epsilon'\\
\ldots&\\
Z'_{j_{k'}}+\gamma \epsilon'\sum_{l=1}^{k'-1}Z'_{j_l}>&\epsilon'\,,
\end{align*}
where $\epsilon'=\epsilon/(1+\gamma \epsilon)$.

\subsubsection{Independent interference cancellation (IIC)}
\label{sss.IIC}
 is a weaker  condition
assuming that all interfering signals to be canceled can be
decoded independently: $Z'_{(j)}>\epsilon'$ for all $j\in[k]\setminus\calU$.

\subsection{Signal to residual interference ratio of the
  $k\,$th strongest signal}
Following Zhang and Haenggi~\cite{zhang2012performance,zhangdecoding} 
let us consider coverage by the $k\,$th strongest signal with all
$k-1$ stronger signals canceled from the interference,
$\calP^{(\calU,k)}(\tau)$ with $\calU=\{k\}$. 
The following result follows from  Proposition~\ref{p.f'k}.
Its second statement  (case $\gamma\tau>1$) with $W=0$ 
was proved in~\cite{zhangdecoding}. 
\begin{proposition} \label{p.Haenggi-residual}
Under the assumptions of Theorem~\ref{mainResult}
we have 
% 2c format
%\begin{align}\nonumber
%&\calP^{(\{k\},k)}(\tau)=\\
%&=\sum_{i=0}^{i_{\max}}
%\frac{(-1)^i}{i!}
%\int_{0}^{1/\gamma}\hspace{-1em}
%\dots\int_{0}^{1/\gamma}
%\Ind\Bigl(
%\gamma\tau'\sum_{i=1}^{k-1}z'_i+z'_{k} >\tau'\Bigr)\nonumber\\
%&\hspace{2em}\times \Ind(z'_1>\ldots>z'_k)\,
%\mu_k^{(k+i)}(z'_1,\dots,z'_k)\,dz'_1\dots dz'_k\,,
%\label{e.residual}
%\end{align} 
\begin{align}
\calP^{(\{k\},k)}(\tau)
=\sum_{i=0}^{i_{\max}}
\frac{(-1)^i}{i!}
\int_{0}^{1/\gamma}\hspace{-1em}
\dots\int_{0}^{1/\gamma}
\Ind\Bigl(
\gamma\tau'\sum_{i=1}^{k-1}z'_i+z'_{k} >\tau'\Bigr)
\Ind(z'_1>\ldots>z'_k)\,
\mu_k^{(k+i)}(z'_1,\dots,z'_k)\,dz'_k\dots dz'_1\,,
\label{e.residual}
\end{align} 
where the upper summation limit  is bounded by
\begin{equation}\label{e.residual-imax}
i_{\max}<1/(\gamma\tau')-1
=1/(\gamma\tau)\,.
\end{equation}
For $\gamma\tau\ge 1$  we have $i_{\max}=0$ and 
% 2c format
%\begin{align}\nonumber
%&\calP^{(\{k\},k)}(\tau)\\
%&=\frac{\calI_{k,\beta}((W/\gamma)a^{-\beta/2})}{\calI_{k,\beta}(0)
%(\gamma\tau)^{2k/\beta}\Gamma(1+2k/\beta)(\Gamma(1-2/\beta))^k}\nonumber\\[2ex]
%&=\frac{2^{k-1}\calI_{k,\beta}((W/\gamma)a^{-\beta/2})}{\beta^{k-1}
%(\gamma\tau)^{2k/\beta}
%\Gamma(1+2k/\beta)(\Gamma(1+2/\beta))^k(\Gamma(1-2/\beta))^{2k}}\,.
%\label{e.Haenggi-residual}
%\end{align}
\begin{align}\nonumber
\calP^{(\{k\},k)}(\tau)&=\frac{\calI_{k,\beta}((W/\gamma)a^{-\beta/2})}{\calI_{k,\beta}(0)
(\gamma\tau)^{2k/\beta}\Gamma(1+2k/\beta)(\Gamma(1-2/\beta))^k}\nonumber\\[2ex]
&=\frac{2^{k-1}\calI_{k,\beta}((W/\gamma)a^{-\beta/2})}{\beta^{k-1}
(\gamma\tau)^{2k/\beta}
\Gamma(1+2k/\beta)(\Gamma(1+2/\beta))^k(\Gamma(1-2/\beta))^{2k}}\,.
\label{e.Haenggi-residual}
\end{align}
\end{proposition}
\begin{proof}
Using Proposition~\ref{p.f'k}, in order to 
calculate the probability
$\calP^{(\{k\},k)}(\tau)=%\Prob\{\,\SINR_{\{k\},k}>\tau\,\}=
\Prob\{\,\gamma\tau'\sum_{i=1}^{k-1}Z'_{(i)}+Z'_{(k)} >
\tau'\,\}$ one needs to integrate the density $f'_{(k)}$ over the domain
$\gamma\tau'\sum_{i=1}^{k-1}z'_i+z'_{k} >\tau'$.
%\{\gamma\tau'\sum_{i=1}^{k}z'_i+(1-\tau')z'_{k} >\tau'\}$.
A necessary condition for the partial density $\mu'^{(k+i)}_k$ present in the
expansion~(\ref{e.f'k}) to be non-null is  
$\sum_{i=1}^{k-1}z'_i+(i+1)z'_k<1/\gamma$ (cf. the proof of Proposition
~\ref{p.f'k}). Combining this condition with the integration domain
we obtain 
$(\tau'-z'_{k})/(\gamma\tau')+(i+1)z'_k<1/\gamma$,
which can be rewritten as
$z'_k(1+i-1/(\gamma\tau')<0$.
Since $z'_k\ge0$ we have $i<1/(\gamma\tau')-1$, which 
proves~(\ref{e.residual-imax}). The fact that $i_{\max}$ does not
depend on the integration variables allows one to interchange the summation
and integration as in~(\ref{e.residual}).

For $\gamma\tau\ge1$ we have  $i_{\max}=0$ and 
the probabilities $\calP^{(\{k\},k)}(\tau)$
can be evaluated as integrals of the $k\,$th factorial moment measure
of the STINR process
over $\gamma\tau'\sum_{i=1}^{k-1}z'_i+z'_{k}>\tau'$. This is another
justification of the equation in (b) in the proof of \cite[Theorem~1]{zhangdecoding} 
and the remaining part of the proof follows the lines presented there,
with the contribution of the non-zero noise resulting from the form of
the factorial moment measure $M'^{(k)}$.
\end{proof}
The related result \cite[Proposition~2]{zhangdecoding} brings also some approximations for 
$\calP^{(\{k\},k)}_{SIC}(\tau)$, that is the probability of the above
signal-to-residual-interference coverage of the $k\,$th strongest
signal with SIC condition (with IC
threshold $\epsilon=\tau$).  Exact value of this probability can be
obtained restricting the integration domain in~(\ref{e.residual})
as explained in Section~\ref{sss.SIC}.

In the next section we will consider a somewhat opposite problem,
namely, how the coverage by the strongest signal can be improved
by  removing  $k-1$ successive strongest signals  from the
interference or by combining them with the strongest signal.
 
\subsection{Improving the strongest signal by  
interference cancellation and cooperation}
In this section we consider
the following two scenarios:
\begin{itemize}
\item When the receiver,
being served by the strongest base station, 
is able to suppress the interference created by the subsequent 
  $k-1$ strongest stations.~\footnote{Cf
      e.g.~\cite{bar1989adaptive,xu2012co} for 
 methods allowing one to remove from the strong desired signal
a weaker interfering  one.}
\item When the subsequent $k-1$ strongest stations can combine their signals 
with the strongest one.
\end{itemize}
More specifically, for any $\epsilon,\tau $, with $0<\epsilon<\tau$, and
$\epsilon'=\epsilon/(1+\gamma\epsilon)$, $\tau'=\tau/(1+\gamma\tau)$, we define  the following
 coverage  probabilities 
% 2c format
%\begin{align*}
%&\calP_{IC}^{(k)}(\tau,\epsilon):=
%\Prob\Biggl\{\,
%\begin{aligned}
%\SINR_{\{1\},k}&> \tau &&\text{when\ } Z_{(k)}>\epsilon\\ 
%Z_{(1)}&> \tau &&\text{otherwise}
%\end{aligned}
%\,\Biggr\}\\[2ex]
%&=\Prob\Bigl\{\,Z'_{(1)}+\gamma\tau'\Bigl(
%Z'_{(2)}+\ldots+Z'_{(k)}\Bigr)\Ind(Z'_{(k)}>\epsilon')>\tau'\,\Bigr\}
%\end{align*}
\begin{align*}
\calP_{IC}^{(k)}(\tau,\epsilon):=
\Prob\Biggl\{\,
\begin{aligned}
\SINR_{\{1\},k}&> \tau &&\text{when\ } Z_{(k)}>\epsilon\\ 
Z_{(1)}&> \tau &&\text{otherwise}
\end{aligned}
\,\Biggr\}
=\Prob\Bigl\{\,Z'_{(1)}+\gamma\tau'\Bigl(
Z'_{(2)}+\ldots+Z'_{(k)}\Bigr)\Ind(Z'_{(k)}>\epsilon')>\tau'\,\Bigr\}
\end{align*}
and
% 2c format
%\begin{align*}
%&\calP_{SC}^{(k)}(\tau,\epsilon):=
%\Prob\Biggl\{\,
%\begin{aligned}
%\SINR_{[k],k}&> \tau &&\text{when\ } Z_{(k)}>\epsilon\\ 
%Z_{(1)}&> \tau &&\text{otherwise}
%\end{aligned}
%\,\Biggr\}\\[2ex]
%&=\Prob\Bigl\{\,Z'_{(1)}+\Bigl(Z'_{(2)}+\ldots+Z'_{(k)}\Bigr)\Ind(Z'_{(k)}>\epsilon')>\tau'\,\Bigr\}\,.
%\end{align*}
\begin{align*}
\calP_{SC}^{(k)}(\tau,\epsilon):=
\Prob\Biggl\{\,
\begin{aligned}
\SINR_{[k],k}&> \tau &&\text{when\ } Z_{(k)}>\epsilon\\ 
Z_{(1)}&> \tau &&\text{otherwise}
\end{aligned}
\,\Biggr\}
=\Prob\Bigl\{\,Z'_{(1)}+\Bigl(Z'_{(2)}+\ldots+Z'_{(k)}\Bigr)\Ind(Z'_{(k)}>\epsilon')>\tau'\,\Bigr\}\,.
\end{align*}
Note that $\calP_{IC}^{(k)}(\tau,\epsilon)$  is the probability of the
coverage by the strongest signal with the {\em cancellation of the
interference} coming from  subsequent $k-1$ strongest signals,
whenever these signals individually can be decoded at the SINR level
$\epsilon$ and no interference cancellation otherwise.
Similarly $\calP_{SC}^{(k)}(\tau,\epsilon)$ is the probability of the
coverage by the strongest signal  {\em combined} with  $k-1$ subsequent
strongest signals, whenever these signals individually can be decoded
at the SINR level
$\epsilon$ and no signal combination otherwise.
Remember that the coverage probability   evaluated in
Corollary~\ref{singletier} corresponds to
$\calP(\tau)=\calP^{(k)}(\tau)=\Prob\{\,Z'_{(1)}>\tau'\,\}$. It is immediately  
seen that $$\calP_{SC}^{(k)}(\tau,\epsilon)\ge
\calP_{IC}^{(k)}(\tau,\epsilon)\ge \calP^{(k)}(\tau)\,$$
for all $0<\epsilon\le \tau$.
In order to study the difference between the coverage probabilities,  denote
% 2c format
%\begin{align}
%\Delta_{IC}^{(k)}&(\tau,\epsilon)\label{e.DIC}\\
%&:=\Prob\Bigl\{\,Z'_{(1)}+\gamma\tau'\Bigl(Z'_{(2)}+\ldots+
%Z'_{(k)}\Bigr)>\tau'\nonumber\\
%&\hspace{4em}\text{and\ } \epsilon'<Z'_{(i)}<\tau', i=1\ldots k\,\Bigr\} \nonumber\\[2ex]
%\Delta_{SC}^{(k)}&(\tau,\epsilon) \label{e.DSC}\\
%&:=\Prob\Bigl\{\,Z'_{(1)}+\ldots+ Z'_{(k)}>\tau'\nonumber\\
%&\hspace{4em}\text{and\ } \epsilon'<Z'_{(i)}<\tau', i=1\ldots k\,\Bigr\}\,. \nonumber
%\end{align}
\begin{align}
\Delta_{IC}^{(k)}&(\tau,\epsilon)\label{e.DIC}
:=\Prob\Bigl\{\,Z'_{(1)}+\gamma\tau'\Bigl(Z'_{(2)}+\ldots+
Z'_{(k)}\Bigr)>\tau'
\text{and\ } \epsilon'<Z'_{(i)}<\tau', i=1\ldots k\,\Bigr\} \\[2ex]
\Delta_{SC}^{(k)}&(\tau,\epsilon) \label{e.DSC}
:=\Prob\Bigl\{\,Z'_{(1)}+\ldots+ Z'_{(k)}>\tau'
\text{and\ } \epsilon'<Z'_{(i)}<\tau', i=1\ldots k\,\Bigr\}\,. 
\end{align}
and note that 
\begin{align*}
\calP_{IC}^{(k)}(\tau,\epsilon)&=\calP(\tau)+\Delta_{IC}^{(k)}(\tau,\epsilon)\\[1.5ex]
\calP_{SC}^{(k)}(\tau,\epsilon)&=\calP(\tau)+\Delta_{SC}^{(k)}(\tau,\epsilon)\,.
\end{align*}

We have the following result regarding the increase 
of the coverage probability of the strongest signal induced by the 
interference cancellation or signal combination with the subsequent
$k-1$ strongest stations.
\begin{proposition}\label{c.Delta}
Under the assumptions of Theorem~\ref{mainResult},
% 2c format
%\begin{align}\label{e.DIC-res}
%&\Delta_{IC}^{(k)}(\tau,\epsilon)=\\
%&\sum_{i=0}^{i_{\max}}
%\frac{(-1)^i}{i!}\int_{\epsilon'}^{\tau'}\hspace{-1em}
%\dots\int_{\epsilon'}^{\tau'}
%\Ind\Bigl(z'_1+\gamma\tau'(z'_2+\ldots+
%z'_k)>\tau'\Bigr)\nonumber\\
%&\hspace{2em}\times \Ind(z'_1>\ldots>z'_k)\,
%\mu_k^{(k+i)}(z'_1,\dots,z'_k)\,dz'_1\dots dz'_k\nonumber
%\end{align} 
\begin{align}\label{e.DIC-res}
\Delta_{IC}^{(k)}(\tau,\epsilon)=
\sum_{i=0}^{i_{\max}}
\frac{(-1)^i}{i!}\int_{\epsilon'}^{\tau'}
\dots\int_{\epsilon'}^{\tau'}
\Ind\Bigl(z'_1+\gamma\tau'(z'_2+\ldots+
z'_k)>\tau'\Bigr)
\Ind(z'_1>\ldots>z'_k)\,
\mu_k'^{(k+i)}(z'_1,\dots,z'_k)\,dz'_k\dots dz'_1
\end{align} 
and 
% 2c format
%\begin{align}\label{e.DSC-res}
%&\Delta_{SC}^{(k)}(\tau,\epsilon)=\\
%&\sum_{i=0}^{i_{\max}}
%\frac{(-1)^i}{i!}\int_{\epsilon'}^{\tau'}\hspace{-1em}
%\dots\int_{\epsilon'}^{\tau'}
%\Ind\Bigl(z'_1+\ldots+z'_k)>\tau'\Bigr)\nonumber\\
%&\hspace{2em}\times \Ind(z'_1>\ldots>z'_k)\,
%\mu_k^{(k+i)}(z'_1,\dots,z'_k)\,dz'_1\dots dz'_k\nonumber
%\end{align} 
\begin{align}\label{e.DSC-res}
\Delta_{SC}^{(k)}(\tau,\epsilon)=
\sum_{i=0}^{i_{\max}}
\frac{(-1)^i}{i!}\int_{\epsilon'}^{\tau'}
\dots\int_{\epsilon'}^{\tau'}
\Ind\Bigl(z'_1+\ldots+z'_k)>\tau'\Bigr)
\Ind(z'_1>\ldots>z'_k)\,
\mu_k'^{(k+i)}(z'_1,\dots,z'_k)\,dz'_k\dots dz'_1
\end{align} 
where the upper summation limit  is bounded  by
\begin{equation}
i_{\max}<\frac{1}{\gamma\epsilon'}-k.
\end{equation} 
\end{proposition}

\begin{proof}
The result follows from Proposition~\ref{p.f'k}.
The fact that the integration domain is bounded away from $0$ by
$\epsilon'$ and~(\ref{e.f'imax}) implies this particular bound on
$i_{\max}$. Since it does not 
depend on the integration variables allows one to interchange the summation
and integration as in~(\ref{e.DIC-res}) and~(\ref{e.DSC-res}).
\end{proof}
Note that the parameter $\epsilon$ mathematically allows the expansions of  $\Delta_{IC}^{(k)}$ and $\Delta_{SC}^{(k)}$ to be  finite ones for  $i_{\max}<\infty$. 
We will further numerically   study this problem in Section~\ref{ss.Num-ICSC}.

\section{Numerical results}

\subsection{Setting}
\subsubsection{Model assumptions}
We illustrate our mathematical framework and methods by
calculating results for several cellular network models. For all
results, we have set $\gamma=1$, $PS=1$ and $K=1$ and 
assumed an interference limited network ($W=0$)~\footnote{Since
 including a noise term often makes little difference except for very
 low network density values} hence the results are scale-invariant
(do not depend on $\lambda$). Exceptions are results for muti-tier
networks, where we will be more specific about $P$ and $\lambda$.

To illustrate the impact of the strength of the  path-loss we  
consider two path-loss exponent values $\beta=3$ and~$5$.
The network simulations (validating our analytic results) were done in
a circular region of radius $10$ length units, which was empirically
found to be sufficiently large to render so-called edge effects negligible,
with number of network simulations being around $10^5$.

The SINR thresholds $\tau,\epsilon$ are expressed in dB; i.e. $\tau(\text{dB})
=10\log_{10}(\tau)\text{dB}$.

\subsubsection{Numerical integration}
We employ MATLAB
  implementation~\cite{paul_matlab_moments} of the quasi-Monte Carlo integration method (based on Sobol
points) for integrating the multi-dimensional integrals, which is
supported by the theory that says this type of numerical integration
generally performs better than regular Monte Carlo
integration~\cite{kuo2005lifting,dick2013high}. From our empirical
findings, the evaluation of the $\calJ_{n,\beta}$ integral is achieved
quickly on a standard machine in a matter of seconds (the number of
sample points mostly ranges from $10^3$ to $10^4$ points). The nature
of the integral such as (\ref{e.DIC-res}) and (\ref{e.DSC-res}) allows the integration with respect to the $z_i'$ and $v_i$ (from $\calJ_{n,\beta}$) to be done in the same computation step, which takes slightly longer to evaluate than $\calJ_{n,\beta}$. Although, Sobol points perform well in general, it should be noted that the type of quasi-random points can be better chosen if a more thorough analysis of the $\calJ_{n,\beta}$ integral kernel is performed, but this beyond the scope of this study.

\subsection{Coverage probability}

\label{ss.Num-Coverage}
%%% 1,2,3 coverage probabilities 
\subsubsection{Single-tier network}
Figures~\ref{Cov3Beta3} and~\ref{Cov3Beta5} present the $k$-coverage
probability $\calP^{(k)}(\tau)$  for  $k=1,2,3$ with $\beta=3$ and $5$,
respectively. Recall that  $\calP^{(k)}(\tau)$ is the tail distribution
function of the SINR related to  the  $k\,$th strongest signal.

\begin{figure}[t!]
\begin{center}
\begin{minipage}{0.45\linewidth}
\begin{center}
\centerline{\includegraphics[scale=0.5]{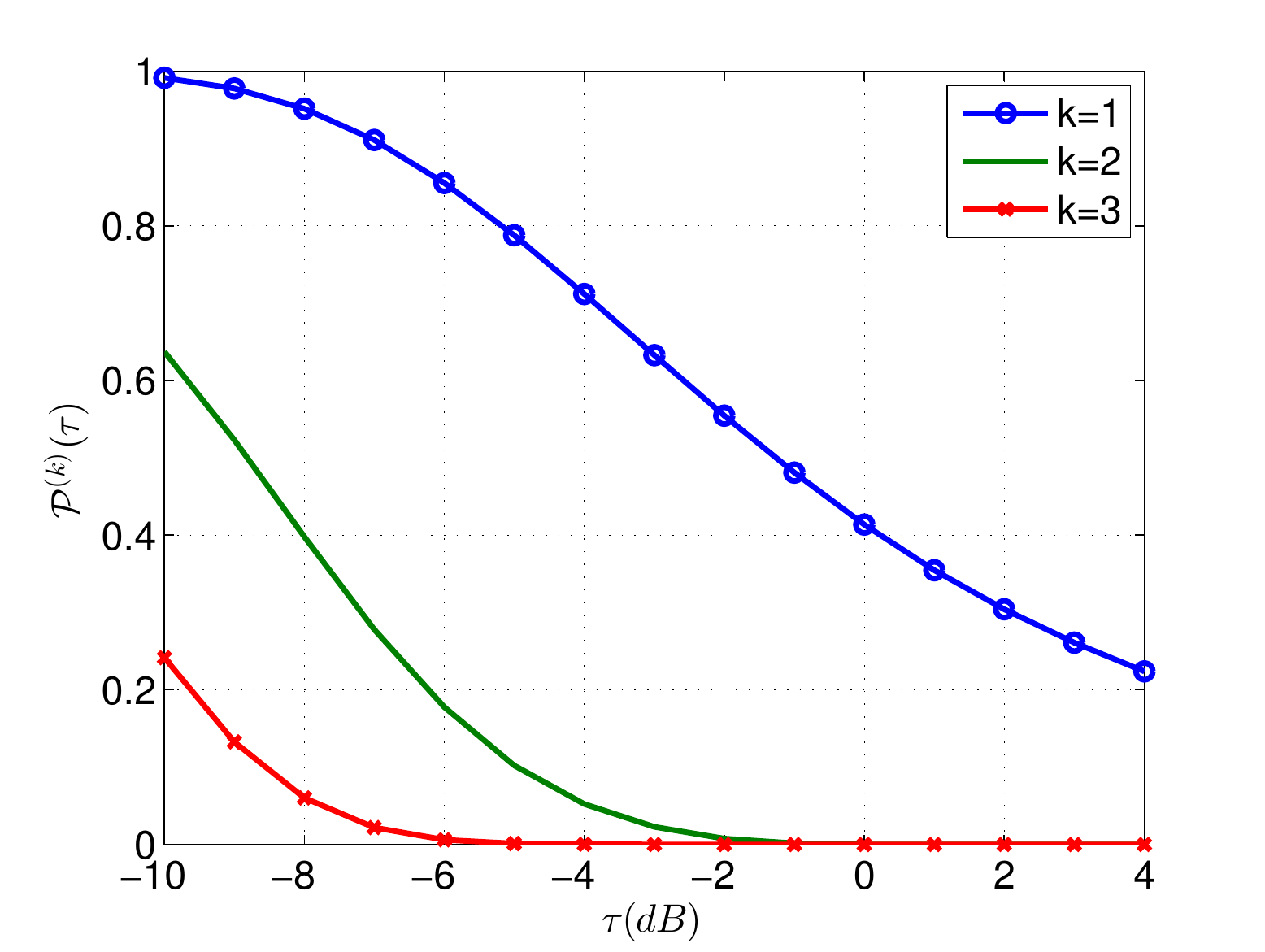}}
\vspace{-2ex}
\caption{For $\beta=3$, $k$-coverage probability $\calP^{(k)}(\tau)$
for a single-tier network.
\label{Cov3Beta3}}
\end{center}
\end{minipage}
\hfill
\begin{minipage}{0.45\linewidth}
\begin{center}
\centerline{\includegraphics[scale=0.5]{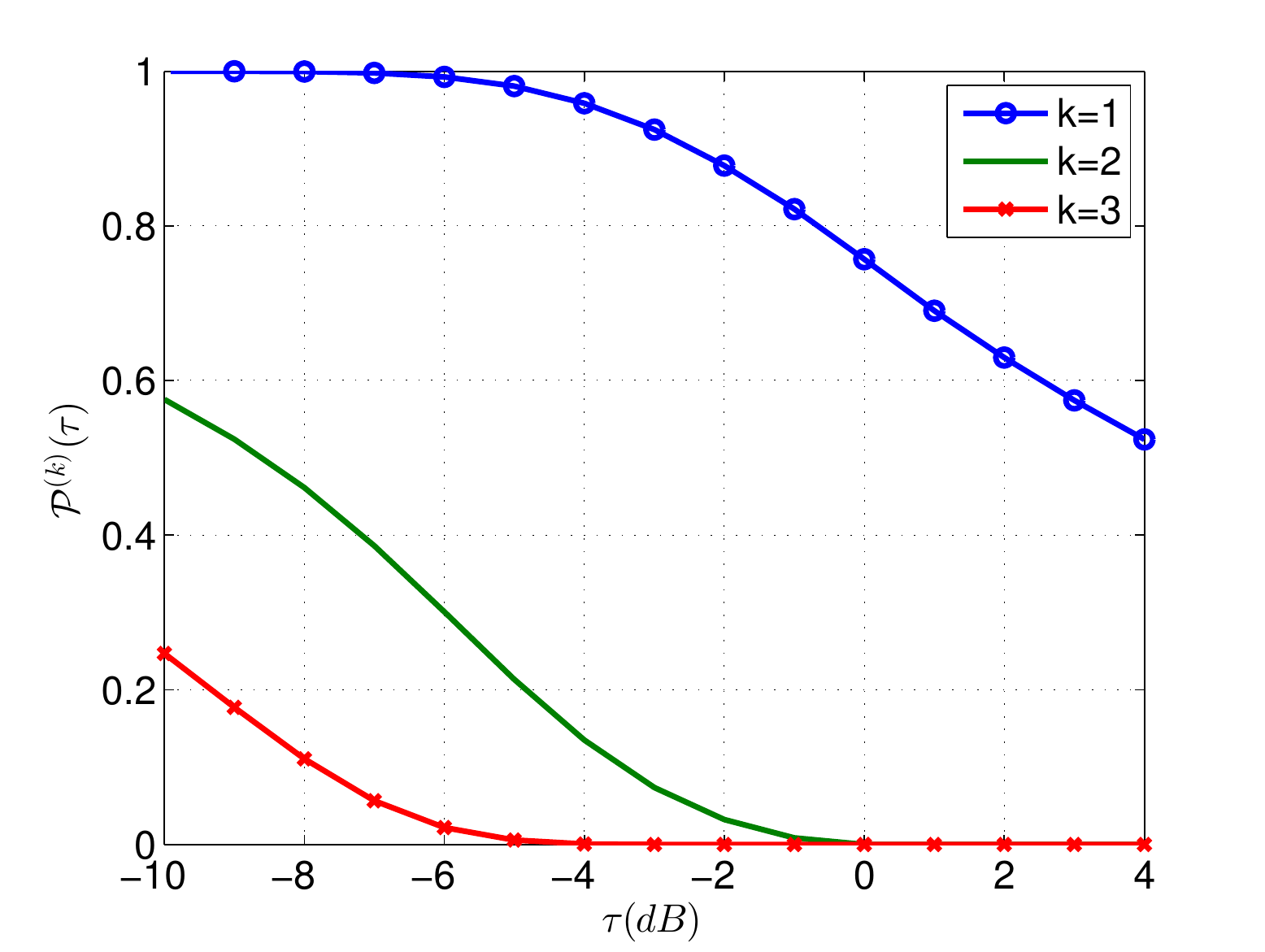}}
\vspace{-2ex}
\caption{The same functions as on Figure~\ref{Cov3Beta3} for
  $\beta=5$.
\label{Cov3Beta5}}
\end{center}
\end{minipage}
\end{center}
\vspace{-2ex}
\end{figure}

\begin{figure}[t!]
\begin{center}
\begin{minipage}[t]{0.45\linewidth}
\begin{center}
\centerline{\includegraphics[scale=0.5]{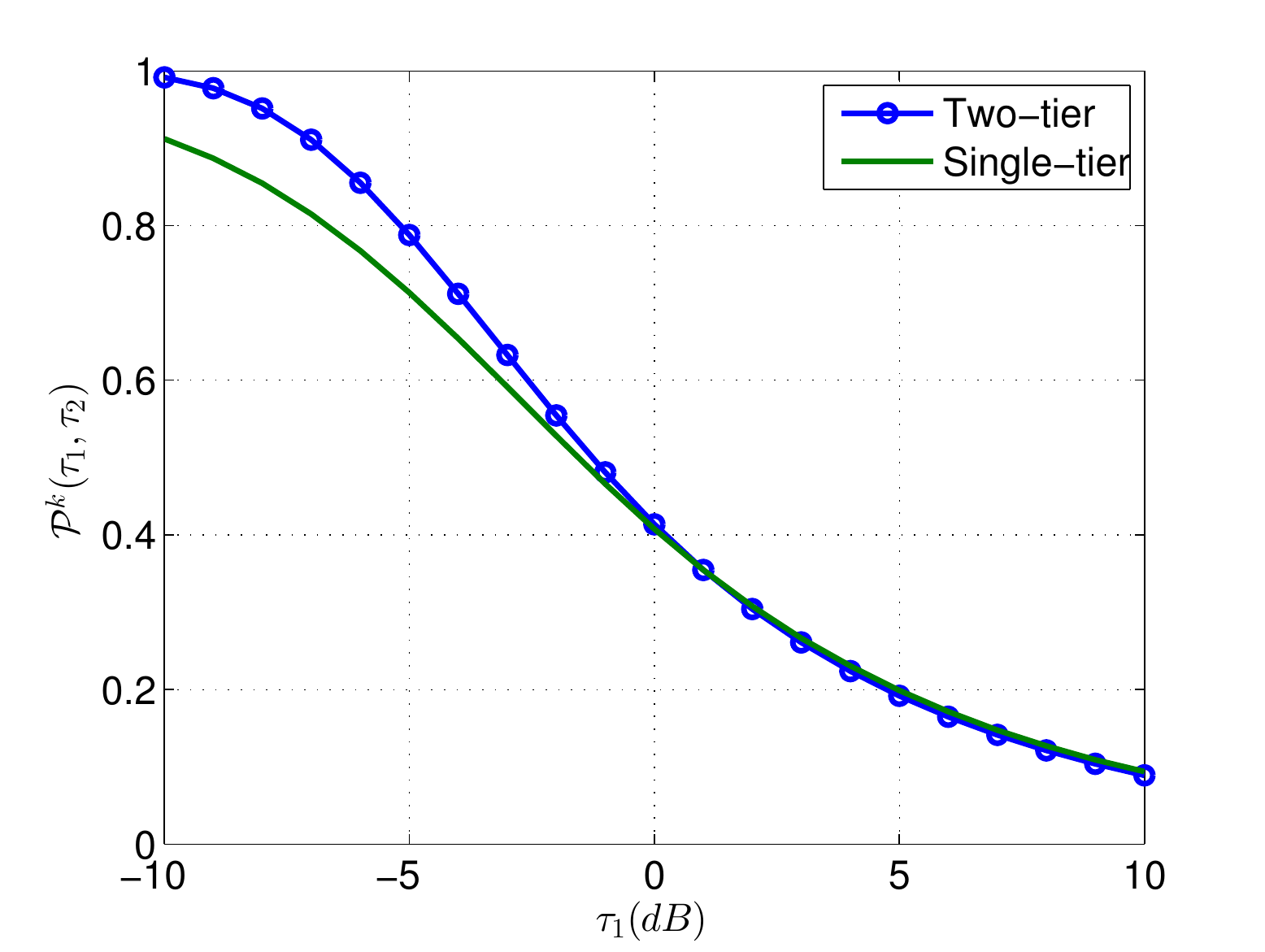}}
\vspace{-2ex}
\caption{For  $\beta=3$ and for a two-tier network with  
$\lambda_1=\lambda_2/2$, $P_1=100 P_2$ $\tau_2=1\text{dB}$, 
the $1$-coverage  probability $\calP^{(1)}$  as a function of
$\tau_1$ compared to a single-tier network model with equivalent propagation process
and constant SINR thresholds.
% ($\lambda=(\lambda_1+\lambda_2)/2$,  $P=(\lambda_1P_1+\lambda_2P_2)/\lambda$).
\label{MTvsSTProbC1}}
\end{center}
\end{minipage}
\hfill
\begin{minipage}[t]{0.45\linewidth}
\begin{center}
%\vspace{-2ex}
%\end{figure}
%\begin{figure}[t!]
\centerline{\includegraphics[scale=0.5]{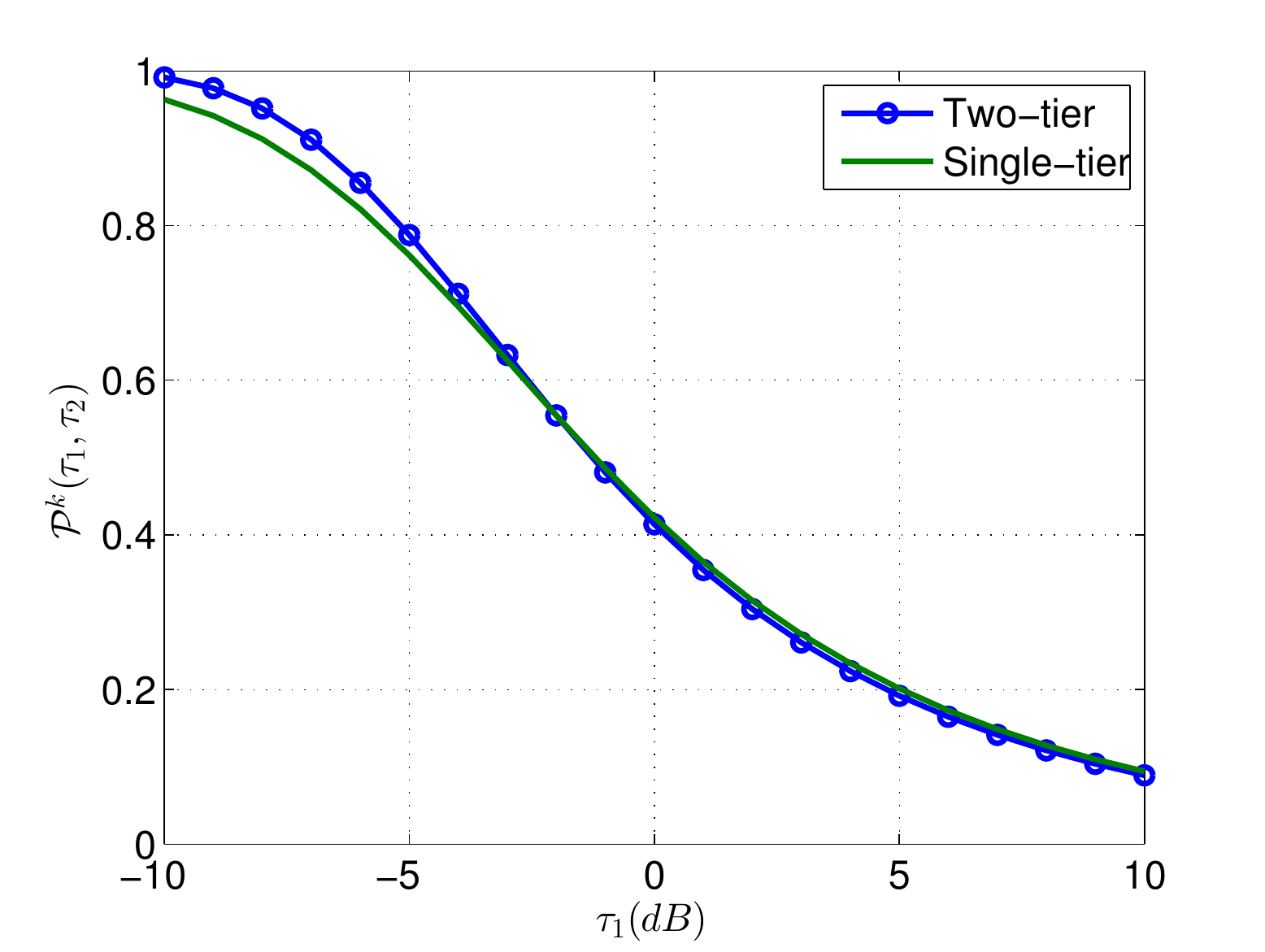}}
\vspace{-2ex}
\caption{1-coverage probability $\calP^{(1)}$ for two-tier and single-tier network as on Figure~\ref{MTvsSTProbC1} with   $\tau_2=-2$dB. 
\label{MTvsSTProbC2}}
\end{center}
\end{minipage}
\end{center}
\vspace{-2ex}
\end{figure}

\begin{figure}[t!]
\begin{center}
\begin{minipage}[t]{0.45\linewidth}
\begin{center}
\centerline{\includegraphics[scale=0.5]{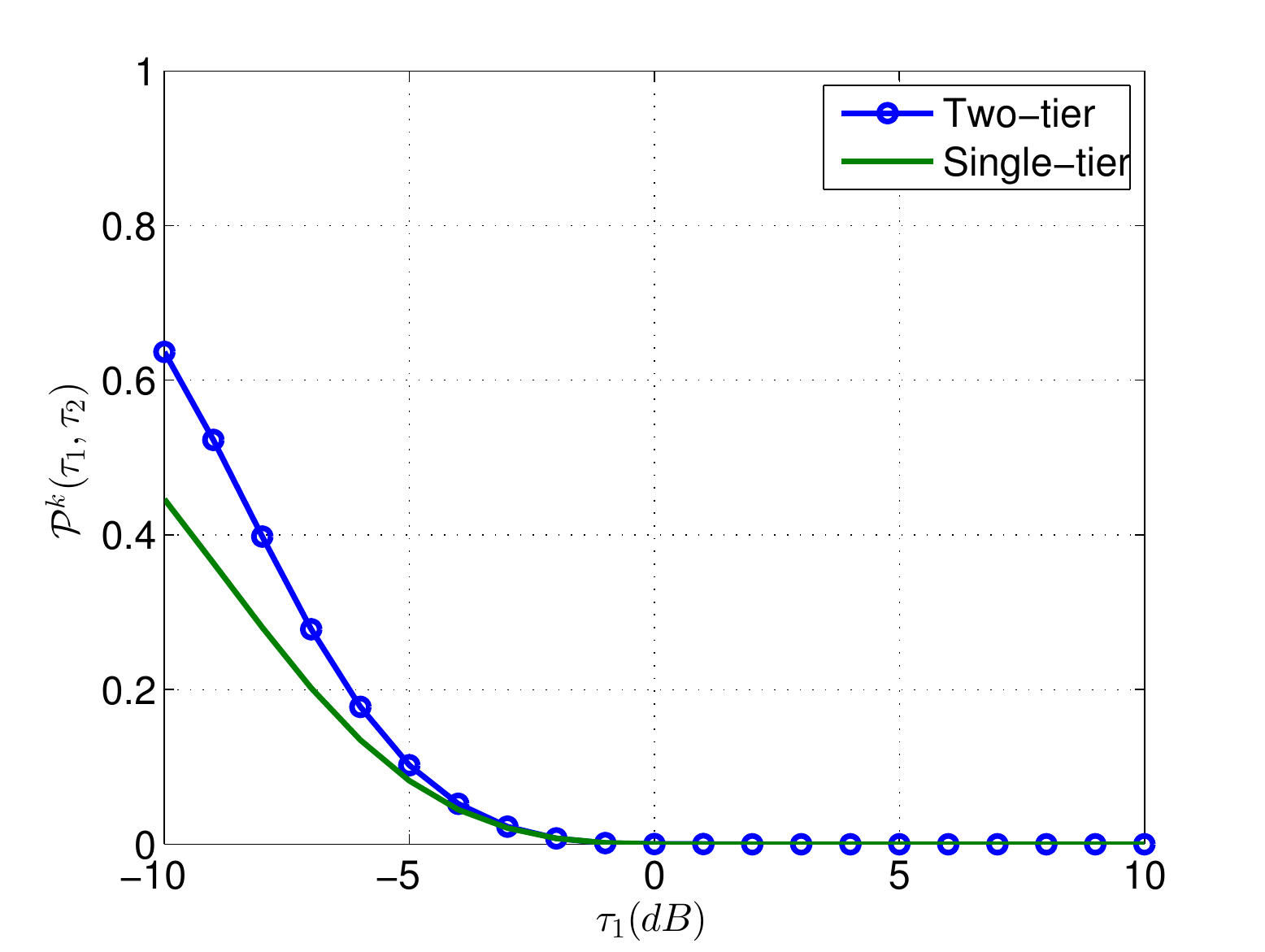}}
\vspace{-2ex}
\caption{2-coverage probability $\calP^{(2)}$  
for two-tier and one tier-network as on Figure~\ref{MTvsSTProbC1} (with   $\tau_2=-2$dB).
\label{MTvsSTProbC3}}
\end{center}
\end{minipage}
\end{center}
\vspace{-4ex}
\end{figure}

\subsubsection{Two-tier network}
We consider two-tier network (cf Section~\ref{s.multitier})
with  $\lambda_1=\lambda_2/2$, $P_1=100 P_2$ and  $\beta=3$.
Considering two values for $\tau_2$, namely $\tau_2=1(\text{dB})$ and 
$\tau_2=-2(\text{dB})$, on Figures~\ref{MTvsSTProbC1},
\ref{MTvsSTProbC2}, \ref{MTvsSTProbC3} we plot 1- and 2-coverage probability
$\calP^{(k)}$, $k=1,2$ as functions of $\tau_1$.
The parameters of the two-tier network are based on previous results~\cite[Fig.7.]{DHILLON2012}.
For $1$-coverage probability, the results agree with those in~\cite{DHILLON2012}, thus validating the generalized symmetric sum method in the multi-tier setting.
%For $K=1$, $\beta=3$ and $\tau_2=1 $ dB, the $1$-coverage probability for a two-tier network ($\lambda_1=\lambda_2/2$, $P_1=100 P_2$) and a single%-tier network with (spatially averaged)  parameters ($\lambda=(\lambda_1+\lambda_2)/2$,  $P=(\lambda_1P_1+\lambda_2P_2)/\lambda$).

For comparison purposes, on Figures~\ref{MTvsSTProbC1},
\ref{MTvsSTProbC2}, \ref{MTvsSTProbC3} we consider  also 
a single-tier network model with intensity $\lambda^*$ given by~(\ref{e.lambda*})
i.e., $\lambda^*=\lambda_1P_1^{2/\beta}+\lambda_2P_2^{2/\beta}$ (which gives equivalent propagation and  SINR processes) for which we
calculate the similar $k$-coverage probabilities with {\em constant} SINR thresholds 
$\tau=\E[T^*]=(\tau_1\lambda_1P_1^{2/\beta}+\tau_2\lambda_2P_2^{2/\beta})/\lambda^*$,
equal to the mean of the random SINR threshold mark in the truly equivalent model. 
Note that for the considered parameter values, 
this model offers quite reasonable approximation of the
original two-tier one, in particular for $\tau>1$ (although visually the corresponding plots almost coincide, it remains
an approximation). 
It would be interesting to see under which
  conditions the above observation can be generalized, i.e.,
 when one can approximate a multi-tier model with the
  single-tier model generating the equivalent propagation process and 
having  averaged $\tau$ vales.

\subsection{Coverage with  interference cancellation and cooperation}
\label{ss.Num-ICSC}
%produced with IntDeltaCoopSIC01 and SimDeltaCoopSIC01
\begin{figure}[t!]
\begin{center}
\begin{minipage}[t]{0.45\linewidth}
\begin{center}
\centerline{\includegraphics[scale=0.5]{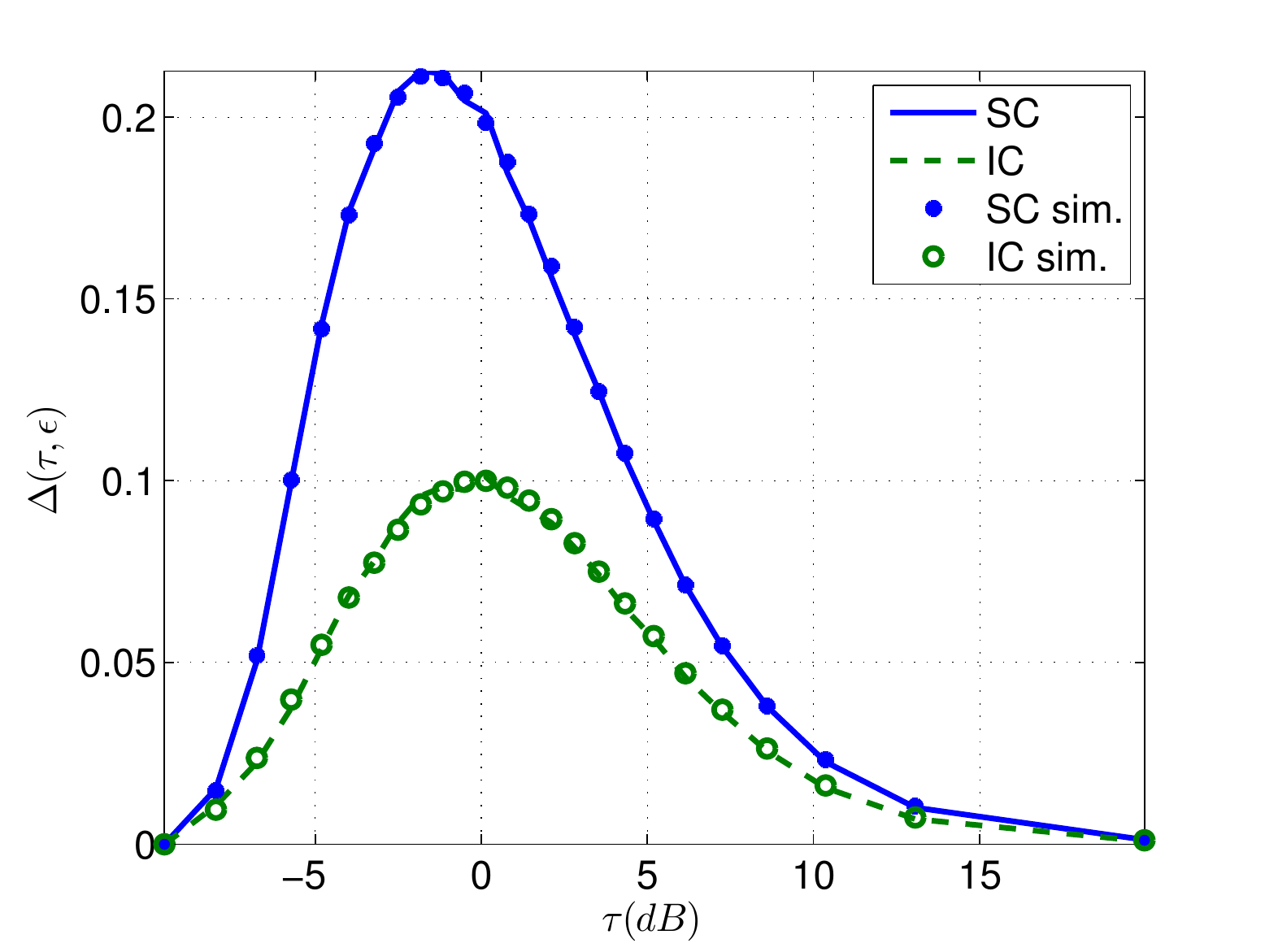}}
\caption{For $\beta=3$, the increase of the coverage probability of the strongest signal induced by the 
interference cancellation $\Delta^{(2)}_{IC}(\tau,\epsilon)$ or signal combination  $\Delta^{(2)}_{SC}(\tau,\epsilon)$ with the second 
strongest station assuming  $\epsilon'=0.1$ (or $\epsilon=-9.5424$dB).
\label{PlotDelta301}}
\end{center}
\end{minipage}
\hfill
%\end{figure}
%
%\begin{figure}[t!]
\begin{minipage}[t]{0.45\linewidth}
\begin{center}
\centerline{\includegraphics[scale=0.5]{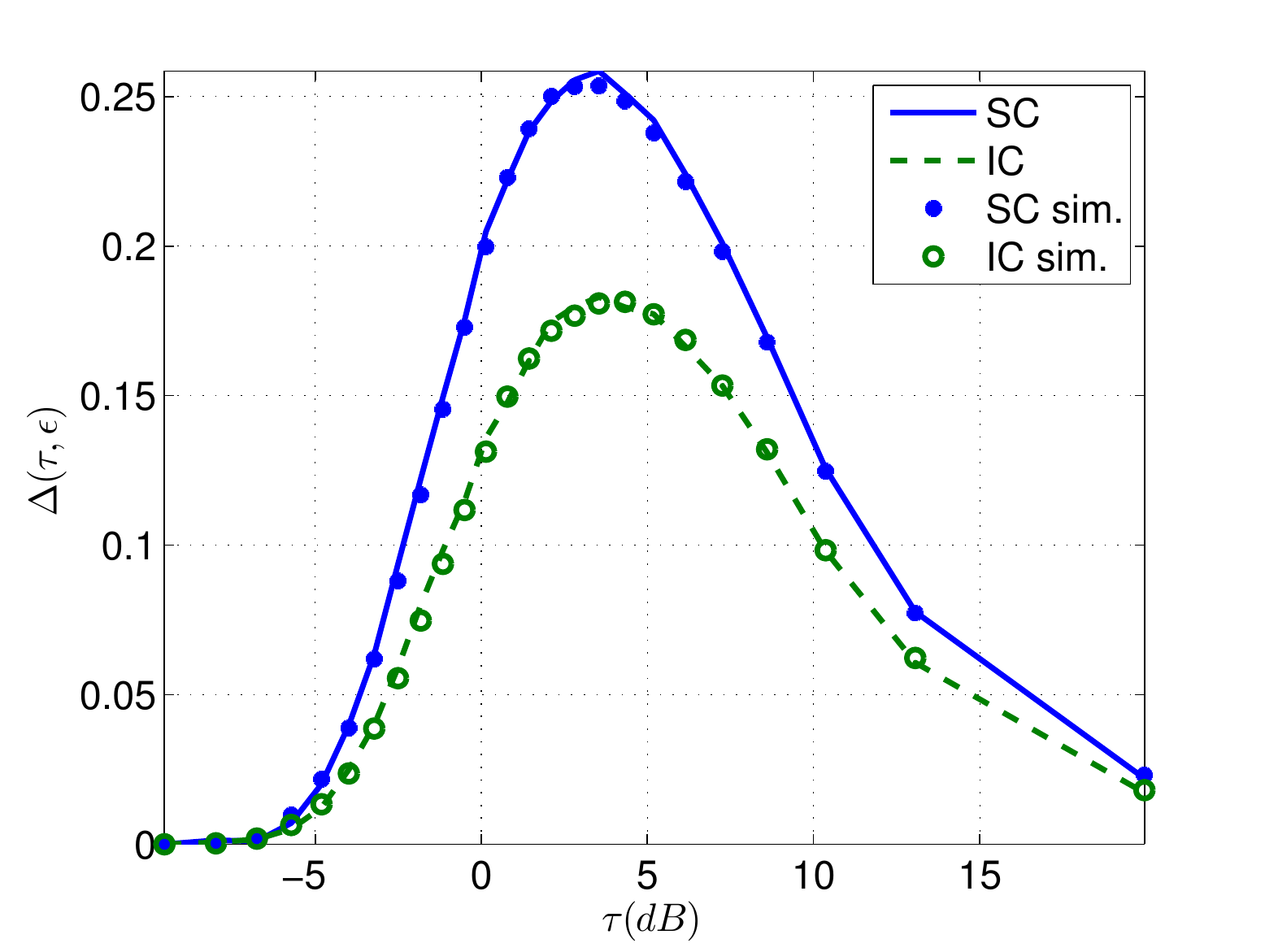}}
\caption{The same functions as on Figure~\ref{PlotDelta301} for $\beta=5$.% and  $\epsilon'=0.1$ (or $\epsilon=-9.5424$dB).
% the $\Delta$ probability for a network with 
\label{PlotDelta501}}
\end{center}
\end{minipage}\\
\begin{minipage}[t]{0.45\linewidth}
\begin{center}
\centerline{\includegraphics[scale=0.5]{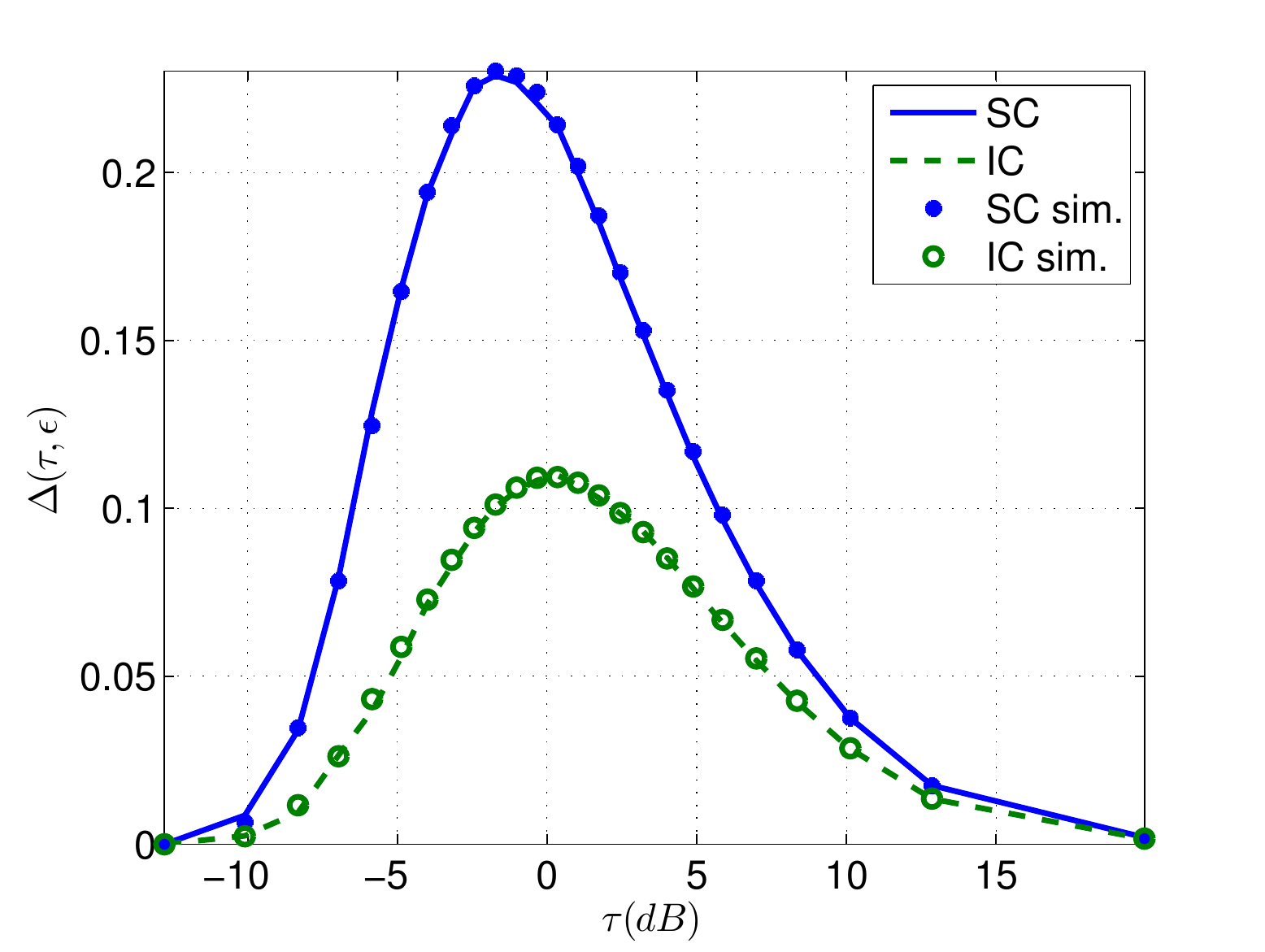}}
\caption{The same functions as on Figure~\ref{PlotDelta301} with $\beta=3$ and  $\epsilon'=0.05$ (or $\epsilon=-12.7875$dB).
\label{PlotDelta3005}}
\end{center}
\end{minipage}
\hfill
\begin{minipage}[t]{0.45\linewidth}
\begin{center}
\centerline{\includegraphics[scale=0.5]{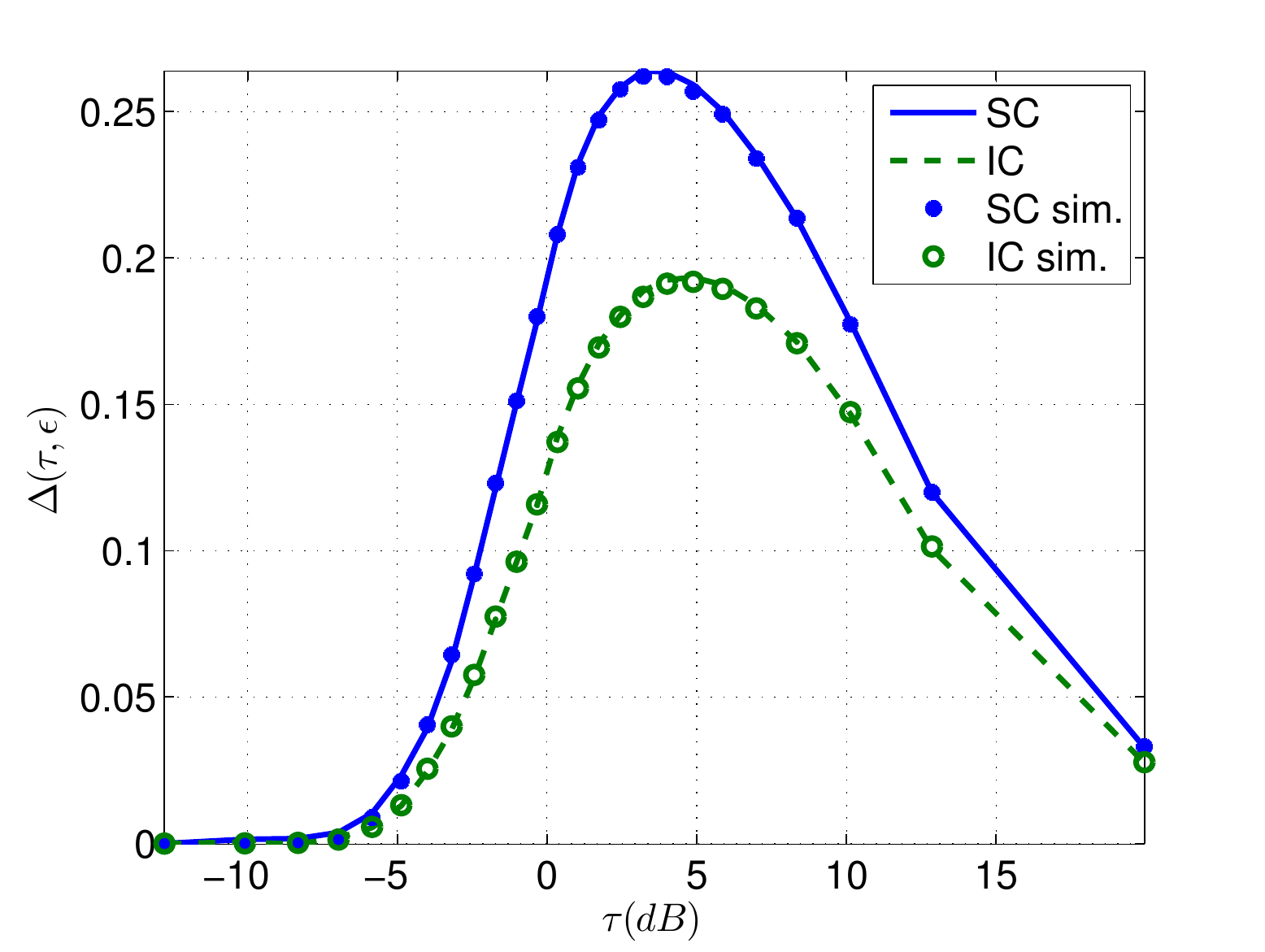}}
\caption{The same functions as on Figure~\ref{PlotDelta301} with $\beta=5$ and $\epsilon'=0.05$ (or $\epsilon=-12.7875$dB).
\label{PlotDelta5005}}
\end{center}
\end{minipage}
\end{center}
\vspace{-4ex}
\end{figure}

On Figures~\ref{PlotDelta301}, \ref{PlotDelta501}, \ref{PlotDelta3005} and~\ref{PlotDelta5005} we study the increase 
of the coverage probability of the strongest signal induced by the 
interference cancellation $\Delta^{(2)}_{IC}(\tau,\epsilon)$ or signal combination  $\Delta^{(2)}_{SC}(\tau,\epsilon)$ with the second 
strongest station. We consider two values of the path-loss exponent $\beta=3$ and $\beta=5$ 
and the decoding threshold (for cancellation and combination) 
$\epsilon=-9.5424$dB and $-12.7875$dB.

%\subsection{Effects of $\beta$ and $\epsilon'$}
The plots of the probabilities reveal that the increase of coverage probability is largest around the threshold $\tau=1$ (0dB). 
Also, as the path-loss exponent $\beta$ increases the difference  between  the gain  under  IC and SC 
 decreases (cf. Figure~\ref{PlotDelta301} to \ref{PlotDelta5005}). This is explained by realizing that the second largest interfering signal, which is removed under IC, weakens as $\beta$ increases.

Recall that the decoding threshold $\epsilon$ mathematically allows
the expansion of  $\Delta_{IC}^{(k)}$ and $\Delta_{SC}^{(k)}$
in~(\ref{e.DIC-res})  and~(\ref{e.DSC-res})
to be finite ones,
hence they are complete expressions when some $i_{\max}<\infty$ terms are included. Although as $\epsilon$ decreases more terms are needed,  
not all the integral terms are necessary to gain a good approximation (Figure~\ref{PlotDelta3terms} and \ref{PlotDelta5terms}). Numerical evidence confirms the mathematical intuition that as $\epsilon$ decreases the integral kernel approaches  the singularity at $t_i=0$ (for some integer $i\in[n]$), hence the peak of the resulting integral increases.  

Results for $\epsilon'=0.05$ (i.e.,  $\epsilon=-12.7875$dB) were obtained with relative ease and speed, but for smaller  
$\epsilon'$, it becomes computationally impracticable.
% (for example, $\epsilon'=0.01$ implies $100$ high-dimensional integrals) as well as the numerical integration not performing well in such higher dimensions. 

Broadly speaking, the more integration points are needed for larger $\beta$ and smaller $\epsilon$. The effect of $\beta$ probably stems from the Sobol points not sampling the integral kernel adequately, which could be improved by way of a suitable importance sampling. That said, for $\beta=5$ and $\epsilon'=0.05$ the $\Delta$ probability plots (Figure~\ref{PlotDelta5005}) were still obtained fast (in seconds) except  for small $\tau$ which required slightly more time.

%% approximations sections
\begin{figure}[t!]
\begin{center}
\begin{minipage}[t]{0.45\linewidth}
\begin{center}
\centerline{\includegraphics[scale=0.5]{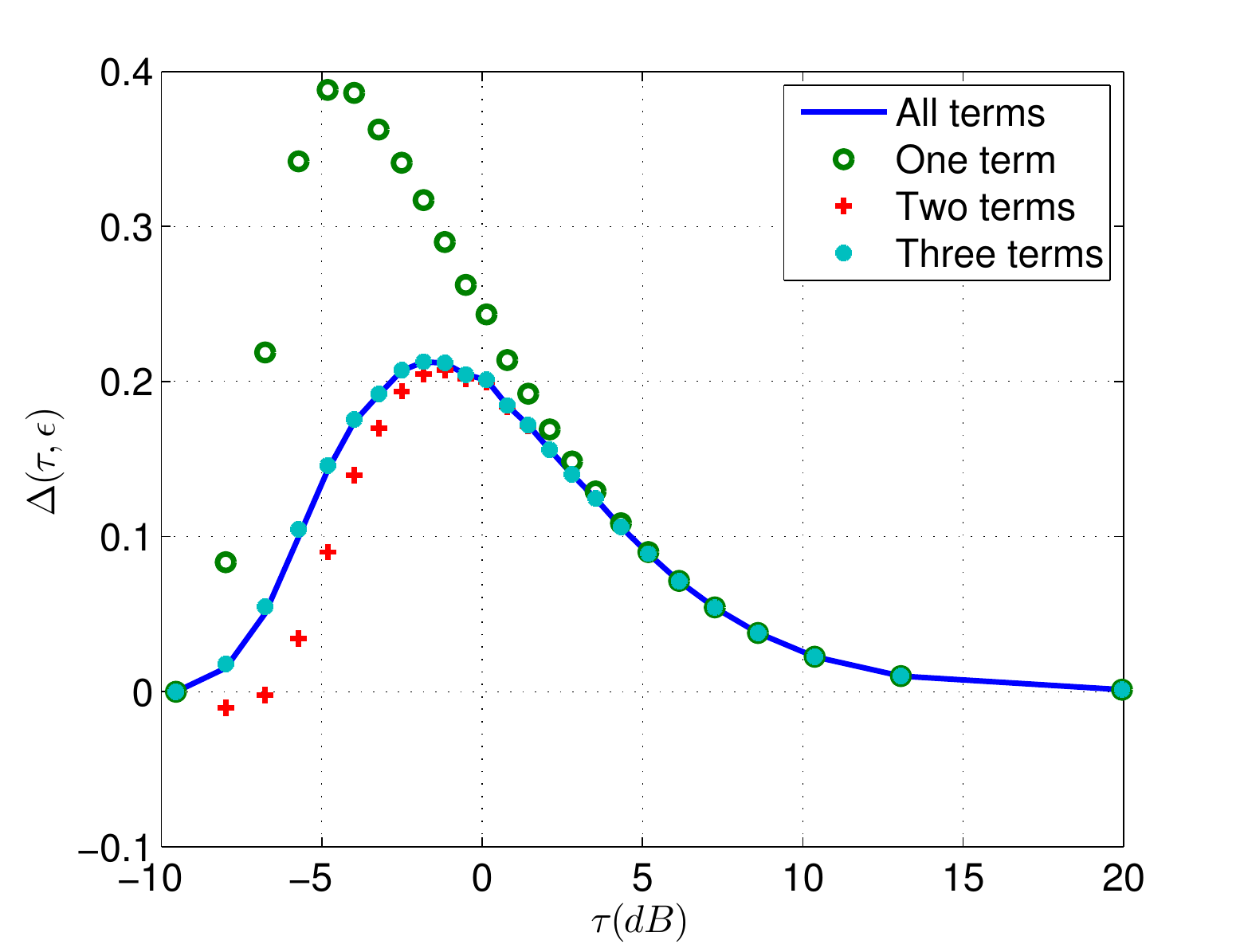}}
\caption{For $\beta=3$, the complete (10-term) expression and 1, 2, and 3-term approximations for the $\Delta^{(2)}_{SC}$ probability
 with $\epsilon'=0.1$ (or $\epsilon=-9.5424$dB).\label{PlotDelta3terms}}
\end{center}
\end{minipage}
\hfill
\begin{minipage}[t]{0.45\linewidth}
\begin{center}
\centerline{\includegraphics[scale=0.5]{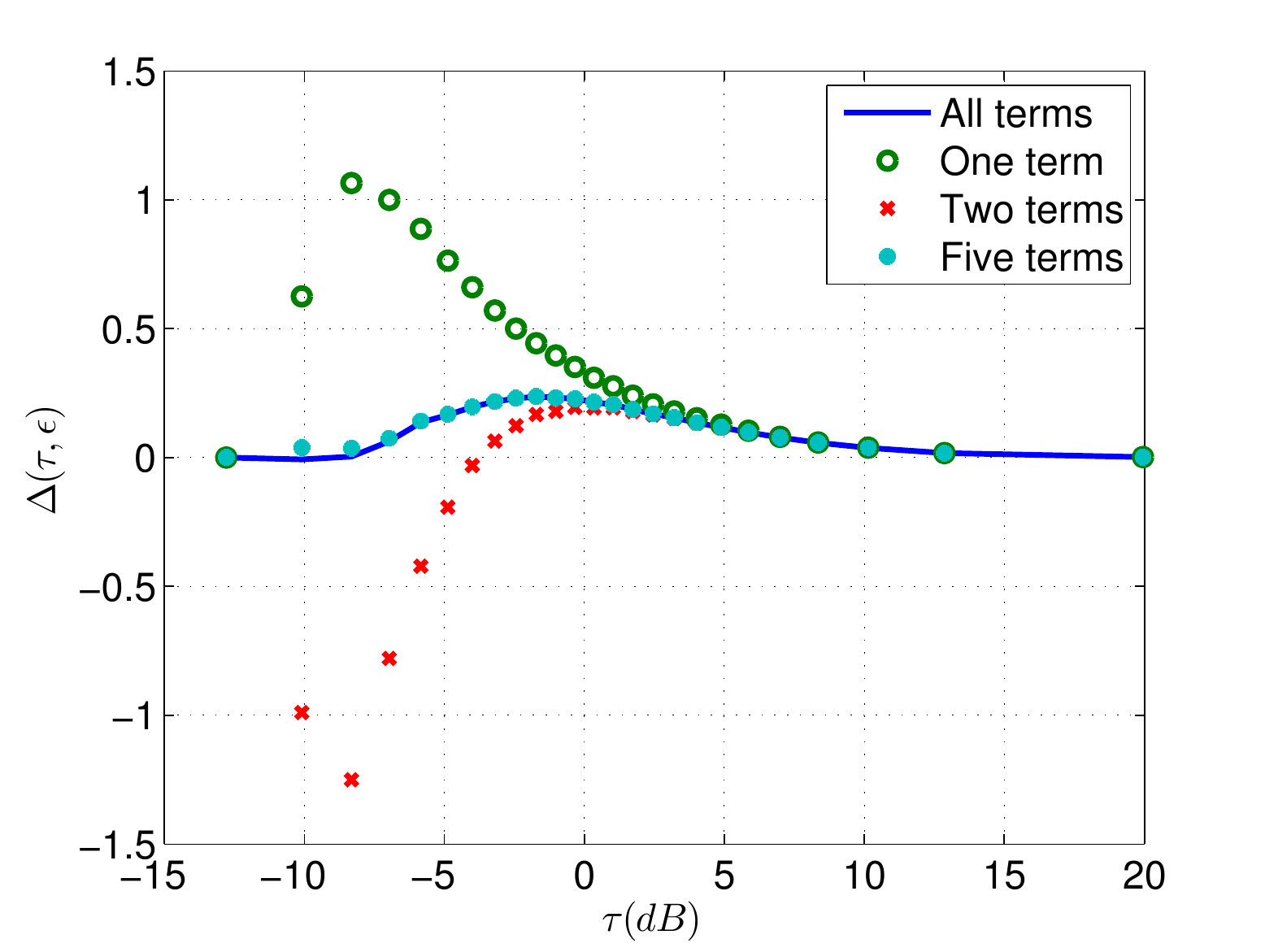}}
\caption{For $\beta=3$, the complete (20-term) expression and 1, 2, and 5-term approximations for the $\Delta_{SC}^{(2)}$ 
probability for a heterogeneous network under two base station cooperation with $\epsilon'=0.05$ (or $\epsilon=-12.7875 $dB).\label{PlotDelta5terms}}
\end{center}
\end{minipage}
\end{center}
\vspace{-4ex}
\end{figure}

\section{Future directions and conclusions}
We  propose some model and technique extensions, discuss their
feasibility and conclude the work.

\subsection{General functions of the SINR process via the
  factorial-moment expansion}

We used the factorial moment measures of the SINR process $\Psi$ to
derive expressions for $k$-coverage probabilities 
via the  famous inclusion-exclusion principle. More general
characteristics of the coverage number can be obtained directly using 
 Schuette-Nesbitt formula  (which we recall for completeness 
in Appendix~\ref{App1}, see also  \cite{GERBER:1995}). 
Using a  result by Handa~\cite[Lemma 5.3]{handa2009two},
we also  demonstrated how to calculate
the joint densities of  the order statistics of the SINR process
using its factorial moment measures.
It turns out that all these results can be seen 
as  instances of some more  general expression allowing one 
 to calculate (at least theoretically) the expectations of general
 functions $\phi$ of simple point processes by means of  
a (finite or infinite) Taylor-like expansion. The terms of this expansion 
are integrals of some kernels (entirely characterized by $\phi$)
with respect to the factorial moment measures of the considered point
process, cf~\cite{fme}.

More specifically, under appropriate convergence conditions (whose
presentation is beyond the scope of this paper)
one can write the expectation of a general function $\phi$ of the (say) STINR point process $\Psi'$ as 
% 2c format
%\begin{align}\label{e.fme}
%\E[&\phi(\Psi')]\\
%=\, &\phi(\emptyset)
%+ \sum\limits_{n=1}^{\infty}
%\int_0^{1/\gamma}\hspace{-3ex}\dots\hspace{-0.8ex}
%\int_0^{1/\gamma}\phi_{t'_1,\dots,t'_n} \, M'^{(n)} (dt'_1,\dots,dt'_n)\,,
%\nonumber
%\end{align}
\begin{align}\label{e.fme}
\E[\phi(\Psi')]
=\, \phi(\emptyset)
+ \sum\limits_{n=1}^{\infty}
\int_0^{1/\gamma}\hspace{-3ex}\dots\hspace{-0.8ex}
\int_0^{1/\gamma}\phi_{t'_1,\dots,t'_n} \, \mu'^{(n)}(t_1',\ldots, t_n')\,dt'_n\dots dt'_1\,,
\end{align}
where $\mu'^{(n)}(t_1',\ldots  t_n')$ is the density of the $n\,$th factorial moment measure of the STINR
process  given in Corollary~\ref{newmu_n},
$\phi(\emptyset)$ denotes the value of the function $\phi$
evaluated on the empty configuration of points (no points)
and 
\begin{align*}
\phi_{t'_1}&=\phi(\{t'_1\})-\phi(\emptyset)\\
\phi_{t'_1,t'_2}&= \frac{1}{2}\Bigl(\phi(\{t'_1,t'_2\})- \phi(\{t'_1\})- \phi(\{t'_2\})+\phi(\emptyset)\Bigr)\\
\dots&\\
%\phi_{Z_1,\dots,Z_n}(\emptyset) &= \sum_{\substack{j\in\{1,\dots,n\}\\i\in \{1,\dots,j\}}}\mkern-10mu  (-1) \Ind (\sum\limits_{k=i}^j Z_i>t)\\
\phi_{t'_1,\ldots,t'_n}&=\frac1{n!}\sum_{k = 0}^{n} (-1)^{n-k}
\sum_{{t'_{i_1},\ldots,t'_{i_k}}\atop{\text{distinct}}} 
\phi(\{t'_{i_1},\ldots,t'_{i_k}\})\,.
\end{align*}
For example, setting  
$$\phi(\{t'_1,\ldots\})=
\Ind(\max_{i\ge1} t'_i>\tau')$$
and applying ~(\ref{e.fme}) 
one obtains~(\ref{e.ss1}) with $k=1$ in the case of constant
STINR threshold $\tau'$.   
More generally, using~(\ref{e.fme}) with 
$$\phi(\{t'_1,\ldots\})=\prod_{k}\Bigl(h(t'_k)\Ind(t'_k>\tau')
+\Ind(t'_k\le \tau')\Bigr)$$
one obtains the following  expansion of the probability generating
function (cf~\cite[Eq. 5.5.4]{daleyPPI2003}) of the SINR process
% 2c format
%\begin{align*}
%&\E\Bigl[\prod_{Z'\in\Psi'}\Bigl(h(Z')\Ind(Z'>\tau')
%+\Ind(Z'\le \tau')\Bigr)\Bigr]\\
%&=1+\sum_{n=1}^{n_{\max}}\frac{1}{n!} \int_{\tau'}^{1/\gamma}\hspace{-1.5em}\dots\hspace{-0.3em}
%\int_{\tau'}^{1/\gamma}\!\!\prod_{k=1}^n
%(h(t'_k)-1)\,M'^{(n)}(d(t'_1,\ldots,t'_n))\,,
%\end{align*}
\begin{align*}
\E\Bigl[\prod_{Z'\in\Psi'}\Bigl(h(Z')\Ind(Z'>\tau')
+\Ind(Z'\le \tau')\Bigr)\Bigr]
=1+\sum_{n=1}^{n_{\max}}\frac{1}{n!} \int_{\tau'}^{1/\gamma}\hspace{-1.5em}\dots\hspace{-0.3em}
\int_{\tau'}^{1/\gamma}\!\!\prod_{k=1}^n
(h(t'_k)-1)\mu'^{(n)}(t_1',\ldots, t_n')\,dt'_n\dots dt'_1\, %\,M'^{(n)}(d(t'_1,\ldots,t'_n))\,,
\end{align*}
where the structure of our point process imposes
$n_{\max}<1/(\gamma\tau')$, for any bounded function  $h(\cdot)$.
This expansion can be easily transformed to the one for STINR via the
mapping~(\ref{Zrelation1}).
 See~\cite[Sec.~5.5]{daleyPPI2003}
for other classical moment expansion results.

Thus the factorial-moment expansion provides a way for calculating
various SINR-based quantities, thus opening a range of other possible
models coupled with analytic methods.

\subsection{Varying path-loss exponent model}
For tractability we used a fixed path-loss exponent $\beta$ and $K$
throughout this work. However, an argument can be made based on the
Hata-like path-loss models that both constants
incorporate base station height~\cite[Section 2.7.3]{STUBER2011}. 
Consequently, more adequate heterogeneous
  network models may require 
$\beta$ and $K$ to be {\em random}  marks of base stations.
Extension of our model to random $K$ is straightforward. Considering
random $\beta$ is more complicated. 
It was recently shown that in terms of propagation losses, a Poisson
network with random $\beta$  and constant base station density
$\lambda$ is equivalent to some network with constant exponents and an
isotropic (but not power-law) density~\cite{equivalence2013}. 
Assuming such density
increases the complexity of the factorial moment integrals,
but they may remain amenable, particularly to numerical means.  
Multi-tier models with values of
$\beta$  depending on tiers have been considered
in~\cite{JO2012,MADBROWN2012}.

\subsection{Other path-loss functions}
Path-loss functions other than the commonly used (singular) power-law
function $\ell$ have been previously proposed, including exponential
types and modified  power-laws with the singularity removed (for example, $\ell(|x|)=\min(1,(|x|)^{\beta}$). Unfortunately, two immediate drawbacks arise when employing other path-loss functions. The propagation invariance (Lemma \ref{l.invariance}) would no longer hold~\footnote{For certain path-loss functions,  less elegant invariance results may exist where propagation processes rely upon two or more propagation moments.}, hence  the factorial moments would need to be calculated for different distributions of propagation effects. Furthermore, the algebraic tractability in our work would quickly diminish; for example, the change of variables used (in Section \ref{mainResultProofPart}) would, most probably, not work in simplifying the integrals.

\subsection{Non-Poisson base station configurations}
Our derivation of the explicit expressions of $\calS_n$, via the factorial moment measures of the SINR process, hinges heavily upon the tractable nature of the Poisson process (in particular, Slivnyak's theorem). Despite this, the $M^{(n)}$-dependent $\calS_n$ terms  and related expressions introduced in Section \ref{s.symsum} do not rely upon the Poisson assumption. More specifically, assume a simple point process $\Phi=\{X\}$ representing the base stations, which gives rise to a corresponding STINR process $\modPsi=\{\modZ\}$. Then deriving the $\calS_n$ terms is again equivalent to the task of finding the corresponding factorial moment measures of $\{\modZ\}$, which in turn immediately give those of $\{Z\}$. 

However, a non-Poisson base station configuration would require considerably more work due to Slivnyak's theorem no longer being applicable. Furthermore, we have leveraged propagation invariance considerably, which flows directly from the Poisson assumption. Without it, the factorial moment measure of $\{\modZ\}$ would need to be calculated for different distributions of propagation effects given that the propagation invariance would not hold. The propagation invariance also granted us the ability to create an equivalent and more tractable network $\Phi^*$ used in Proposition~\ref{generalSn}.
A cellular network model with base stations  positioned according to a Ginibre process 
 was recently considered in~\cite{Miyoshi2014,nakata2014spatial}
where an expression for the SINR coverage probability by the closest
base station was derived assuming Rayleigh fading (which does {\em not}
impact the choice of the serving station, as in~\cite{VU2012},
\cite[Section~IV.C]{kcovsingle} for the Poisson model).

%\subsection{Other model modifications}
%Further model parameters and extensions may be considered and incorporated into the general model framework laid out here. For example, interference reduction was represented by a soul parameter $\gamma$ whereas a more sophisticated model could be proposed. This also raises the question of  how the SINR process varies between  open-access and closed-access networks, which was recently examined under a multi-tier model~\cite{DHILLON2012}. In the context of cell selection, there may be scope in the model framework for incorporating the recently proposed `biasing'~\cite{JO2012}. There is increasing research in the  analysis of energy-efficient networks, with particular focus on sleep schemes, by using techniques from stochastic geometry.  

%Dynamic analysis involving queueing models. Other networks such as cognitive radio networks

\subsection{Conclusions}
We have presented a novel framework for studying 
general down-link characteristics in a heterogeneous cellular network. 
It is based on the explicit 
evaluation of the factorial moment measures of the point process 
formed by the  SINR values observed by a  typical user  with respect to all
network stations. They allow one, in particular, to 
express the finite dimensional distributions of the ordered vectors of $k$
strongest SINR values (SINR order statistics).
We have demonstrated the benefits of this approach
by extending existing results and deriving new ones
regarding the SINR coverages probabilities, in particular pertaining
to the usage of signal combination and interference cancellation   
techniques.

\appendix

\subsection{Remarks on $\calJ_{n,\beta}$ }\label{Jnremarks}
%\subsubsection{Closed-form solution of $\calJ_{2,\beta}$ }
The integral $\calJ_{n,\beta}(x_1,\dots,x_n)$ can be written as
% 2c format
%\begin{align*}
%& \calJ_{n,\beta} (x_1,\dots,x_n)\\
%&=\frac{1}{n}\sum\limits_{j=1}^{n} \int_{[0,1]^{n-1}}
% \frac{   \prod\limits_{i=1}^{n-1}   v_i^{i(2/\beta+1)-1}(1-v_i)^{2/\beta}
%  }{ \prod\limits_{i\neq j}  (x_i+\eta_i)} dv_1\dots
%dv_{n-1}.
%\end{align*}
\begin{align*}
 \calJ_{n,\beta} (x_1,\dots,x_n)
=\frac{1}{n}\sum\limits_{j=1}^{n} \int_{[0,1]^{n-1}}
 \frac{   \prod\limits_{i=1}^{n-1}   v_i^{i(2/\beta+1)-1}(1-v_i)^{2/\beta}
  }{ \prod\limits_{i\neq j}  (x_i+\eta_i)} dv_1\dots
dv_{n-1}.
\end{align*}
For $n=1$  the integral $\calJ_{1,\beta} (x_1)=1$  and for $n=2$ the integral reduces to a sum of two integrals of the same type
% 2c format
%\begin{align*}
% \calJ_{2,\beta} (x_1,x_2)
%&=\frac{1}{2}\int_0^1   v_1^{2/\beta }(1-v_1)^{2/\beta} \\ \times
%& \left[
% \frac{ 1  }{  (x_1+v_1)}   +  \frac{ 1  }{  (x_2+1-v_1)}   \right]  dv_1 \,.
%\end{align*}
\begin{align*}
 \calJ_{2,\beta} (x_1,x_2)
=\frac{1}{2}\int_0^1   v_1^{2/\beta }(1-v_1)^{2/\beta} 
 \left[
 \frac{ 1  }{  (x_1+v_1)}   +  \frac{ 1  }{  (x_2+1-v_1)}   \right]  dv_1 \,.
\end{align*}
The general solution to this integrals is
% 2c format
%\begin{align*}
%\int_0^{1}&
%\frac{v^{2/\beta}(1-v)^{2/\beta}}{
%x+v }dv\\
%=&\frac{1}{x }B(2/\beta+1,2/\beta+1) \\&\times {}_2F_1(1,2/\beta+1;2(2/\beta+1);-1/x
%),
%\end{align*}
\begin{align*}
\int_0^{1}
\frac{v^{2/\beta}(1-v)^{2/\beta}}{
x+v }dv
=\frac{1}{x }B(2/\beta+1,2/\beta+1) \; {}_2F_1(1,2/\beta+1;2(2/\beta+1);-1/x
),
\end{align*}
where is ${}_2F_1$ a hypergeometric function ~\cite[Equation 15.6.1]{DLMF}.
% \begin{equation}
% \int_0^1
%\frac{t^{b-1}(1-t)^{c-b-1}}{(1-zt)^a}dt=\frac{\Gamma(b)\Gamma(c-b)}{\Gamma(c)}
%{}_2F_1(a,b;c;z).
%\end{equation}

%\subsubsection{$\calJ_{n,\beta}$  and the beta distribution}
Given the definition of $\calJ_{n,\beta}$, one can write 
% 2c format
%\begin{align*}
% \calJ_{n,\beta} (x_1,\dots,x_n)=&\frac{1}{n}  \E\left[H(x_1,\dots,x_n,\widetilde{V}_1,\dots,\widetilde{V}_n) \right] \\
%& \times \prod\limits_{i=1}^{n-1}B(1+2/\beta , i(2/\beta+1)),
%\end{align*}
\begin{align*}
 \calJ_{n,\beta} (x_1,\dots,x_n)=\frac{1}{n}  \E\left[H(x_1,\dots,x_n,\widetilde{V}_1,\dots,\widetilde{V}_n) \right]  \prod\limits_{i=1}^{n-1}B(1+2/\beta , i(2/\beta+1)),
\end{align*}
where the permutation-invariant function
% 2c format
%\begin{align*}
%H(x_1,\dots,x_n,&\widetilde{V}_1,\dots,\widetilde{V}_n)\\
%&=  \sum\limits_{j=1}^{n} (x_j+\widetilde{V}_j) \prod\limits_{i=1}^{n}   (x_i+\widetilde{V}_i)^{-1}   \\
%&= (\sum\limits_{j=1}^{n} x_j+1)  \prod\limits_{i=1}^{n}   (x_i+\widetilde{V}_i)^{-1}   
%\end{align*}
\begin{align*}
H(x_1,\dots,x_n,\widetilde{V}_1,\dots,\widetilde{V}_n)
&=  \sum\limits_{j=1}^{n} (x_j+\widetilde{V}_j) \prod\limits_{i=1}^{n}   (x_i+\widetilde{V}_i)^{-1}   \\
&= (\sum\limits_{j=1}^{n} x_j+1)  \prod\limits_{i=1}^{n}   (x_i+\widetilde{V}_i)^{-1}   
\end{align*}
and  $\widetilde{V}_i$ are  beta random variables with distributions
$B(2/\beta+1,i(2/\beta+1))$. This interpretation may offer an
alternative way for numerically evaluating $\calJ_{n,\beta}$; for
example, simulating $V_i$ and estimating the mean of $H$. Furthermore,
as remarked earlier there seems to be a deeper connection between the
STINR process and Poisson-Dirichlet processes owing to them
both being connected to beta variables of this type, cf.~\cite{handa2009two}.

\subsection{Proof of Theorem
  \ref{mainResult}}\label{mainResultProofPart}
In view of the propagation invariance
  Lemma~\ref{l.invariance}, without loss of generality we can assume
  $K=1$ and $(PS)$ having exponential distribution with mean one and 
  replace $\lambda$ with  $a/(\pi \Gamma(1+2/\beta))$, where $a$ is
  given by~(\ref{e.a}) (note that  
$\Gamma(2/\beta+1)$ is the $2/\beta\,$th moment of exponential variable of mean one).

 The definition of $\Psi'=\{\modZ\}$ and its factorial moment
 measure~(\ref{e.Mn})  
give via the  Campbell's formula (for factorial moment measures) and Slivnyak theorem (for example, see~\cite{SKM:1995}) %we arrive at
% 2c format
%\begin{align*}
%\modM&^{(n)}(t'_1,\dots  t'_n)
%\\ =&a^n  \int\limits_{(\R^2)^n} \E \left(
%\prod_{j=1}^{n}  \Ind\left(\frac{\hatS_j t'_j}{W+\gamma[I+\sum\limits_{i=1}^n\hatS_i t'_i] }>t'_j\right)\right) \\
%&\times  dx_1\dots dx_n ,
%\end{align*} 
\begin{align*}
\modM^{(n)}(t'_1,\dots  t'_n)
=a^n  \int\limits_{(\R^2)^n} \E \left(
\prod_{j=1}^{n}  \Ind\left(\frac{\hatS_j t'_j}{W+\gamma[I+\sum\limits_{i=1}^n\hatS_i t'_i] }>t'_j\right)\right) 
 dx_1\dots dx_n ,
\end{align*} 
where  $\hatS_i:=(P_iS_i)/(\ell(r_i)t'_i)$. The integral reduces to
% 2c format
%\begin{align*}
%&\modM^{(n)}(t'_1,\dots  t'_n)
%\\ &=(2\pi a)^n  \int\limits_{(\R_+)^n} \E \left(
%\prod_{j=1}^{n}  \Ind\left(\frac{\hatS_j}{W+\gamma[I+\sum\limits_{i=1}^n \hatS_i t'_i] }>1 \right)\right) \nonumber  \\[1ex]
%&\hspace{0.7\linewidth}\times   r_1dr_1\dots r_ndr_n 
%\\[2ex] &=(2\pi a)^n \!\!\!\int\limits_{(\R_+)^n} \!\!\!\Prob \left(
%\frac{\min[\hatS_1,\dots,\hatS_n]}{W+\gamma[I+\sum\limits_{i=1}^n\hatS_it'_i ]}>1 \right)\,r_1dr_1\dots r_ndr_n.
%\end{align*} 
\begin{align*}
\modM^{(n)}(t'_1,\dots  t'_n)
&=(2\pi a)^n  \int\limits_{(\R_+)^n} \E \left(
\prod_{j=1}^{n}
\Ind\left(\frac{\hatS_j}{W+\gamma[I+\sum\limits_{i=1}^n \hatS_i t'_i]
  }>1 \right)\right) 
   r_1dr_1\dots r_ndr_n\\[2ex] 
&=(2\pi a)^n \!\!\!\int\limits_{(\R_+)^n} \!\!\!\Prob \left(
\frac{\min[\hatS_1,\dots,\hatS_n]}{W+\gamma[I+\sum\limits_{i=1}^n\hatS_it'_i ]}>1 \right)\,r_1dr_1\dots r_ndr_n.
\end{align*} 
Since $P_iS_i$ are exponentially distributed with unit means, then $\hatS_i$ are also exponentially distributed with means $1/\mu_i=1/(t'_i\ell(r_i))$, and let
\[
 \hatM=\min(\hatS_1,\dots,\hatS_n), 
\]
which is another exponential variable with parameter $\mu_M=\sum\limits_{i=1}^n \mu_i $.
Let
\[
\quad {D}=\sum_{i=1}^n \hatS_it'_i -\hatM\sum_{i=1}^n t'_i=
\sum_{i=1}^n (\hatS_i-\hatM) t'_i\,.
\]
Note, $\hatM=\hatS_j$ with probability $\mu_j/(\sum\limits_{i=1}^n\mu_i)$. Given this event
%${D}=\sum\limits_{i\not=j}^n (\hatS_i-\hatS_j) t'_i$, 
$\hatS_i-\hatS_j$ ($i\not=j$) are conditionally independent, exponential variables with parameters $\mu_i$, and independent of $\hatM$.  
(This is the well known  memoryless property of the exponential distribution.) 
Consequently $\hatM$ and
${D}$ are independent and the latter has the following mixed exponential distribution
%\begin{align}\label{DistD}
%& \Prob ({D} \leq
%d)\\ 
%& =\frac{1}{\sum\limits_{i=1}^n\mu_i}\left[\mu_1\Prob(\sum\limits_{i=2}^n \hatS_it_i\leq
%d)+\dots+ \mu_n\Prob(\sum\limits_{i=1}^{n-1}\hatS_it_i\leq d)\right].
%\end{align}
\begin{equation}\label{DistD}
\Prob ({D} \leq
d) =\frac{1}{\sum\limits_{i=1}^n\mu_i}\left[\sum\limits_{j=1}^n \mu_j\Prob(\sum\limits_{i\neq j} \hatS_it'_i\leq
d)\right].
\end{equation}
Observe that
% 2c format
%\begin{align*}
% \Prob & \left(\frac{\min(\hatS_1,\dots,\hatS_n)}{W+\gamma[I+\sum\limits_{i=1}^n\hatS_it'_i ]}>1 \right)  \nonumber
% \\ &=
% \Prob \left(\frac{\hatM}{W+\gamma[I+D+\hatM\sum\limits_{i=1}^n t'_i  ]}>1 \right) \\
%& = 
% \Prob \left(\frac{\hatM}{W+\gamma[I+D]}>\hatT_n/\gamma \right) \\
%& = 
% \Prob \left(\frac{\hatM}{W/\gamma+I+D}>\hatT_n \right),
%\end{align*} 
\begin{align*}
 \Prob  \left(\frac{\min(\hatS_1,\dots,\hatS_n)}{W+\gamma[I+\sum\limits_{i=1}^n\hatS_it'_i ]}>1 \right)  
&=
 \Prob \left(\frac{\hatM}{W+\gamma[I+D+\hatM\sum\limits_{i=1}^n t'_i  ]}>1 \right) \\
& = 
 \Prob \left(\frac{\hatM}{W+\gamma[I+D]}>\hatT_n/\gamma \right) \\
& = 
 \Prob \left(\frac{\hatM}{W/\gamma+I+D}>\hatT_n \right),
\end{align*} 
where
$$
\hatT_{n}:=\hatT_{n}(t_1,\dots,t_n)=\gamma/(1-\gamma\sum\limits_{i=1}^n t'_i).
$$
Given that $W$, $I$, $D$, and 
$\hatM$ are mutually independent, and  that $\hatM$ is exponentially distributed, hence
% 2c format
%\begin{align*}
%\Prob \left(\frac{\hatM}{W/\gamma+I+D}>\hatT_n \right)&\\ 
%=\Lap_{W/\gamma}(\mu_{M}\hatT_n)&\Lap_I(\mu_{M}\hatT_n)\Lap_{{D}}(\mu_{M}\hatT_n),
%\end{align*}
\begin{align*}
\Prob \left(\frac{\hatM}{W/\gamma+I+D}>\hatT_n \right)
=\Lap_{W/\gamma}(\mu_{M}\hatT_n)&\Lap_I(\mu_{M}\hatT_n)\Lap_{{D}}(\mu_{M}\hatT_n),
\end{align*}
which is a product of three Laplace transforms. The first transform  
$$
 \Lap_{W/\gamma}(\xi)=e^{-\xi W/\gamma}
$$
and the second (see equation $2.25$ in \cite{FnT1})
$$
 \Lap_I(\xi )=e^{-a\xi^{2/\beta}\pi C'(\beta)/K^2}
$$
 immediately follow.
Given expression (\ref{DistD}), then the Laplace transform of a
general exponential variable
and the convolution theorem
imply that the distribution of ${D}$ has the transform
$$
 \Lap_{{D}}(\xi)
=\frac{\prod\limits_{i=1}^{n}\mu_i}{\sum\limits_{i=1}^{n}\mu_i} \left[
\sum\limits_{j=1}^{n} \frac{1}{\left(\prod\limits_{i\neq j} [\mu_j+t'_j\xi
]\right)}\right].
$$
After substituting the explicit path-loss function (\ref{PATHLOSS}) with $K=1$, which will be recovered later, and some algebra we arrive at
% 2c format
%\begin{align*}
% \Prob &\left(\frac{\hatM}{W/\gamma+I+D}>\hatT_n \right)  \nonumber \\ &=\frac{\left(\prod\limits_{i=1}^{n}
%t'_ir_i^{\beta} e^{-(W/\gamma)\hatT_n t'_ir_i^{\beta}} \right)e^{-a(\hatT_n\sum\limits_{i=1}^{n}
%t'_ir_i^{\beta})^{2/\beta}
%\pi C'(\beta)} }{\sum\limits_{i=1}^{n} t'_i r_i^{\beta}} \nonumber\\
%&\times\sum\limits_{j=1}^{n} \left(\frac{1}{\prod\limits_{i\neq j}
%[t'_ir_i^{\beta}+t_i\hatT_n \sum\limits_{k=1}^{n} t'_k r_k^{\beta}
%]}\right).
%\end{align*}
\begin{align*}
 \Prob &\left(\frac{\hatM}{W/\gamma+I+D}>\hatT_n \right)  \nonumber \\[1.5ex] 
&=\frac{\left(\prod\limits_{i=1}^{n}
t'_ir_i^{\beta} e^{-(W/\gamma)\hatT_n t'_ir_i^{\beta}} \right)e^{-a(\hatT_n\sum\limits_{i=1}^{n}
t'_ir_i^{\beta})^{2/\beta}
\pi C'(\beta)} }{\sum\limits_{i=1}^{n} t'_i r_i^{\beta}} 
\sum\limits_{j=1}^{n} \left(\frac{1}{\prod\limits_{i\neq j}
[t'_ir_i^{\beta}+t_i\hatT_n \sum\limits_{k=1}^{n} t'_k r_k^{\beta}
]}\right).
\end{align*}
Hence, the integral
% 2c format
%\begin{align*}
%&\modM^{(n)}\left(t'_1,\dots ,t'_n\right) \nonumber
%\\ &=(2\pi a)^n \int\limits_{(\R_+)^n} \frac{\left(\prod\limits_{i=1}^{n}
%t'_ir_i^{\beta} e^{-(W/\gamma)\hatT_n t'_ir_i^{\beta}} \right) }{\sum\limits_{i=1}^{n} t'_i r_i^{\beta}} \nonumber  \\
%& \times e^{-a(\hatT_n\sum\limits_{i=1}^{n}
%t'_ir_i^{\beta})^{2/\beta}
%\pi C'(\beta)} \nonumber  \\
%&\times\sum\limits_{j=1}^{n} \left(\frac{1}{\prod\limits_{i\neq j}
%[t'_ir_i^{\beta}+t'_i\hatT_n \sum\limits_{k=1}^{n} t'_k r_k^{\beta}
%]}\right) r_1dr_1\dots r_ndr_n,
%\end{align*} 
\begin{align*}
&\modM^{(n)}\left(t'_1,\dots ,t'_n\right) \nonumber
\\ &=(2\pi a)^n \int\limits_{(\R_+)^n} \frac{\left(\prod\limits_{i=1}^{n}
t'_ir_i^{\beta} e^{-(W/\gamma)\hatT_n t'_ir_i^{\beta}} \right)
}{\sum\limits_{i=1}^{n} t'_i r_i^{\beta}} 
 e^{-a(\hatT_n\sum\limits_{i=1}^{n}
t'_ir_i^{\beta})^{2/\beta}
\pi C'(\beta)} 
\sum\limits_{j=1}^{n} \left(\frac{1}{\prod\limits_{i\neq j}
[t'_ir_i^{\beta}+t'_i\hatT_n \sum\limits_{k=1}^{n} t'_k r_k^{\beta}
]}\right) r_1dr_1\dots r_ndr_n,
\end{align*} 
with the variable change $s_i:=r_it_i'^{1/\beta}(a \hatT_n^{2/\beta}\pi
C'(\beta))^{1/2}$ reduces to 
% 2c format
%\begin{align}\label{integral_si}
%&\modM^{(n)}\left(t'_1,\dots ,t'_n\right) \nonumber
%\\ &=\frac{2^n \hatT_n^{-2n/\beta}}{(C'(\beta))^n} \nonumber  \\ 
%&\times \int\limits_{(\R_+)^n} \frac{\left(\prod\limits_{i=1}^{n}
%t_i'^{-2/\beta} s_i^{\beta+1} e^{- (W/\gamma)  (a\pi
%C'(\beta))^{-\beta/2}  s_i^{\beta}} \right) }{\sum\limits_{i=1}^{n}  s_i^{\beta}}   \nonumber \\
%&\times e^{-(\sum\limits_{i=1}^{n}
%s_i^{\beta})^{2/\beta}}  \sum\limits_{j=1}^{n} \left(\frac{1}{\prod\limits_{i\neq j}
%[s_i^{\beta}+t'_i\hatT_n\sum\limits_{k=1}^{n}  s_k^{\beta}
%]}\right) ds_1\dots ds_n .
%\end{align} 
\begin{align}\label{integral_si}
&\modM^{(n)}\left(t'_1,\dots ,t'_n\right) \nonumber
\\ &=\frac{2^n \hatT_n^{-2n/\beta}}{(C'(\beta))^n} 
 \int\limits_{(\R_+)^n} \frac{\left(\prod\limits_{i=1}^{n}
t_i'^{-2/\beta} s_i^{\beta+1} e^{- (W/\gamma)  (a\pi
C'(\beta))^{-\beta/2}  s_i^{\beta}} \right) }{\sum\limits_{i=1}^{n}  s_i^{\beta}}  
e^{-(\sum\limits_{i=1}^{n}
s_i^{\beta})^{2/\beta}}  \sum\limits_{j=1}^{n} \left(\frac{1}{\prod\limits_{i\neq j}
[s_i^{\beta}+t'_i\hatT_n\sum\limits_{k=1}^{n}  s_k^{\beta}
]}\right) ds_1\dots ds_n .
\end{align} 

We introduce a change of variables inspired by the $n$-dimensional spherical coordinates 
\begin{align*}
 s_1&=u[\sin\theta_1 \sin\theta_2\dots\sin\theta_{n-1}]^{2/\beta}\\
s_2&=u[\cos\theta_1 \sin\theta_2\dots\sin\theta_{n-1}]^{2/\beta}\\
s_3&=u[\cos\theta_2 \sin\theta_3\dots\sin\theta_{n-1}]^{2/\beta}\\
&\cdots \nonumber \\
s_n&=u[\cos\theta_{n-1}]^{2/\beta}.
\end{align*}
Note that $u^{\beta}=\sum_{i=1}^{n} s_i^{\beta} $ and  $
 \prod_{i=1}^{n} s_i  =u^n [ \prod_{i=1}^{n} q_i]^{2/\beta}$ where we use the shorthand $q_i=q_i(\theta_1,\dots,\theta_n):=(s_i/u_i)^{\beta/2}$ for $u_i>0$. For $\beta=2$,  we obtain $n$-dimensional
spherical coordinates with a Jacobian
$$
 \hat{J}(\hat{u},\hat{\theta}_1,\dots,\hat{\theta}_n)=\hat{u}^{n-1} \prod_{i=1}^{n-1} \sin^{i-1}\hat{\theta}_{i},
$$
while our variables have the Jacobian
% 2c format
%\begin{align*}
% &J(u,\theta_1,\dots,\theta_n)\\&=\left(\frac{2}{\beta}\right)^{n-1}\hat{J}(u,
%\theta_1,\dots,\theta_n)  \left[\prod_{i=1}^{n-1}\sin^{i}\theta_i\cos\theta_i
%\right]^{2/\beta-1}.
%\end{align*}
\begin{align*}
J(u,\theta_1,\dots,\theta_n)=\left(\frac{2}{\beta}\right)^{n-1}\hat{J}(u,
\theta_1,\dots,\theta_n)  \left[\prod_{i=1}^{n-1}\sin^{i}\theta_i\cos\theta_i
\right]^{2/\beta-1}.
\end{align*}

Denote $z:= (W/\gamma) (a\Gamma(1-2/\beta))^{-\beta/2}$. The change of variable renders the integral as
% 2c format
%\begin{align*}
%&\modM^{(n)}\left(t_1,\dots ,t_n\right)
%\\ 
%&=\frac{2^n \hatT_n^{-2n/\beta}}{(C'(\beta))^n} \int\limits_{(\R_+)^n} \frac{\left(\prod\limits_{i=1}^{n}
%t_i'^{-2/\beta} s_i^{\beta+1} e^{-z s_i^{\beta}} \right) }{\sum\limits_{i=1}^{n}  s_i^{\beta}}\\
%&\times e^{-(\sum\limits_{i=1}^{n}
%s_i^{\beta})^{2/\beta}}  \sum\limits_{j=1}^{n} \left(\frac{1}{\prod\limits_{i\neq j}
%[s_i^{\beta}+t'_i \hatT_n \sum\limits_{k=1}^{n}  s_k^{\beta}
%]}\right) ds_1\dots ds_n \\
%&=\frac{2^n \hatT_n^{-2n/\beta}}{(C'(\beta))^n} \left(\frac{2}{\beta}\right)^{n-1} \prod\limits_{i=1}^{n}
%t_i'^{-2/\beta} \int\limits_0^{\infty} u^{2n-1}  e^{-(
%u^2 +z u^{\beta}) } du \\ 
%& \times  \mkern-23mu \int\limits_{[0,\pi/2]^{n-1}}\!   \sum\limits_{j=1}^{n}  \!  \frac{
%  \prod\limits_{i=1}^{n-1}  \left[ \sin^{i}\theta_i\cos\theta_i \right]^{4/\beta+1} 
%   [\sin\theta_{i}]^{i-1} }{   \prod\limits_{i\neq j}  [q_i^2 +t'_i\hatT_n 
%] }   d\theta_1\dots d\theta_{n-1} .
%\end{align*}
\begin{align*}
&\modM^{(n)}\left(t_1,\dots ,t_n\right)
\\ 
&=\frac{2^n \hatT_n^{-2n/\beta}}{(C'(\beta))^n} \int\limits_{(\R_+)^n} \frac{\left(\prod\limits_{i=1}^{n}
t_i'^{-2/\beta} s_i^{\beta+1} e^{-z s_i^{\beta}} \right)
}{\sum\limits_{i=1}^{n}  s_i^{\beta}}
 e^{-(\sum\limits_{i=1}^{n}
s_i^{\beta})^{2/\beta}}  \sum\limits_{j=1}^{n} \left(\frac{1}{\prod\limits_{i\neq j}
[s_i^{\beta}+t'_i \hatT_n \sum\limits_{k=1}^{n}  s_k^{\beta}
]}\right) ds_1\dots ds_n \\
&=\frac{2^n \hatT_n^{-2n/\beta}}{(C'(\beta))^n} \left(\frac{2}{\beta}\right)^{n-1} \prod\limits_{i=1}^{n}
t_i'^{-2/\beta} \int\limits_0^{\infty} u^{2n-1}  e^{-(
u^2 +z u^{\beta}) } du \mkern-23mu \int\limits_{[0,\pi/2]^{n-1}}\!   \sum\limits_{j=1}^{n}  \!  \frac{
  \prod\limits_{i=1}^{n-1}  \left[ \sin^{i}\theta_i\cos\theta_i \right]^{4/\beta+1} 
   [\sin\theta_{i}]^{i-1} }{   \prod\limits_{i\neq j}  [q_i^2 +t'_i\hatT_n 
] }   d\theta_1\dots d\theta_{n-1} .
\end{align*}
We then substitute $v_i=\sin^2 \theta_i$, and define $\eta_i$ and
$\hatt_i$ accordingly. 
We define $\calJ_n(x_1,\dots,x_n)$  to be analogous (to equation (15) in \cite{kcovsingle}), 
% 2c format
%\begin{align*}
%& \calJ_{n,\beta} (x_1,\dots,x_n)\nonumber\\
%&=\frac{1}{n}  \int\limits_{[0,1]^{n-1}} \sum\limits_{j=1}^{n}
% \frac{   \prod\limits_{i=1}^{n-1}   v_i^{i(2/\beta+1)-1}(1-v_i)^{2/\beta}
%  }{ \prod\limits_{i\neq j}  (x_i+\eta_i)} dv_1\dots
%dv_{n-1},
%\end{align*}
\begin{align*}
 \calJ_{n,\beta} (x_1,\dots,x_n)=\frac{1}{n}  \int\limits_{[0,1]^{n-1}} \sum\limits_{j=1}^{n}
 \frac{   \prod\limits_{i=1}^{n-1}   v_i^{i(2/\beta+1)-1}(1-v_i)^{2/\beta}
  }{ \prod\limits_{i\neq j}  (x_i+\eta_i)} dv_1\dots
dv_{n-1},
\end{align*}
but observe, since $\sum_{j=1}^{n}   \eta_j=1$, that the above reduces to integral (\ref{Jn}),  completing the proof. 

\subsection{Partial densities of $M'^{(n)}$}\label{densityMn}
We will now detail how to calculate the
 partial  derivatives $\frac{\partial^k}{\partial t'_1,\ldots\partial
  t'_k}$ of $M'^{(n)}$, for $k\le n$ and thus obtain explicit
expressions for $\mu'^{k+i}_k$. Note first that the integrals
$\calI_{n,\beta}$ do not  depend on $\{t'_i\}$ and write
\begin{align*}
\frac{M'^{(n)}(t'_1,\ldots,t'_n)}{ n! \calI_{n,\beta}(W a^{-2/\beta}) } =&   \prod\limits_{i=1}^{n} \hatt_i^{-2/\beta} \calJ_{n,\beta}(\{\hatt_i\})\\
%=&\frac{1}{n}\int_{[0,1]^{n-1}}
% \frac{   (1+\sum_{j=1}^{n}\hatt_j)
%  }{ \prod\limits_{i=1}^n  \hatt_i^{2/\beta}   (\hatt_i+\eta_i)} \\ 
%&\times 
% \prod\limits_{i=1}^{n-1}   v_i^{i(2/\beta+1)-1}(1-v_i)^{2/\beta}   dv_i \\
=&\frac{1}{n} \int_{[0,1]^{n-1}} R(\{v_i\})Q(\{\hatt_i\},\{\eta_i\})  dv_i,
\end{align*} 
where
\begin{equation}\label{e.Q100}
 Q(\{\hatt_i\},\{\eta_i\})= \frac{  (1+\sum_{j=1}^{n}  \hatt_j)
  }{ \prod\limits_{i=1}^n    (\hatt_i^{2/\beta+1} +\hatt_i^{2/\beta} \eta_i)} ,
\end{equation}
\begin{equation}
 R(\{v_i\}):= \prod\limits_{i=1}^{n-1}   v_i^{i(2/\beta+1)-1}(1-v_i)^{2/\beta} ,
\end{equation}
$\{\hatt_{i}\}$ are related to $\{t'_i\}$ via~(\ref{hattn})
and $\{\eta_i\}$ are  related to $\{v_i\}$ via~(\ref{e.eta-v}).
%\newpage
%\subsection{Second order partial derivative of $Q^{(n)}$}
Note that only $Q$ depends on $\{t'_i\}$,
specifically 
% 2c format
%\begin{align*}
%Q:=Q(\{ t_i' \})= & \frac{  \gamma^{-{b} }(1-\gamma\sum_{j=1}^{n}
%t_j')^{{b} } }{ \prod\limits_{i=1}^n 
% [t_i'^{\alpha+1} +(\eta_i/\gamma)(1-\gamma\sum_{j=1}^{n}  t_%j' )t_i'^{\alpha}  ] }\\
%&\times  \frac{ \gamma  \sum_{j=1}^{n}  t_j'}{(1-\gamma\sum_{j=1}^{n}  t_j') } ,
%\end{align*} 
\begin{align*}
Q:=Q(\{ t_i' \})=  \frac{  \gamma^{-{b} }(1-\gamma\sum_{j=1}^{n}  t_j')^{{b} } }{ \prod\limits_{i=1}^n   [t_i'^{\alpha+1} +(\eta_i/\gamma)(1-\gamma\sum_{j=1}^{n}  t_j' )t_i'^{\alpha}  ] }  \frac{ \gamma  \sum_{j=1}^{n}  t_j'}{(1-\gamma\sum_{j=1}^{n}  t_j') } ,
\end{align*} 
where $\alpha=2/\beta$ and ${b}=n(\alpha+1)$.  

This can be written as 
\[
Q: = \bar Q+\sum_{j=1}^{n} h^{(j)} \bar Q
\]
where 
\[
\bar Q(\{ t_i' \}):=\bar Q =\frac{  \gamma^{-{b} }(1-\gamma\sum_{j=1}^{n}  t_j')^{{b} }}{ \prod\limits_{i=1}^n    [t_i'^{\alpha+1} +(\eta_i/\gamma)(1-\gamma\sum_{j=1}^{n}  t_j' )t_i'^{\alpha}  ]} ,
\]
and 
\[
  h^{(j)}(\{ t_i' \}):=h^{(j)}=\frac{\gamma t_j'}{(1-\gamma\sum_{i=1}^{n}  t_i')}.
\]

\begin{remark}
Calculating the partial derivatives $\frac{\partial^k}{\partial t'_1,\ldots\partial
  t'_k}$ of $\modM^{(n)}\left(t_1',\dots ,t_n'\right)$ reduces to the
calculation of these derivatives for  $\bar Q$ and $h^{(j)}\bar
Q$. This is a tedious but
straightforward task. One can use computer symbolic integration
(for example, Maple) to easily obtain explicit expressions for these
derivatives. Due to the space constraint we will not develop here
explicit expressions for the general case $k\le n$ but only for $k=2$ and
arbitrary $n\ge 2$. These expressions are used in the numerical
examples presented in this paper.
\end{remark}

\subsubsection{Second-order derivatives}
Assume $k=2$ and $n\ge 2$. Let us introduce the auxiliary functions
\begin{align*}
A:=&\gamma^{-{b} }(1-\gamma\sum_{j=1}^{n}  t_j')^{{b} }\\
 B:= &\prod\limits_{i=1}^n    [t_i'^{\alpha+1} +(\eta_i/\gamma)(1-\gamma\sum_{j=1}^{n}  t_j' )t_i'^{\alpha}  ]^{-1},
\end{align*} 
write $\bar Q=AB$, and adopt a subscript notation to denote partial
derivatives $\frac{\partial^2}{\partial t'_1,\partial '_2}$ such that
\begin{align*}
\bar Q_{12}:=&\frac{\partial^2 \bar Q}{\partial t_1'  t_2' }
=A_{12}B+A_{1}B_{2}+A_{2}B_{1}+AB_{12},
\end{align*} 
and 
% 2c format
%\begin{align*}
%[h^{(j)} \bar Q]_{12}:=&\frac{\partial^2[h^{(j)} \bar Q] }{\partial t_1'  t_2' }\\
%&=h^{(j)}_{12}\bar Q+h^{(j)}_{1}\bar Q_{2}+h^{(j)}_{2}\bar Q_{1}+h^{(j)}\bar Q_{12},
%\end{align*} 
\begin{align*}
[h^{(j)} \bar Q]_{12}:=\frac{\partial^2[h^{(j)} \bar Q] }{\partial t_1'  t_2' }
=h^{(j)}_{12}\bar Q+h^{(j)}_{1}\bar Q_{2}+h^{(j)}_{2}\bar Q_{1}+h^{(j)}\bar Q_{12},
\end{align*} 
where
\[
\bar Q_{1}=A_{1}B+AB_1, \quad \bar Q_{2}=A_{2}B+AB_2.
\]
Now only the partial derivatives of the functions $h^{(j)}$, $A$ and
$B$ remain to be calculated, which we do in what follows.

\paragraph{Derivatives of $h^{(j)}$}
For $j=1$ or $2$,  $h^{(j)}$ has the first order derivative
\[
h^{(j)}_j=\frac{\gamma [1-\gamma(\sum_{i=1}^{n}  t_i' - t_j')]}{(1-\gamma\sum_{i=1}^{n}  t_i')^2}
\]
and the second order derivative 
\[
h^{(j)}_{12}=\frac{-\gamma^2 +2\gamma^2[1-\gamma(\sum_{i=1}^{n}  t_i' - t_j')](1-\gamma\sum_{i=1}^{n}  t_i')^{-1}}{(1-\gamma\sum_{i=1}^{n}  t_i')^2}
\]
For $k>2$ and $j=1$ or $2$,  the first order derivative is
\[
h^{(k)}_j=\frac{\gamma^2 (t_k') }{(1-\gamma\sum_{i=1}^{n}  t_i')^2},
\]
while the second order derivative is
\[
h^{(k)}_{12}=\frac{2\gamma^3 (t_k') }{(1-\gamma\sum_{i=1}^{n}  t_i'^3)}.
\]

\paragraph{Derivatives of $A$}
For the function $A$, the partial derivatives  are immediately given
\begin{align*}
A_1&=A_2=-{b}\gamma^{-{b}+1 }(1-\gamma\sum_{j=1}^{n}  t_j')^{{b}-1 }\\
 A_{12}&= {b}({b}-1)\gamma^{-{b}+2 }(1-\gamma\sum_{j=1}^{n}  t_j')^{{b}-2 }.
\end{align*}

%\[
%A_1=A_2={b}\gamma^{-{b}+1 }(1-\gamma\sum_{j=1}^{n}  t_j')^{{b}-1 },
%\]
%as well as the second
%\[
% A_{12}= {b}({b}-1)\gamma^{-{b}+2 }(1-\gamma\sum_{j=1}^{n}  t_j')^{{b}-2 }.
%\]

\paragraph{Derivatives of $B$}
The complexity of $B$ requires a couple of extra steps motivating further auxiliary functions
\[
C^{(n)}:=\prod_{j=3}^{n}   B^{(j)},
\]
where 
\[
B^{(j)} =[D^{(j)}- \eta_j t_j'^{\alpha}(t_1'+t_2')]^{-1}
\]
and
\[
D^{(j)} = t_j'^{\alpha+1}+(\eta_j/\gamma )(1 -\gamma\sum_{i=3}^{n}  t_i')   t_j'^{\alpha},
\]
while noting that $D^{(1)} $ and $D^{(2)}$ are respectively independent of $t_2'$ and $t_1'$. One can  write $B= B^{(1)} B^{(2)} C^{(n)}$ (for $n=2$, the empty product implies $C^{(2)}=1$), hence the partial derivative is  given by
\begin{align*}
B_{12}=\,&B^{(1)}_{12} B^{(2)} C^{(n)}+B^{(1)}_1 B^{(2)} _2C^{(n)}+B^{(1)}_2 B^{(2)}_1 C^{(n)}\\
&+B^{(1)} B^{(2)}_{12} C^{(n)}+B^{(1)} B^{(2)}_1 C^{(n)}_2+B^{(1)} B^{(2)}_2 C^{(n)}_1\\
&+B^{(1)}_1 B^{(2)} C^{(n)}_2+B^{(1)}_2 B^{(2)} C^{(n)}_1+B^{(1)} B^{(2)} C^{(n)}_{12}.
\end{align*}
The first order derivatives are given by
\begin{align*}
B^{(1)}_1=&-(D^{(1)}_1 - \eta_1[ (\alpha+1)t_1'^{\alpha} +\alpha t_1'^{\alpha-1}t_2]) (B^{(1)})^2 \\
B^{(2)}_2=&-(D^{(2)}_2 - \eta_2[ (\alpha+1)t_2'^{\alpha} +\alpha t_2'^{\alpha-1}t_1]) (B^{(2)})^2,
\end{align*}
where for $j=1$ or $2$, 
\[
D^{(j)}_j = (\alpha+1)t_j'^{\alpha}+(\eta_j/\gamma )\alpha t_j'^{\alpha-1}(1 -\gamma\sum_{i=3}^{n}  t_i') .
\]
Furthermore,
\begin{align*}
B^{(1)}_2=&\eta_1\alpha t_1'^{\alpha} (B^{(1)})^2 \\
B^{(2)}_1=&\eta_2\alpha t_2'^{\alpha} (B^{(2)})^2 .
\end{align*}
The second order derivatives are
% 2c format
%\begin{align*}
%B^{(1)}_{12}=\,&\alpha \eta_1 t_1'^{\alpha-1}(B^{(1)})^2 \\
%&
%-2\eta_1 t_1'^{\alpha}(D^{(1)}_1 - \eta_1[ (\alpha+1)t_1'^{\alpha} +\alpha t_1'^{\alpha-1}t_2]) (B^{(1)})^3 \\
%B^{(2)}_{12}=\,&\alpha \eta_2 t_2'^{\alpha-1}(B^{(2)})^2 \\
%& 
%-2\eta_2 t_2'^{\alpha}(D^{(2)}_2 - \eta_2[ (\alpha+1)t_2'^{\alpha} +\alpha t_2'^{\alpha-1}t_1]) (B^{(2)})^3 .
%\end{align*}
\begin{align*}
B^{(1)}_{12}=\,&\alpha \eta_1 t_1'^{\alpha-1}(B^{(1)})^2 
-2\eta_1 t_1'^{\alpha}(D^{(1)}_1 - \eta_1[ (\alpha+1)t_1'^{\alpha} +\alpha t_1'^{\alpha-1}t_2]) (B^{(1)})^3 \\[1.5ex]
B^{(2)}_{12}=\,&\alpha \eta_2 t_2'^{\alpha-1}(B^{(2)})^2  
-2\eta_2 t_2'^{\alpha}(D^{(2)}_2 - \eta_2[ (\alpha+1)t_2'^{\alpha} +\alpha t_2'^{\alpha-1}t_1]) (B^{(2)})^3 .
\end{align*}
The product function $C^{(n)}$ has identical first derivatives
\[
C^{(n)}_{1}=C^{(n)}_2=C^{(n)}\sum_{i=3}^{n}  \eta_i t_i'^{\alpha}B^{(i)} ,
\]
while
\[
C^{(n)}_{12}=C^{(n)}_2\sum_{i=3}^{n}  \eta_i t_i'^{\alpha}B^{(i)}+ C^{(n)}\sum_{i=3}^{n} [\eta_i t_i'^{\alpha}B^{(i)}]^2 .
\]

\subsection{Matrix form of $M'^{(n)}$}\label{matrixformMn}
Factorial moment measures  $M'^{(n)}$ admit some matrix-determinant 
representation.~\footnote{This representation is
    different from that used to define determinantal point processes, which suggests, but does not
    prove, that the STINR process is not a determinantal
    point process.} We introduce some notation.
For $n\ge 1$, 
denote by $\Delta(\bsz)$ the $n\times n$ diagonal matrix 
with entries $\bsz=(z_1,\ldots,z_n)\in\R^n$ 
on the diagonal, by $\1=(1,\ldots, 1)^T$ the
$n$-dimensional column vector of 1's ($()^T$ stands for the matrix
transpose). Denote by $\bbI=\Delta(\1)$  the $n\times n$ identity matrix. 
We will use ``$\times$'' symbol for the matrix multiplication
and $|\cdot|$ or $\det()$  for the (square) matrix determinant.

The following  identity will be used in what follows
(cf \cite[Lemma 1.1]{ding2007eigenvalues} or~\cite[Theorem
13.7.3]{harville2008matrix}).  For $n\times n$ matrix  $\bbA$, 
$n$-dimensional column vector $\bf x$ and the row vector $\bsz$ 
\begin{equation}\label{e.MDL}
|\bbA + \bfx \bsz |=  |\bbA|(1+\bsz \bbA^{-1} \bfx) \,.
\end{equation}

For given functions $\varphi,\psi$ from $\R^+\times[0,1]$ to $\R^+$ 
and two vectors $\bfz=(z_1,\ldots,z_n)\in\R^n$, $\bfeta=(\eta_1,\ldots,\eta_n)\in[0,1]$  
denote
%2c format
%\begin{align}\label{e.bbK}
%\bbK[\varphi,\psi](\bsz,\bfeta)=
%&\Delta((\varphi(z_1,\eta_1),\ldots,\varphi(z_n,\eta_n))\\
%&+\1\times ((\psi(z_1,\eta_1),\ldots,\psi(z_n,\eta_n))\,.\nonumber
%\end{align}
\begin{align}\label{e.bbK}
\bbK[\varphi,\psi](\bsz,\bfeta)=
\Delta((\varphi(z_1,\eta_1),\ldots,\varphi(z_n,\eta_n))
+\1\times ((\psi(z_1,\eta_1),\ldots,\psi(z_n,\eta_n))\,.
\end{align}
Applying the determinant identity~(\ref{e.MDL}) we obtain
\begin{equation}\label{e.detK}
\det\Bigl(\bbK[\varphi,\psi](\bsz,\bfeta)\Bigr)=
\prod_{i=1}^n\varphi(z_i,\eta_i)\Bigl(1+\sum_{i=1}^n
\frac{\varphi(z_i,\eta_i)}{\psi(z_i,\eta_i)}\Bigr)\,.
\end{equation}
Note that different rows of the matrix $\bbK[\varphi,\psi](\bsz,\bfeta)$  
depend on different entries of $\bsz$ and consequently by the Leibniz
formula for the determinant
\begin{equation}\label{e.derK}
\frac{\partial^n\det(\bbK[\varphi,\psi](\bsz,\bfeta))}%
{\partial z_1\ldots\partial z_n}
=\det(\bbK[\varphi',\psi'](\bsz,\bfeta))\,,
\end{equation}
where $\varphi'=\varphi'(z,\eta)=\frac{d}{dz}\varphi(z,\eta)$
and similarly for $\psi'$. Applying again the identity~(\ref{e.MDL}) we obtain
\begin{align}
\noalign{$\displaystyle{\frac{\partial^n\det(\bbK[\varphi,\psi](\bsz,\bfeta))}%
{\partial z_1\ldots\partial z_n}}$}
&=\prod_{i=1}^n\varphi'(z_i,\eta_i)\Bigl(1+\sum_{i=1}^n
\frac{\varphi'(z_i,\eta_i)}{\psi'(z_i,\eta_i)}\Bigr)\,.
\label{e.derK1}
\end{align}

Define the following two functions on $\R^+$ and $\R\times[0,1]$, respectively,
\begin{equation}\label{e.G_H}
 G(z)=z^{-1},\quad H(z,\eta)=z^{-2/\beta}(z+\eta)^{-1}\,.
\end{equation}

\begin{remark}
Using the determinant identity~(\ref{e.MDL}) and~(\ref{e.Q100}) it is easy to see that 
\begin{align}\label{e.Jxprod}
Q(\{\hatt_i\},\{\eta_i\})=\det\Bigl(\bbK\bigl[H,\frac HG\bigr](\hatbst,\bfeta)\Bigr).
\end{align}
The above representation may offer  an alternative way
of calculating the partial derivatives of $Q$,
following
the lines explained in
~\cite[Sections~15.8--15.9]{harville2008matrix}
and using the Sherman--Morrison formula; \cite[eq.~(2.25)]{harville2008matrix}).
\end{remark}

\subsection{Proof of Lemma \ref{l.markedinvariance}}\label{proofmarkedinvariance}
%\begin{proof} 
By the displacement theorem~\cite[Section 1.3.3]{FnT1} and Campbell's theorem~\cite[Corolloary 2.2]{FnT1}, $\Psi$ is a Poisson point process with intensity measure
\begin{align}
\Lambda(s,t)&=\E[\sum_{(Y,T)\in\markTheta} \Ind ( Y\leq s, T\leq t )] \\
&=\lambda \E \int_{\R^2}\Ind (Y\leq s) \Ind (  T\leq t ) dx \\
%&=\lambda (2\pi)\E \int_0^{\infty}\Ind (\ell_i(r)/\tildeS_i \leq s) \Ind ( T_i\leq t ) rdr \\
&=\lambda (2\pi)\E \int_0^{\infty}\Ind (r \leq (sPS)^{1/\beta}) \Ind ( T\leq t ) rdr.
%&=\lambda \pi \E \left[(s\tildeS_i)^{2/\beta_i} \Ind ( T_i\leq t ) \right] \label{e.LambdaEnd}.
\end{align}
%\end{proof} 

\subsection{Schuette-Nesbitt formula}\label{App1}
The following result is often used in insurance mathematics
(see, for example, \cite{GERBER:1995} for a  proof). 
Let $B_1,B_2,\dots B_m$ denote arbitrary events, and define the $n$\,th
symmetric sum for them as 
$\mathcal{S}_n = \sum\limits_{{\,i_1,i_2,\dots,i_n\atop \text{distinct}}}
\Prob\Bigl(\,\bigcap_{j=1}^n B_{i_j}\,\Bigr)$.
Let the random variable $\mathcal{N}$ denote the number of the aforementioned events that occur, that is
$\mathcal{N}=\sum\limits_{i=1}^{m}\Ind_{B_i}$.

\begin{theorem}[Schuette-Nesbitt formula]
For arbitrary coefficients $c_1,c_2,\dots c_m$ the following holds true
\[
 \sum\limits_{n=0}^{m} c_n \Prob(\,\mathcal{N}=n\,)
=\sum\limits_{n=0}^{m} (\Delta^nc)_0 \calS_n,
\]
where $\Delta$ is the forward difference operator; that is, $(\Delta
c)_k=c_{k+1}-c_{k}$ and  $\Delta^k=\Delta\Delta^{k-1}$.
 \end{theorem}

\vspace{-2ex}
%\addtocounter{section}{1}
%\addcontentsline{toc}{section}{References} 
%\pdfbookmark[0]{References}{References} 
\bibliographystyle{IEEEtran}
\bibliography{coveragemultitier}

% Generated by IEEEtran.bst, version: 1.13 (2008/09/30)
\begin{thebibliography}{10}
\providecommand{\url}[1]{#1}
\csname url@samestyle\endcsname
\providecommand{\newblock}{\relax}
\providecommand{\bibinfo}[2]{#2}
\providecommand{\BIBentrySTDinterwordspacing}{\spaceskip=0pt\relax}
\providecommand{\BIBentryALTinterwordstretchfactor}{4}
\providecommand{\BIBentryALTinterwordspacing}{\spaceskip=\fontdimen2\font plus
\BIBentryALTinterwordstretchfactor\fontdimen3\font minus
  \fontdimen4\font\relax}
\providecommand{\BIBforeignlanguage}[2]{{%
\expandafter\ifx\csname l@#1\endcsname\relax
\typeout{** WARNING: IEEEtran.bst: No hyphenation pattern has been}%
\typeout{** loaded for the language `#1'. Using the pattern for}%
\typeout{** the default language instead.}%
\else
\language=\csname l@#1\endcsname
\fi
#2}}
\providecommand{\BIBdecl}{\relax}
\BIBdecl

\bibitem{ANDREWS2011}
J.~Andrews, F.~Baccelli, and R.~Ganti, ``A tractable approach to coverage and
  rate in cellular networks,'' \emph{IEEE Trans. Commun.}, vol.~59, no.~11, pp.
  3122 --3134, november 2011.

\bibitem{hextopoi}
B.~B{\l}aszczyszyn, M.~K. Karray, and H.~P. Keeler, ``Using {P}oisson processes
  to model lattice cellular networks,'' in \emph{INFOCOM, 2013 Proceedings
  IEEE}, 2013, pp. 773--781.

\bibitem{hextopoi-journal}
------, ``Wireless networks appear {P}oissonian due to strong shadowing,''
  arxiv:1409.4739, 2014.

\bibitem{brown2000cellular}
T.~X. Brown, ``Cellular performance bounds via shotgun cellular systems,''
  \emph{Selected Areas in Communications, IEEE Journal on}, vol.~18, no.~11,
  pp. 2443--2455, 2000.

\bibitem{fme}
B.~B{\l}aszczyszyn, ``Factorial-moment expansion for stochastic systems,''
  \emph{Stoch. Proc. Appl.}, vol.~56, pp. 321--335, 1995.

\bibitem{fme_spatial}
B.~B{\l}aszczyszyn, E.~Merzbach, and V.~Schmidt, ``A note on expansions for
  functional of spatial marked processes,'' \emph{Statist. and Probab. Lett.},
  vol.~36, pp. 299--306, 1997.

\bibitem{equivalence2013}
B.~B{\l}aszczyszyn and H.~Keeler, ``Equivalence and comparison of heterogeneous
  cellular networks,'' in \emph{Proc. of WDN-CN2013}, 2013.

\bibitem{GILBERT1960}
E.~N. Gilbert and H.~O. Pollak, ``{Amplitude distribution of shot noise},''
  \emph{Bell Systems Technical Journal}, vol.~39, pp. 333--350, 1960.

\bibitem{LOWEN1990}
S.~Lowen and M.~C. Teich, ``Power-law shot noise,'' \emph{Information Theory,
  IEEE Transactions on}, vol.~36, no.~6, pp. 1302--1318, 1990.

\bibitem{HAENGGI2008}
M.~Haenggi, ``A geometric interpretation of fading in wireless networks: Theory
  and applications,'' \emph{Information Theory, IEEE Transactions on}, vol.~54,
  no.~12, pp. 5500--5510, 2008.

\bibitem{blaszczyszyn2010impact}
B.~B{\l}aszczyszyn, M.~Karray, F.~Klepper \emph{et~al.}, ``Impact of the
  geometry, path-loss exponent and random shadowing on the mean interference
  factor in wireless cellular networks,'' in \emph{Third Joint IFIP Wireless
  and Mobile Networking Conference (WMNC)}, 2010.

\bibitem{PINTO2012}
P.~Pinto, J.~Barros, and M.~Win, ``Secure communication in stochastic wireless
  networks -- {P}art {I}: Connectivity,'' \emph{Information Forensics and
  Security, IEEE Transactions on}, vol.~7, no.~1, pp. 125--138, 2012.

\bibitem{dhillon2011tractable}
H.~S. Dhillon, R.~K. Ganti, and J.~G. Andrews, ``A tractable framework for
  coverage and outage in heterogeneous cellular networks,'' in
  \emph{Information Theory and Applications Workshop (ITA), 2011}.\hskip 1em
  plus 0.5em minus 0.4em\relax IEEE, 2011, pp. 1--6.

\bibitem{DHILLON2012}
H.~Dhillon, R.~Ganti, F.~Baccelli, and J.~Andrews, ``Modeling and analysis of
  {K}-tier downlink heterogeneous cellular networks,'' \emph{IEEE J. Sel. Areas
  Commun.}, vol.~30, no.~3, pp. 550--560, april 2012.

\bibitem{MADBROWN2011}
P.~Madhusudhanan, J.~Restrepo, Y.~Liu, T.~Brown, and K.~Baker, ``Multi-tier
  network performance analysis using a shotgun cellular system,'' in
  \emph{Global Telecommunications Conference (GLOBECOM 2011), 2011 IEEE}, 2011,
  pp. 1--6.

\bibitem{MADBROWN2012}
P.~Madhusudhanan, J.~Restrepo, Y.~Liu, and T.~Brown, ``Downlink coverage
  analysis in a heterogeneous cellular network,'' in \emph{Global
  Communications Conference (GLOBECOM), 2012 IEEE}, 2012, pp. 4170--4175.

\bibitem{mukherjee2011Allerton}
S.~Mukherjee, ``Downlink {SINR} distribution in a heterogeneous cellular
  wireless network with max-{SINR} connectivity,'' in \emph{Communication,
  Control, and Computing (Allerton), 2011 49th Annual Allerton Conference
  on}.\hskip 1em plus 0.5em minus 0.4em\relax IEEE, 2011, pp. 1649--1656.

\bibitem{kcovsingle}
H.~Keeler, B.~Blaszczyszyn, and M.~Karray, ``{SINR}-based $k$-coverage
  probability in cellular networks with arbitrary shadowing,'' in
  \emph{Information Theory Proceedings (ISIT), 2013 IEEE International
  Symposium on}, 2013, pp. 1167--1171.

\bibitem{MUKHERJEE2012}
S.~Mukherjee, ``Distribution of downlink {SINR} in heterogeneous cellular
  networks,'' \emph{Selected Areas in Communications, IEEE Journal on},
  vol.~30, no.~3, pp. 575--585, 2012.

\bibitem{JO2012}
H.-S. Jo, Y.~J. Sang, P.~Xia, and J.~Andrews, ``Heterogeneous cellular networks
  with flexible cell association: A comprehensive downlink {SINR} analysis,''
  \emph{Wireless Communications, IEEE Transactions on}, vol.~11, no.~10, pp.
  3484--3495, 2012.

\bibitem{VU2012}
T.~T. Vu, L.~Decreusefond, and P.~Martins, ``An analytical model for evaluating
  outage and handover probability of cellular wireless networks,'' in
  \emph{Proc of WPMC}, 2012, pp. 643--647.

\bibitem{mukherjee2012ICC}
S.~Mukherjee, ``Downlink {SINR} distribution in a heterogeneous cellular
  wireless network with biased cell association,'' in \emph{Communications
  (ICC), 2012 IEEE International Conference on}.\hskip 1em plus 0.5em minus
  0.4em\relax IEEE, 2012, pp. 6780--6786.

\bibitem{giovanidis2013stochastic}
A.~Giovanidis and F.~Baccelli, ``A stochastic geometry framework for analyzing
  pairwise-cooperative cellular networks,'' \emph{arXiv preprint
  arXiv:1305.6254}, 2013.

\bibitem{akoum2010limited}
S.~Akoum and R.~W. Heath, ``Limited feedback for temporally correlated {MIMO}
  channels with other cell interference,'' \emph{Signal Processing, IEEE
  Transactions on}, vol.~58, no.~10, pp. 5219--5232, 2010.

\bibitem{zhang2012performance}
X.~Zhang and M.~Haenggi, ``The performance of successive interference
  cancellation in random wireless networks,'' in \emph{Global Communications
  Conference (GLOBECOM), 2012 IEEE}.\hskip 1em plus 0.5em minus 0.4em\relax
  IEEE, 2012, pp. 2317--2321.

\bibitem{zhangdecoding}
------, ``{On decoding the $k\,$th strongest user in {P}oisson networks with
  arbitrary fading distribution},'' in \emph{47th Asilomar Conference of
  Signals, Systems and Computers (Asilomar'13)}, Pacific Grove, CA, Nov. 2013.

\bibitem{zhangsuccessive}
------, ``{Successive Interference Cancellation in Downlink Heterogeneous
  Cellular Networks},'' in \emph{IEEE Global Communications Conference
  (GLOBECOM-HetSNets'13)}, Atlanta, GA, 2013.

\bibitem{wildemeersch2013successive}
M.~Wildemeersch, T.~Q. Quek, M.~Kountouris, A.~Rabbachin, and C.~H. Slump,
  ``Successive interference cancellation in heterogeneous cellular networks,''
  \emph{arXiv preprint arXiv:1309.6788}, 2013.

\bibitem{daleyPPI2003}
D.~J. Daley and D.~Vere-Jones, \emph{An introduction to the theory of point
  processes. {V}ol. {I}}, 2nd~ed., ser. Probability and its Applications (New
  York).\hskip 1em plus 0.5em minus 0.4em\relax New York: Springer, 2003.

\bibitem{daleyPPII2008}
------, \emph{An introduction to the theory of point processes. {V}ol. {II}},
  2nd~ed., ser. Probability and its Applications (New York).\hskip 1em plus
  0.5em minus 0.4em\relax New York: Springer, 2008.

\bibitem{SKM:1995}
D.~Stoyan, W.~Kendall, and J.~Mecke, \emph{Stochastic Geometry and its
  Applications}, 2nd~ed.\hskip 1em plus 0.5em minus 0.4em\relax Wiley, 1995.

\bibitem{fme_marked}
D.~P. Kroese and V.~Schmidt, ``\BIBforeignlanguage{English}{Light-traffic
  analysis for queues with spatially distributed arrivals},''
  \emph{\BIBforeignlanguage{English}{Mathematics of Operations Research}},
  vol.~21, no.~1, pp. pp. 135--157, 1996.

\bibitem{fmewireless}
R.~Ganti, F.~Baccelli, and J.~Andrews, ``Series expansion for interference in
  wireless networks,'' \emph{Information Theory, IEEE Transactions on},
  vol.~58, no.~4, pp. 2194--2205, 2012.

\bibitem{sinrPD}
H.~P. Keeler and B.~B{\l}aszczyszyn, ``{SINR} in wireless networks and the
  two-parameter {P}oisson-{D}irichlet process,'' \emph{IEEE Wireless Comm
  Letters}, vol.~3, no.~5, 2014.

\bibitem{pitman1997two}
J.~Pitman and M.~Yor, ``The two-parameter {Poisson--Dirichlet} distribution
  derived from a stable subordinator,'' \emph{The Annals of Probability},
  vol.~25, no.~2, pp. 855--900, 1997.

\bibitem{handa2009two}
K.~Handa, ``The two-parameter {Poisson--Dirichlet} point process,''
  \emph{Bernoulli}, vol.~15, no.~4, pp. 1082--1116, 2009.

\bibitem{gunnarsson2004power}
F.~Gunnarsson, ``Power control in wireless networks: characteristics and
  fundamentals,'' in \emph{Wireless communications systems and networks},
  M.~Guizani, Ed., 2004, pp. 179--208.

\bibitem{panchenko2013sherrington}
D.~Panchenko, \emph{The {S}herrington-{K}irkpatrick model}.\hskip 1em plus
  0.5em minus 0.4em\relax Springer, 2013.

\bibitem{KINGMAN:1993}
J.~F.~C. Kingman, \emph{Poisson Processes}, 1st~ed.\hskip 1em plus 0.5em minus
  0.4em\relax Oxford University Press, 1993.

\bibitem{FnT1}
F.~Baccelli and B.~B{\l}aszczyszyn, \emph{Stochastic Geometry and Wireless
  Networks, Volume I --- Theory}, ser. Foundations and Trends in
  Networking.\hskip 1em plus 0.5em minus 0.4em\relax NoW Publishers, 2009, vol.
  3, No 3--4.

\bibitem{Fel68}
W.~Feller, \emph{An Introduction to Probability Theory and its Applications,
  vol I.}, 3rd~ed.\hskip 1em plus 0.5em minus 0.4em\relax New York: J. Wiley \&
  Sons, 1968.

\bibitem{GERBER:1995}
H.~U. Gerber, \emph{Life Insurance Mathematics}.\hskip 1em plus 0.5em minus
  0.4em\relax Springer, 1995.

\bibitem{weber2007transmission}
S.~P. Weber, J.~G. Andrews, X.~Yang, and G.~De~Veciana, ``Transmission capacity
  of wireless ad-hoc networks with successive interference cancellation,''
  \emph{Information Theory, IEEE Transactions on}, vol.~53, no.~8, pp.
  2799--2814, 2007.

\bibitem{bar1989adaptive}
Y.~Bar-Ness and H.~Bunin, ``Adaptive co-channel interference cancellation and
  signal separation method,'' in \emph{Communications, 1989. ICC'89,
  BOSTONICC/89. Conference record.'World Prosperity Through Communications',
  IEEE International Conference on}.\hskip 1em plus 0.5em minus 0.4em\relax
  IEEE, 1989, pp. 825--830.

\bibitem{xu2012co}
W.~Xu and S.~Sezginer, ``Co-channel interference cancellation in reuse-1
  deployments of wimax system,'' in \emph{Wireless Communications and
  Networking Conference (WCNC), 2012 IEEE}.\hskip 1em plus 0.5em minus
  0.4em\relax IEEE, 2012, pp. 342--346.

\bibitem{paul_matlab_moments}
H.~P. Keeler, ``Studying the {SINR} process in {P}oisson networks by using its
  factorial moment measures,'' {MATLAB} Central File Exchange, 2014, available
  at
  \url{http://www.mathworks.com.au/matlabcentral/fileexchange/45299-studying-the-sinr-process-in-poisson-networks-by-using-its-factorial-moment-measures}.

\bibitem{kuo2005lifting}
F.~Kuo and I.~Sloan, ``Lifting the curse of dimensionality,'' \emph{Notices of
  the AMS}, vol.~52, no.~11, pp. 1320--1328, 2005.

\bibitem{dick2013high}
J.~Dick, F.~Y. Kuo, and I.~H. Sloan, ``High-dimensional integration--the
  {Quasi-Monte Carlo} way,'' \emph{Acta Numerica}, vol.~22, pp. 133--288, 2013.

\bibitem{STUBER2011}
G.~L. St$\ddot{\mathrm{u}}$ber, \emph{Principles of Mobile Communication},
  2nd~ed.\hskip 1em plus 0.5em minus 0.4em\relax NYC, NY, USA: Springer, 2011.

\bibitem{Miyoshi2014}
N.~Miyoshi and T.~Shirai, ``A cellular network model with {G}inibre
  configurated base stations,'' \emph{Adv. Appl. Probab.}, vol.~46, no.~3, pp.
  832--845, 2014.

\bibitem{nakata2014spatial}
I.~Nakata and N.~Miyoshi, ``Spatial stochastic models for analysis of
  heterogeneous cellular networks with repulsively deployed base stations,''
  \emph{Performance Evaluation}, vol.~78, pp. 7--17, 2014.

\bibitem{DLMF}
\BIBentryALTinterwordspacing
(2012, {Accessed on the 10th of September}) {Digital Library of Mathematical
  Functions}. National Institute of Standards and Technology. Release 1.0.5 of
  2012-10-01. [Online]. Available: \url{http://dlmf.nist.gov/}
\BIBentrySTDinterwordspacing

\bibitem{ding2007eigenvalues}
J.~Ding and A.~Zhou, ``Eigenvalues of rank-one updated matrices with some
  applications,'' \emph{Applied Mathematics Letters}, vol.~20, no.~12, pp.
  1223--1226, 2007.

\bibitem{harville2008matrix}
D.~A. Harville, \emph{Matrix algebra from a statistician's perspective}.\hskip
  1em plus 0.5em minus 0.4em\relax Springer, 2008.

\end{thebibliography}
\end{document}